\documentclass[runningheads]{llncs}
\pdfoutput=1

\usepackage{amsmath}
\usepackage{amssymb}
\usepackage{mathtools}
\usepackage{marvosym}
\usepackage[heavycircles]{stmaryrd}
\usepackage{mathpartir}
\usepackage{listings}
\usepackage{multicol}
\usepackage{proof-dashed}
\usepackage{varwidth}
\usepackage{xparse} 
\usepackage{color}
\usepackage[caption=false]{subfig}
\usepackage{microtype}
\definecolor{darkgreen}{RGB}{0, 192, 0}
\usepackage[breaklinks=true,
colorlinks=true,
urlcolor=magenta,
citecolor=darkgreen]{hyperref}
\usepackage{breakcites}

\newcommand*{\jrule}[1]{\text{\textsc{\MakeLowercase{#1}}}}
\newcommand*{\irule}[1]{#1\jrule{I}}
\newcommand*{\erule}[1]{#1\jrule{E}}

\newcommand*{\ms}[1]{\mathsf{#1}}
\newcommand*{\mb}[1]{\mathbf{#1}}

\newcommand*{\ctx}{\Gamma}
\NewDocumentCommand{\ctxe}{s}{\IfBooleanT{#1}{(} \cdot \IfBooleanT{#1}{)}}

\newcommand*{\tensor}{\mathbin{\otimes}}
\newcommand*{\one}{\mathord{\mathbf{1}}}
\NewDocumentCommand{\plus}{s}{\IfBooleanTF{#1}{\plusn}{\mathbin{\oplus}}}
\NewDocumentCommand{\plusn}{o m}{\mathopen{\oplus}\lbrace #2 \rbrace\IfValueT{#1}{\sb{#1}}}
\newcommand*{\dn}{\mathopen{\downarrow}}
\newcommand*{\imp}{\mathbin{\rightarrow}}
\NewDocumentCommand{\with}{s}{\IfBooleanTF{#1}{\withn}{\mathbin{\binampersand}}}
\NewDocumentCommand{\withn}{o m}{\mathopen{\binampersand}\lbrace #2 \rbrace\IfValueT{#1}{\sb{#1}}}
\newcommand*{\up}{\mathopen{\uparrow}}

\newcommand*{\p}[1]{#1^{+}}
\newcommand*{\n}[1]{#1^{-}}

\newcommand*{\pair}[1]{\langle#1\rangle}
\newcommand*{\inj}[2]{#1\cdot #2}
\newcommand*{\thunk}[1]{\dn #1}

\newcommand*{\lam}[2]{\lambda #1.\,#2}
\newcommand*{\app}[2]{#1\,#2}
\NewDocumentCommand{\record}{o m}{\lbrace #2 \rbrace\IfValueT{#1}{\sb{#1}}}
\newcommand*{\proj}[2]{#2.#1}
\newcommand*{\return}[1]{\up #1}
\newcommand*{\letup}[3]{\mathsf{let}\;\up #1=#2\;\mathsf{in}\;#3}
\newcommand*{\letpair}[3]{\mathsf{let}\:\pair{#1}=#2\;\mathsf{in}\;#3}
\NewDocumentCommand{\force}{s m}{\mathsf{force}\;\IfBooleanT{#1}{(} #2 \IfBooleanT{#1}{)}}
\RenewDocumentCommand{\case}{o m m}{\mathsf{case}\;#2\;(\casepat#3\relax)\IfValueT{#1}{\sb{#1}}}
\def\casepat#1#2=>#3\relax{\inj{#1}{#2} \Rightarrow #3}

\newcommand*{\metaall}[1]{\forall (#1)\colon}

\newcommand*{\subst}[2]{[#1]#2}

\newcommand*{\term}{\;\mathsf{term}}
\newcommand*{\val}{\;\mathsf{value}}
\newcommand*{\reduces}{\mapsto}

\newcommand*{\mimp}{\mathbin{\supset}}

\newcommand*{\sof}[2]{#1 \in #2}

\newcommand*{\defd}{\triangleq}

\newcommand*{\emp}{\;\mathsf{empty}}
\newcommand*{\full}{\;\mathsf{full}}

\newcommand*{\rsub}[1]{\mathord{#1}\jrule{SUB}}

\newcommand*{\remp}[1]{\mathord{#1}\jrule{Emp}}
\newcommand*{\rfull}[1]{\mathord{#1}\jrule{Full}}

\newcommand*{\sig}{\Sigma}
\NewDocumentCommand{\sige}{s}{\IfBooleanT{#1}{(}\cdot\IfBooleanT{#1}{)}}

\ExplSyntaxOn
\cs_new:Npn \__cbtranslation:Nn #1#2 {
  \NewDocumentCommand{#1}{s m}{
    \IfBooleanT{##1}{(}
      \tl_if_empty:nTF{##2}{\mathord{-}}{##2}
    \IfBooleanT{##1}{)}^{#2}
  }
}
\__cbtranslation:Nn \cbv { \boxplus }
\__cbtranslation:Nn \cbn { \boxminus }
\__cbtranslation:Nn \cbvp { \boxplus }
\__cbtranslation:Nn \cbvn { \boxminus }
\ExplSyntaxOff

\renewcommand*{\return}[1]{\mathsf{return}\; #1}
\renewcommand*{\thunk}[1]{\mathsf{thunk}\; #1}
\renewcommand*{\letup}[3]{\mathsf{let}\;\mathsf{return}\;#1=#2\;\mathsf{in}\;#3}
\renewcommand*{\letpair}[3]{\mathsf{match}\;#2\;(\pair{#1} \Rightarrow #3)}
\RenewDocumentCommand{\case}{o m m}{\mathsf{match}\;#2\;(\casepat#3\relax)\IfValueT{#1}{\sb{#1}}}
\def\casepat#1#2=>#3\relax{\inj{#1}{#2} \Rightarrow #3}
\RenewDocumentCommand{\sof}{s o m m}{#3 \IfBooleanTF{#1}{\inhat}{\in}\IfValueT{#2}{\sb{#2}} #4}
\NewDocumentCommand{\nsof}{s o m m}{#3 \IfBooleanTF{#1}{\not\inhat}{\not\in}\IfValueT{#2}{\sb{#2}} #4}
\renewcommand*{\term}{\;\mathsf{terminal}}
\newcommand{\inhat}{\mathrel{\hat\in}}

\newcommand*{\ntype}{\;\n{\mathsf{type}}}
\newcommand*{\ptype}{\;\p{\mathsf{type}}}
\newcommand*{\vsig}{\;\mathsf{sig}}
\newcommand*{\pctx}{\;\mathsf{ctx}}
\newcommand*{\structural}{\;\mathsf{struct}}

\newcommand*{\cycle}[1]{\jrule{Cycle}{\scriptstyle (#1)}}

\newcommand*{\foldcase}[3]{\mathsf{match}\;#2\;(\mathsf{fold_{\mu}}\;{#1} \Rightarrow #3)}
\newcommand*{\unfoldfold}[1]{\mathsf{\unfold{(\mathsf{fold_{\nu}}\;{#1})}}}

\newcommand*{\foldmu}{\mathsf{fold_{\mu}}}
\newcommand*{\foldnu}{\mathsf{fold_{\nu}}}
\newcommand*{\unfold}{\mathsf{unfold}}

\renewcommand*{\unfold}[1]{\mathsf{unfold}\; #1}
\renewcommand*{\foldmu}[1]{\mathsf{fold_{\mu}}\; #1}
\renewcommand*{\foldnu}[1]{\mathsf{fold_{\nu}}\; #1}

\definecolor{darkred}{rgb}{0.55, 0.0, 0.0}
\definecolor{darkblue}{rgb}{0.0, 0.0, 0.55}

\newcommand*{\synths}{\;{\color{darkred}\Rightarrow}\;}
\newcommand*{\checks}{\;{\color{darkblue}\Leftarrow}\;}

\begin{document}

\title{Polarized Subtyping}

\author{Zeeshan Lakhani\inst{1}(\Letter) \and
 Ankush Das\inst{3} \and
 Henry DeYoung\inst{1} \and
 Andreia Mordido\inst{2} \and
 Frank Pfenning\inst{1}}

\authorrunning{Z. Lakhani et al.}
%
\institute{Carnegie Mellon University, Pittsburgh, PA, USA \\
\email{\{zlakhani,hdeyoung,fp\}@cs.cmu.edu}
\and
LASIGE, Faculdade de Ci\^encias, Universidade de Lisboa, Lisbon, Portugal
\email{afmordido@fc.ul.pt}\\
\and
Amazon, Cupertino, CA, USA\footnote{work performed prior to joining Amazon}
\email{daankus@amazon.com}}

\maketitle

\begin{abstract}
  Polarization of types in call-by-push-value naturally leads to the separation of
  inductively defined observable values (classified by positive types), and coinductively
  defined computations (classified by negative types), with adjoint modalities mediating
  between them.  Taking this separation as a starting point, we develop a semantic
  characterization of typing with step indexing to capture observation depth of recursive
  computations.  This semantics justifies a rich set of subtyping rules for an
  equirecursive variant of call-by-push-value, including variant and lazy records.  We
  further present a bidirectional syntactic typing system for both values and computations
  that elegantly and pragmatically circumvents difficulties of type inference in the
  presence of width and depth subtyping for variant and lazy records.  We demonstrate the
  flexibility of our system by systematically deriving related systems of subtyping for
  (a) isorecursive types, (b) call-by-name, and (c) call-by-value, all using a structural
  rather than a nominal interpretation of types.

\end{abstract}

\section{Introduction}
\label{sec:intro}

Subtyping is an important concept in programming languages because it simultaneously
allows more programs to be typed and more precise properties of programs to be expressed
as types.  The interaction of subtyping with parametric polymorphism
and recursive types is complex and despite a lot of progress and research, not
yet fully understood.

In this paper we study the interaction of subtyping with equirecursive types in
\emph{call-by-push-value}~\cite{Levy01phd,Levy06hosc}, which separates the language of
types into \emph{positive} and \emph{negative} layers.  This polarization elegantly
captures that positive types classifying observable values are \emph{inductive}, while
negative types classifying (possibly recursive) computations are \emph{coinductive}.  It
lends itself to a particularly simple \emph{semantic definition of typing} using a mixed
induction/coinduction~\cite{Amadio93toplas,Brandt98fi,Danielsson10mpc}.  From this
definition, we can immediately derive a form of \emph{semantic
subtyping}~\cite{Castagna05ppdp, Frisch02lics, Frisch08acm}.  Concretely, we realize the mixed
induction/coinduction via step-indexing and carry out our metatheory in Brotherston and
Simpson's system \textbf{CLKID$^\omega$} of circular proofs~\cite{Brotherston11jlc}.  This
includes a novel proof that syntactic versions of typing and subtyping are sound with
respect to our semantic definitions.  While we also conjecture that subtyping is precise
(in the sense of~\cite{Ligatti17toplas}), we postpone this more syntactic property to
future work.

Because our foundation is call-by-push-value, a paradigm that synthesizes call-by-name and
call-by-value based on the logical principle of polarization, we obtain several additional
results in relatively straightforward ways.  For example, both width and depth subtyping
for variant and lazy records are naturally included.  Furthermore, following Levy's
interpretation of call-by-value and call-by-name functional languages into
call-by-push-value, we extract subtyping relations and algorithms for these languages and
prove them sound and complete.  We also note that we can directly interpret the
\emph{isorecursive} types in Levy's original formulation of
call-by-push-value~\cite{Levy01phd}.

We further provide a systematic notion of bidirectional typing that avoids some
complexities that arise in a structural type system with variant and lazy records. The
resulting decision procedure for typing is quite precise and suggests clear locations for
noting failure of typechecking.  The combination of equirecursive call-by-push-value with
bidirectional typing achieves some of the goals of refinement
types~\cite{Davies05phd,Freeman91}, which fit a structural system inside a generative type
language.  Here we have considerably more freedom and less redundancy.  However, we do not
yet treat intersection types or polymorphism.

We summarize our main contributions:
\begin{enumerate}
\item A simple semantics for types and subtyping in call-by-push-value, interpreting
  positive types inductively and negative types coinductively, realized via step indexing
  (Sections~\ref{sec:semantic-typing} and~\ref{sec:subtyping})
\item A new decidable system of equirecursive subtyping for call-by-push-value including
  width and depth subtyping for variant and lazy records (Section~\ref{sec:subtyping})
\item A novel application of Brotherston and Simpson's system
  \textbf{CLKID$^\omega$}~\cite{Brotherston11jlc} of circular proofs to give a
  particularly elegant and flexible soundness proof for subtyping
  (Section~\ref{sec:type-sound})
\item A system of bidirectional typing that captures a straightforward and precise typechecking
  algorithm (Section~\ref{sec:bidirectional})
\item A simple interpretation of Levy's original isorecursive types for
  call-by-push-value~\cite{Levy01phd} into our equirecursive setting
  (Section~\ref{sec:isorecursive})
\item Subtyping rules for call-by-name and call-by-value, derived via Levy's translations
  of such languages into call-by-push-value (Section~\ref{sec:call-by-name-and-value})
\end{enumerate}
These are followed by a discussion of related work and a conclusion.
Additional material and proofs are provided in an \hyperref[app:type-examples]{appendix}.

\section{Equirecursive Call-by-Push-Value}
\label{sec:syntax}

Call-by-push-value~\cite{Levy01phd,Levy06hosc} is characterized by a
separation of types in \emph{positive} $\p{\tau}$ and \emph{negative}
$\n{\sigma}$ layers, with shift modalities going back and forth
between them.  The intuition is that positive types classify
\emph{observable values} $v$ while negative types classify
\emph{computations} $e$.
\begin{align*}
 \p{\tau}, \p{\sigma} &\Coloneqq \p{\tau_1} \tensor \p{\tau_2} \mid \one \mid \plus*[\ell \in L]{\ell\colon \p{\tau_\ell}} \mid \dn \n{\sigma} \mid \p{t} \\
  \n{\sigma}, \n{\tau} &\Coloneqq \p{\tau} \imp \n{\sigma} \mid \with*[\ell \in L]{\ell\colon \n{\sigma_\ell}} \mid \up \p{\tau} \mid \n{s}
\end{align*}
The usual binary product $\tau \times \sigma$ splits into two:
$\p{\tau} \tensor \p{\sigma}$ for eager, observable products
inhabited by pairs of values, and
$\with*[\ell \in L]{\ell \colon \n{\sigma_\ell}}$ for
lazy, unobservable records with a finite set $L$ of fields we can
project out.
Binary sums are also generalized to
variant record types $\plus*[\ell \in L]{\ell\colon \p{\tau_\ell}}$.\footnote{We borrow the notation $\oplus$ from linear
  logic even though no linearity is implied.}  These
are not just a programming
convenience but allow for richer subtyping: lazy and variant record types support both width and depth subtyping, whereas the usual binary products and sums support only the latter.
For example, width subtyping means that $\plus*{\mb{false}\colon \one}$ is a subtype of
$\p{\ms{bool}} = \plus*{\mb{false}\colon \one, \mb{true}\colon \one}$,
while $\one$ would not be a subtype of the usual binary $\one + \one$.
Neither is $\one$ a subtype of $\p{\ms{bool}}$, demonstrating the utility of variant record types with one label, such as $\plus*{\mb{false}\colon \one}$.
Similar examples exist for lazy record types. This way, we recover some of the benefits of
refinement types without the syntactic burden of a distinct refinement layer.

The shift $\dn \n{\sigma}$ is inhabited by an unevaluated computation of type $\n{\sigma}$
(a ``thunk'').  Conversely, the shift $\up \p{\tau}$ includes a value as a trivial
computation (a ``return'').  Levy~\cite{Levy01phd} writes $U\, \underline{B}$ instead of
$\dn \n{\sigma}$ and $F\, A$ instead of $\up \p{\tau}$.

Finally, we model recursive types not by explicit constructors
$\mu \p{\alpha}.\, \p{\tau}$ and $\nu \n{\alpha}.\, \n{\sigma}$ but by \emph{type names}
$\p{t}$ and $\n{s}$ which are defined in a global signature $\Sigma$.  They may mutually
refer to each other.  We treat these as \emph{equirecursive} (see
Section~\ref{sec:semantic-typing}) and we require them to be \emph{contractive}, which
means the right-hand side of a type definition cannot itself be a type name.
Since we would like to directly observe the values of positive types, the definitions of
type names $\p{t} = \p{\tau}$ are \emph{inductive}.  This allows inductive reasoning about
values returned by computations.  On the other hand, negative type definitions
$\n{s} = \n{\sigma}$ are recursive rather than coinductive in the usual sense, which would
require, for example, stream computations to be productive.  Because we do not wish to
restrict recursive computations to those that are productive in this sense, they are
``productive'' only in the sense that they satisfy a standard progress theorem.

Next, we come to the syntax for values $v$ of a positive type and computations $e$ of
a negative type.  Variables $x$ always stand for values and therefore have a positive type.
We use $j$ to stand for labels, naming fields of variant records or lazy records, where
$\inj{j}{v}$ injects value $v$ into a sum with alternative labeled $j$ and $\proj{j}{e}$
projects field $e$ out of a lazy record.  When we quantify over a (always finite) set of
labels we usually write $\ell$ as a metavariable for the labels.
\begin{align*}
  v &\Coloneqq x \mid \pair{v_1,v_2} \mid \pair{} \mid \inj{j}{v} \mid \thunk{e} \\
  e &\Coloneqq \begin{array}[t]{@{}l@{}}
                 \lam{x}{e} \mid \app{e}{v} \mid \record[\ell \in L]{\ell = e_\ell} \mid \proj{j}{e} \mid \return{v} \mid \letup{x}{e_1}{e_2} \mid f \\
                 \mathllap{\mid {}} \letpair{x,y}{v}{e} \mid \letpair{}{v}{e} \mid \case[\ell \in L]{v}{{\ell}{x_\ell} => e_\ell} \mid \force{v}
               \end{array} \\
  \Sigma &\Coloneqq \cdot \mid \Sigma, t^+ = \p{\tau} \mid \Sigma, s^- = \n{\sigma} \mid \Sigma, f \colon
           \n{\sigma} = e
\end{align*}
In order to represent recursion, we use equations $f = e$ in the
signature where $f$ is a defined \emph{expression name},which we distinguish from variables, and all
equations can mutually reference each other.  An alternative would
have been explicit fixed point expressions $\mathsf{fix}\, f.\,e$, but
this mildly complicates both typing and mutual recursion.  Also, it
seems more elegant to represent all forms of recursion at the level of
types and expressions in the same manner.  We also choose to fix a type
for each expression name in a signature.  Otherwise, each occurrence
of $f$ in an expression could potentially be assigned a different type,
which strays into the domain of parametric polymorphism and intersection types.

Following Levy, we do not allow names for values because this
would add an undesirable notion of computation to values, and,
furthermore, circular values would violate the inductive
interpretation of positive types.  As discussed in~\cite[Chapter 4]{Levy01phd},
they could be added back conservatively under some conditions.

\subsection{Dynamics}

For the operational semantics, we use a judgment $e \reduces e'$
defined inductively by the following rules which may reference a
global signature $\Sigma$ to look up the definitions of expression
names $f$.  In contrast, values do not reduce.
The dynamics of call-by-push-value are defined as follows:
\begin{mathpar}
  \infer{\app{(\lam{x}{e})}{v} \reduces \subst{v/x}{e}}{}
  \and
  \infer{\app{e}{v} \reduces \app{e'}{v}}
  {e \reduces e'}
  \and
  \infer{\letup{x}{\return{v}}{e_2} \reduces \subst{v/x}{e_2}}{}
  \\
  \infer{\letup{x}{e_1}{e_2} \reduces \letup{x}{e'_1}{e_2}}
  {e_1 \reduces e'_1}
  \quad
  \infer{\proj{j}{\record[\ell \in L]{\ell = e_\ell}} \reduces e_j}
  {\text{($j \in L$)}}
  \quad
  \infer{\proj{j}{e} \reduces \proj{j}{e'}}
  {e \reduces e'}
  \\
  \infer{\letpair{x,y}{\pair{v_1,v_2}}{e} \reduces \subst{v_1/x}{\subst{v_2/y}{e}}}{}
  \and
  \infer{\letpair{}{\pair{}}{e} \reduces e}{}
  \and
  \infer{\case[\ell \in L]{(\inj{j}{v})}{{\ell}{x_\ell} => e_\ell} \reduces \subst{v/x_j}{e_j}}
  {\text{($j \in L$)}}
  \quad
  \infer{\force*{\thunk{e}} \reduces e}{}
  \quad
  \infer{f \reduces e}{f \colon \n{\sigma} = e \in \Sigma}
\end{mathpar}

Note that some computations, specifically $\lam{x}{e}$,
$\record[\ell \in L]{\ell = e_\ell}$, and $\return{v}$, do not reduce
and may be considered values in other formulations.  Here, we call
them \emph{terminal computations} and use the judgment $e \term$ to
identify them.
\begin{mathpar}
  \infer{\lam{x}{e} \term}{}
  \and
  \infer{\record[\ell \in L]{\ell = e_\ell} \term}{}
  \and
  \infer{\return{v} \term}{}
\end{mathpar}

We will silently use simple properties of computations in the remainder of the
paper which follow by straightforward induction.

\begin{lemma}[Computation]
  \label{lm:computation}
  \label{lm:determinism}
  \label{lm:terminal}
  \mbox{}
  \begin{enumerate}
  \item If $e \reduces e'$ and $e \reduces e''$ then $e' = e''$
  \item It is not possible that both $e \reduces e'$ and $e \term$.
  \end{enumerate}
\end{lemma}

\subsection{Some Sample Programs}
\label{sec:sample-programs}

\begin{example}[Computing with Binary Numbers]
  \label{ex:bin}
  We show some example programs for binary numbers in ``little endian''
  representation (least significant bit first) and in standard form, that is, without leading zeros.
  \begin{tabbing}
    $\p{\ms{bin}} =$ \= $\plusn{\mb{e} : \one, \mb{b0} : \ms{bin}, \mb{b1} : \ms{bin}}$ \\
    $\p{\ms{std}} =$ \> $\plusn{\mb{e} : \one, \mb{b0} : \ms{pos}, \mb{b1} : \ms{std}}$ \\
    $\p{\ms{pos}} =$ \> $\plusn{\phantom{\mb{e} : \one,}\, \mb{b0} : \ms{pos}, \mb{b1} : \ms{std}}$
  \end{tabbing}
  We expect the subtyping relationships $\ms{pos} \leq \ms{std} \leq{\ms{bin}}$ to
  hold, because every positive standard number is a standard number, and every standard
  number is a binary number.  According to our definition and rules in
  Sections~\ref{sec:semantic-typing} and~\ref{sec:syntactic-typing} these will hold semantically as well as syntactically.

  We now show some simple definitions $f : \n{\sigma} = e$.
  \begin{tabbing}
    $\ms{six}$ \= $\null : \up \ms{pos}
    \null = \return{\inj{\mb{b0}}{\inj{\mb{b1}}{\inj{\mb{b1}}{\inj{\mb{e}}{\pair{}}}}}}$
  \end{tabbing}
  The increment function on binary numbers implements the carry with
  a recursive call, which has to be wrapped in a let/return.
  \begin{tabbing}
    $\ms{inc}$ \= $\null : \ms{std} \imp \up \ms{pos}$ \\
    \> $\null = \lambda x.\, \mathsf{match}\; x$ \= $(\,\inj{\mb{e}}{u} \Rightarrow \return{\inj{\mb{b1}}{\inj{\mb{e}}{u}}}$ \\
      \> \> $\mid \inj{\mb{b0}}{x'} \Rightarrow \return{\inj{\mb{b1}}{x'}}$ \\
      \> \> $\mid \inj{\mb{b1}}{x'} \Rightarrow \letup{y'}{\app{\ms{inc}}{x'}}{\return{\inj{\mb{b0}}{y'}}}\,)$
  \end{tabbing}
  By subtyping, we also have $\ms{inc} : \ms{std} \imp \up \ms{std}$, for example, but
  \emph{not} $\ms{inc} : \ms{bin} \imp \up \ms{bin}$ since $\ms{bin} \not\leq \ms{std}$.
  However, the definition could be separately checked against this type, which points
  towards an eventual need for intersection types.

  The following incorrect version of the decrement function does \emph{not} have
  the indicated desired type!
  \begin{tabbing}
    $\ms{dec}_0$ \= $\null : \ms{pos} \imp \up \ms{std}$ \qquad \textbf{\%\ incorrect!} \\
    \> $\null = \lambda x.\, \mathsf{match}\; x$ \= $(\, \inj{\mb{b0}}{x'} \Rightarrow
    \letup{y'}{\app{\ms{dec}_0}{x'}}{\return{\inj{\mb{b1}}{y'}}}$ \\
    \> \> $\mid \inj{\mb{b1}}{x'} \Rightarrow \return{\inj{\mb{b0}}{x'}}\,)$
  \end{tabbing}
  The error here is quite precisely located by the bidirectional type checker (see
  Section~\ref{sec:bidirectional}): When we inject $\inj{\mb{b0}}{x'}$ in the second
  branch it is not the case that $x' : \ms{pos}$ as required for standard numbers!  And,
  indeed,
  $\app{\ms{dec}_0}{\inj{\mb{b1}}{\inj{\mb{e}}{\pair{}}}} \reduces^*
  \return{\inj{\mb{b0}}{\inj{\mb{e}}{\pair{}}}}$ which is not in standard form.  On the other hand,
  the fact that a branch for $\inj{\mb{e}}{u}$ is missing is correct because the type
  $\ms{pos}$ does not have an alternative for this label.

  We can fix this problem by discriminating one more level of the input (which could
  be made slightly more appealing by a compound syntax for nested pattern matching).
  \begin{tabbing}
    $\ms{dec}$ \= $\null : \ms{pos} \imp \up \ms{std}$ \\
    \> $\null = \lambda x.\, \mathsf{match}\; x$ \= $(\, \inj{\mb{b0}}{x'} \Rightarrow
    \letup{y'}{\app{\ms{dec}}{x'}}{\return{\inj{\mb{b1}}{y'}}}$ \\
    \> \> $\mid \inj{\mb{b1}}{x'} \Rightarrow \mathsf{match}\; x'$ \= $(\, \inj{\mb{e}}{u} \Rightarrow \return{\inj{\mb{e}}{u}}$ \\
    \> \> \> $\mid \inj{\mb{b0}}{x''} \Rightarrow \return{\inj{\inj{\mb{b0}}{\mb{b0}}}{x''}}$ \\
    \> \> \> $\mid \inj{\mb{b1}}{x''} \Rightarrow \return{\inj{\inj{\mb{b0}}{\mb{b1}}}{x''}}\,)\,)$
  \end{tabbing}
\end{example}

\begin{example}[Computing with Streams]
  We present an example of a type with mixed polarities: a stream of standard numbers with
  a finite amount of padding between consecutive numbers.  Programmer's intent is for the
  stream to be lazy and infinite, i.e., no end-of-stream is provided.  But because we do not
  restrict recursion even a well-typed implementation may diverge and fail to produce another
  number.  On the other hand, the padding must always be finite because the meaning of
  positive types is inductive.  We present padded streams as two mutually dependent type
  definitions, one positive and one negative.  Because our type definitions are
  equirecursive this isn't strictly necessary, and we could just substitute out the
  definition of $\n{\ms{pstream}}$.

  For our example, we also define a subtype with zero padding, as forcing a
  single padding label $\mb{none}$ between any two elements could also be expressed.
  \begin{tabbing}
    $\n{\ms{pstream}} = \up (\ms{std} \tensor \ms{padding})$ \\
    $\p{\ms{padding}} = \plusn{\mb{none} : \ms{padding}, \mb{some} : \dn \ms{pstream}}$ \\[1ex]
    $\n{\ms{zstream}} = \up (\ms{std} \tensor \plusn{\mb{some} : \dn \ms{zstream}})$
  \end{tabbing}
  In $\ms{zstream}$, we see the significance of variant record types with just
  one label: $\mb{some}$.  We exploit this in Section~\ref{sec:isorecursive} to interpret
  isorecursive types into equirecursive ones.  We have that
  $\ms{zstream} \leq \ms{pstream}$, which means we can pass a stream with zero padding
  into any function expecting one with arbitrary padding.

  We now program two mutually recursive functions to create a stream with zero
  padding from a stream with arbitrary (but finite!) padding.
  \begin{tabbing}
    $\ms{compress} : (\dn \ms{pstream}) \imp \ms{zstream}$ \\
    $\ms{omit} : \ms{padding} \imp \ms{zstream}$ \\[1ex]
    $\ms{compress} = \lambda s.\,$ \= $\letup{np}{\force{s}}{}$ \\
    \> $\letpair{n,p}{np}{\return{\pair{n,\inj{\mb{some}}{\thunk{(\app{\ms{omit}}{p})}}}}}$ \\
    $\ms{omit} = \lambda p.\, \ms{match}\, p$ \= $(\, \inj{\mb{none}}{p'} \Rightarrow \app{\ms{omit}}{p'}$ \\
    \> $\mid \inj{\mb{some}}{s} \Rightarrow \app{\ms{compress}}{s} \,)$
  \end{tabbing}
\end{example}

\begin{example}[Omega]
  As a final example in this section we consider the embedding of the untyped
  $\lambda$-calculus.  The untyped term under consideration is
  $\app{(\lam{x}{\app{x}{x}})}{(\lam{x}{\app{x}{x}})}$.  The first thing to notice is
  that this term is not even \emph{syntactically} well-formed because $x$ stands for a
  value, but in $\app{x}{x}$ the function parts needs to be an expression.  Closely
  related is that the ``usual'' definition for the embedding of the untyped
  $\lambda$-calculus (see, for example,~\cite{Harper16book})
  $\ms{U} = \ms{U} \imp \ms{U}$ isn't properly polarized.
  So, we define it as $\n{\ms{U}} = (\dn \ms{U}) \imp \ms{U}$ instead:
    \begin{tabbing}
      $\omega : (\dn \ms{U}) \imp \ms{U}$ \qquad\qquad\= $\Omega : \ms{U}$ \\
      $\omega = \lam{x}{\app{(\force{x})}{x}}$ \> $\Omega = \app{\omega}{(\thunk{\omega})}$
    \end{tabbing}
  Because our type definitions are equirecursive, both of these definitions are
  well-typed.  Moreoever, we also have $\omega : \ms{U}$ and in fact the embedding of
  every untyped $\lambda$-term will have type $\ms{U}$.
  We also observe that
  $\app{\omega}{(\thunk{\omega})} \reduces^3 \app{\omega}{(\thunk{\omega})}$ and therefore
  represents a well-typed diverging term.  Of course, $f : \ms{U} = f$ is also well-typed
  and reduces to itself in one step.

  Remarkably, with our notion of \emph{semantic typing} we will see that $\Omega$
  will have \emph{every} type $\n{\sigma}$ and not just $\ms{U}$ (Appendix~\ref{app:semtypes}, Example~\ref{ex:omega-full})!
\end{example}

\section{Semantic Typing}
\label{sec:semantic-typing}

Our aim is to justify both typing and subtyping by semantic means.  We therefore start
with \emph{semantic typing} of closed values and computations, written $\sof{v}{\p{\tau}}$
and $\sof{e}{\n{\sigma}}$.  From this we can, for example, define semantic subtyping for
positive types $\p{\tau} \subseteq \p{\sigma}$ as
$\forall v.\, \sof{v}{\p{\tau}} \mimp \sof{v}{\p{\sigma}}$.

Conceptually, semantic typing is a mixed inductive/coinductive definition.  Values are
typed inductively, which yields the correct interpretation of purely positive types such
as natural numbers, lists, or trees, describing finite data structures.  Computations are
typed coinductively because they include the possibility of infinite computation by
unbounded recursion.  While we assume we can observe the structure of values, computations
$e$ cannot be observed directly.  Different notions of observation for computation would
yield different definitions of semantic typing.  For our purposes, since we want to allow
unfettered recursion, we posit we can (a) observe the fact that a computation \emph{steps}
according to our dynamics, even if we cannot examine the computation itself, and (b) when
a computation is \emph{terminal} we can observe its behavior by applying elimination forms
(for types $\p{\tau} \imp \n{\sigma}$ and
$\with*[\ell \in L]{\ell \colon \n{\sigma_\ell}}$) or by observing its returned value (for
the type $\up \p{\tau}$).

Besides capturing a certain notion of observability, our semantics incorporates the usual
concept of \emph{type soundness} which is important both for implementations and for
interpreting the results of computations.  These are:
\begin{description}
\item[Semantic Preservation] (Theorem~\ref{thm:sem-preservation}) If $\sof{e}{\n{\sigma}}$ and
  $e \reduces e'$ then $\sof{e'}{\n{\sigma}}$.
\item[Semantic Progress] (Theorem~\ref{thm:sem-progress}) If $\sof{e}{\n{\sigma}}$ then either
  $e \reduces e'$ for some $e'$ or $e$ is \emph{terminal} (but not both).  This captures
  the usual slogan that ``\textit{well-typed programs do not go wrong}''~\cite{Milner78}.
  An implementation will not accidentally treat a pair as a function or try to decompose a
  function as if it were a pair.
\item[Semantic Observation] If $\sof{v}{\p{\tau}}$ then the structure of the value $v$ is
  determined (inductively) by the type $\p{\tau}$.  Similarly, a \emph{terminal
  computation} $\sof{e}{\up \p{\tau}}$ must have the form $e = \return{v}$ with
  $\sof{v}{\p{\tau}}$.
\end{description}
These combine to the following: if we start a computation for $\sof{e}{\up \p{\tau}}$ then
either $e \reduces^* \return{v}$ for an observable value $\sof{v}{\p{\tau}}$ after a
finite number of steps, or $e$ does not terminate.

These are close to their usual syntactic analogues, but the fact that we do not rely on
any form of syntactic typing is methodologically significant.  For example, if we have a
program that does not obey a syntactic typing discipline but behaves correctly according
to our semantic typing, our results will apply and this program, in combination with
others that are well typed, will both be safe (semantic progress) and return meaningfully
observable results (semantic preservation and observation).
This point has been made passionately by Dreyer et al.~\cite{Dreyer19sigplan} and applied, for
example, to trusted libraries in Rust~\cite{Jung18popl}.  Another example can be found in
gradual typing~\cite{Garcia20wgt,New19popl}.  As long as we can \emph{prove} by any means
that the ``dynamically typed'' portion of the program is semantically
well-typed (even if not syntactically so), the combination is sound and can be executed
without worry, returning a correctly observable result.  A third example is provided by
\emph{session types} for message-passing concurrency~\cite{Hinrichsen21cpp}. While it is important to
have a syntactic type discipline, processes in a distributed system may be programmed in a
variety of languages some of which will have much weaker guarantees. Being able to
\emph{prove} their semantic soundness then guarantees the behavioral soundness of the
composed system.

Semantic typing in the context of call-by-push-value is
well-suited for encoding computational effects, such as input/output,
memory mutation, nontermination, etc.
Call-by-push-value was designed as a study for the $\lambda$-calculus with
effects~\cite[Sec.\ 2.4]{Levy01phd}, stratifying terms into values
(which have no side-effects) and computations (which might). Through the lens of
semantic typing, we can ensure behavioral soundness in the presence of
effects.

\subsection{Semantic Typing with Observation Depth}

Despite the extensive work on mixed inductive and coinductive
definitions~\cite{Abel07aplas,Barwise89csli,Cockett01cmcs,Cohen20ijcar,Danielsson10mpc,Hermida98ic,Komendantsky11tfp,Lepigre16corr,Nakata10sos,Park79ass,Raffa94phd},
there is no widely accepted style in presenting such definitions and reasoning with them concisely in
an mathematical language of discourse.  With some regret, we therefore present
our semantic definition by turning the coinductive part into an inductive one, following
the basic idea underlying \emph{step indexing}~\cite{Ahmed04phd,Ahmed06esop,Appel01toplas,Dreyer09lics}.
Since the coinduction has priority over the induction,
arguments proceed by nested induction, first over the step index and second over the
structure of the inductive definition.
This representation of mixed definitions implies that reasoning
over step indices has lexicographic priority over values.

An alternative point of view is provided by \emph{sized types}~\cite{Abel13icfp,Abel16jfp}.
Both sized types and step
indexing employ the same concept of observation depth; however, for sized types,
we would observe data constructors, whereas for step indexing we observe
computation steps. General recursion is supported in our system because
``productivity'' in the negative layer means that computations can step rather
than produce a data constructor.  The step index is actually the (universally
quantified) observation depth for a coinductively defined predicate. We do not index the
(existentially quantified) size of the inductive predicate but use its structure directly
since values are finite and become smaller.
All step indices $k$, $i$ and occasionally $j$
range over natural numbers.  We use three judgments,
\begin{enumerate}
\item\label{pt:def-comp} $\sof[k]{e}{\n{\sigma}}$ ($e$ has semantic type $\n{\sigma}$ at index $k$)
\item\label{pt:def-term} $\sof*[k+1]{e}{\n{\sigma}}$ (terminal $e$ has semantic type $\n{\sigma}$ at index $k+1$)
\item\label{pt:def-val} $\sof[k]{v}{\p{\tau}}$ ($v$ has semantic type $\p{\tau}$ at index $k$)
\end{enumerate}
They should be defined by nested induction, first on $k$ and second on the structure of
$v/e$, where part~\ref{pt:def-term} can rely on part~\ref{pt:def-comp} for a computation
that is not terminal.  We write $v < v'$ when $v$ is a strict subexpression of $v'$.  The
clauses of the definition
can be found in Figure~\ref{fig:semantic-typing}.

\begin{figure}
\begin{align*}
  \sof[k]{v}{t} &\defd \mbox{$\sof[k]{v}{\p{\tau}}$ for $t = \p{\tau} \in \Sigma$} \\
  \sof[k]{v}{\p{\tau_1} \tensor \p{\tau_2}} &\defd \mbox{$v = \pair{v_1,v_2}$, $\sof[k]{v_1}{\p{\tau_1}}$, and
                                          $\sof[k]{v_2}{\p{\tau_2}}$ for some $v_1$, $v_2$} \\
  \sof[k]{v}{\one} &\defd \mbox{$v = \pair{}$} \\
  \sof[k]{v}{\plus*[\ell \in L]{\ell \colon \p{\tau_\ell}}}
                &\defd \mbox{$v = \inj{j}{v_j}$ and $\sof[k]{v_j}{\p{\tau_j}}$ for some $j \in L$} \\
  \sof[k]{v}{\dn \n{\sigma}} &\defd \mbox{$v = \thunk{e}$ and $\sof[k]{e}{\n{\sigma}}$ for some $e$}\\[1em]
  \sof[0]{e}{\n{\sigma}} &\quad\mbox{always} \\
  \sof[k+1]{e}{\n{\sigma}} &\defd \mbox{($e \reduces e'$ and $\sof[k]{e'}{\n{\sigma}}$) or
                             ($e \term$ and $\sof*[k+1]{e}{\n{\sigma}}$)} \\
  \sof*[k+1]{e}{s} &\defd \mbox{$\sof*[k+1]{e}{\n{\sigma}}$ for $s = \n{\sigma} \in \Sigma$} \\
  \sof*[k+1]{e}{\p{\tau} \imp \n{\sigma}} &\defd \mbox{$\sof[k+1]{\app{e}{v}}{\n{\sigma}}$
                                          for all $i \leq k$ and $v$ with $\sof[i]{v}{\p{\tau}}$} \\
  \sof*[k+1]{e}{\with*[\ell \in L]{\ell \colon \n{\sigma_\ell}}}
                         &\defd \mbox{$\sof[k+1]{\proj{j}{e}}{\n{\sigma_j}}$ for all $j \in L$} \\
  \sof*[k+1]{e}{\up \p{\tau}} &\defd \mbox{$e = \return{v}$ for some $\sof[k]{v}{\p{\tau}}$}
  \\[1em]
  \sof{v}{\p\tau} & \defd \mbox{$\sof[k]{v}{\p\tau}$ for all $k$} \\
  \sof{e}{\n\sigma} & \defd \mbox{$\sof[k]{e}{\n\sigma}$ for all $k$}
\end{align*}
\caption{Definition of Semantic Typing}
\label{fig:semantic-typing}
\end{figure}

A few notes on these definitions.  When expanding type definitions $t = \p{\tau}$ and
$s = \n{\sigma}$ we rely on the assumption that type definitions are contractive, so one
of the immediately following cases will apply next.  This means that unlike many
definitions in this style the types do not necessarily get smaller.  For the inductive
part (typing of values), the values do get smaller and for the coinductive part (typing of
computations) the step index will get smaller because in the case of functions and records
the constructed expression is not terminal.

A number of variations on this definition are possible.  A particularly interesting one
avoids decreasing the step index unless recursion is
unrolled~\cite{Ahmed06esop,Dreyer09lics,New19popl} so sources of nontermination can be
characterized more precisely.  It may also be possible to keep the step index constant
when analyzing a terminal computation of type $\up \p{\tau}$.  Stripping the $\mathsf{return}$
constructor constitutes a form of observation and therefore decreasing the index seems
both appropriate and simplest.

The quantification over $i \leq k$ in the case of terminal computations of function type
seems necessary because we need the relation to be \emph{downward closed} so that it
defines a \emph{deflationary fixed point}~\cite{Abel12fics,Gradel03lics}.
Values and computations are then semantically well-typed if they are well-typed for
\emph{all} step indices.

\begin{lemma}[Downward Closure]
  \label{lm:downward-closure}\mbox{}
  \begin{enumerate}
  \item\label{pt:dc-comp} $\sof[k]{e}{\n{\sigma}}$ implies $\sof[i]{e}{\n{\sigma}}$ for all $i \leq k$
  \item\label{pt:dc-term} $\sof*[k+1]{e}{\n{\sigma}}$ implies $\sof*[i+1]{e}{\n{\sigma}}$ for all $i \leq k$
  \item\label{pt:dc-val} $\sof[k]{v}{\p{\tau}}$ implies $\sof[i]{v}{\p{\tau}}$ for all $i \leq k$
  \end{enumerate}
\end{lemma}
\begin{proof}
  By a routine nested induction on $k$ and the structure of $v/e$ where part~\ref{pt:dc-term} can
  appeal to part~\ref{pt:dc-comp} when $e$ is not terminal.
\end{proof}

Here are some semantic types that can easily be verified (see
Appendix~\ref{app:semtypes}).

\begin{example}[Semantic Typing]
  \label{ex:semantic-typing}
  \mbox{}
  \begin{enumerate}
  \item $\sof{\lam{x}{\return{x}}}{\p{\tau} \imp \up \p{\tau}}$ for all $\p\tau$.
  \item Define $s_0 = \one \imp s_0$ and $e_0 = \lam{x}{e_0}$.  Then $\sof{e_0}{s_0}$.
  \item Define $\omega = \lam{x}{\app{(\force{x})}{x}}$ and $\Omega = \app{\omega}{(\thunk{\omega})}$.
    Then $\sof{\Omega}{\n{\sigma}}$ for every $\n{\sigma}$.
  \item Define $t_0 = \one \tensor t_0$.  Then there is no $v$ such that $\sof{v}{t_0}$.
  \item Assume $\sof{e}{\n{\rho}}$ for some $\n{\rho}$.  Then
    $\sof{e}{t_0 \imp \n{\sigma}}$ for every $\n{\sigma}$.
  \end{enumerate}
\end{example}

\subsection{Properties of Semantic Typing}

The properties of semantic preservation and progress follow immediately just by applying
the definitions and Lemma~\ref{lm:terminal}, so we elide their proofs.

\begin{theorem}[Semantic Preservation]
  \label{thm:sem-preservation}
  If $\sof{e}{\n{\sigma}}$ and $e \reduces e'$ then $\sof{e'}{\n{\sigma}}$.
\end{theorem}

\begin{theorem}[Semantic Progress]
  \label{thm:sem-progress}
  If $\sof{e}{\n{\sigma}}$ then either $e \reduces e'$ or $e$ is terminal, but not both.
\end{theorem}

\section{Subtyping}
\label{sec:subtyping}

The semantics of subtyping is quite easy to express using semantic typing.
\begin{definition}[Semantic Subtyping]
  \label{def:sem-subtyping}
  \mbox{}
  \begin{enumerate}
  \item $\p{\tau} \subseteq \p{\sigma}$ iff $\sof{v}{\p{\tau}}$ implies $\sof{v}{\p{\sigma}}$
    for all $v$.
  \item $\n{\tau} \subseteq \n{\sigma}$ iff $\sof{e}{\n{\tau}}$ implies $\sof{e}{\n{\sigma}}$
    for all $e$.
  \end{enumerate}
\end{definition}
We would now like to give a syntactic definition of subtyping that expresses an algorithm
and show it both sound and complete with respect to the given semantic definition.  The
intuitive rules for subtyping shouldn't be surprising, although to our knowledge our
formulation is original.

\subsection{Empty and Full Types}
\label{sec:empty-full}

A first observation is that $\p{\tau} \subseteq \p{\sigma}$ whenever $\p{\tau}$ is an empty
type, regardless of $\p{\sigma}$, because the necessary implication holds vacuously.  So
we need an algorithm to determine \emph{emptiness of a positive type}.  For the most
streamlined presentation (which is also most suitable for an implementation) we first put
the signature into a normal form that alternates between structural types and type
names.
\begin{align*}
  \p{\tau} &\Coloneqq t_1 \tensor t_2 \mid \one \mid \plus*[\ell \in L]{\ell \colon t_\ell} \mid \dn s \\
  \n{\sigma} &\Coloneqq t \imp s \mid \with*[\ell \in L]{\ell \colon s_\ell} \mid \up t \\
  \Sigma & \Coloneqq \cdot \mid \Sigma, t = \p{\tau} \mid \Sigma, s = \n{\sigma} \mid \Sigma, f : \n{\sigma} = e
\end{align*}

A usual presentation of emptiness maintains a collection of recursive types in a context
in order to do a kind of loop detection.  For example, the type $t = \one \tensor t$ is
empty because we may assume that $t$ is empty while testing $\one \tensor t$.  Instead, we
express this and similar kinds of arguments using valid circular reasoning.  If one were
to formalize it, it would be in \textbf{CLKID$^\omega$}~\cite{Brotherston11jlc}, although
the succedent of any sequent is either empty or a singleton (as in
\textbf{CLJID$^\omega$}~\cite{Berardi18cmcs}).

We construct circular derivations for $t \emp$ where $t$ is a positive type name.  Note
that negative types are never empty.  We can form a valid cycle when we
encounter a goal $t \emp$ as a proper subgoal of $t \emp$. Since we fix a
signature $\Sigma$ once and for all before defining each judgment such as emptiness or
subtyping, we omit the index $\Sigma$ since it never changes.  The rules can be found
in Figure~\ref{fig:empty}.
\begin{figure}
\begin{mathpar}
  \infer[\remp{\oplus}]
  {t \emp}
  {t = \plus*[\ell \in L]{\ell : t_\ell} \in \Sigma &
    t_j \emp \; (\forall j \in L)}
  \and
  \mbox{(no rules for $t = \one$ or $t = \dn s$)}
  \and
  \infer[\remp{\tensor}_1]
  {t \emp}
  {t = t_1 \tensor t_2 \in \Sigma & t_1 \emp}
  \and
  \infer[\remp{\tensor}_2]
  {t \emp}
  {t = t_1 \tensor t_2 \in \Sigma & t_2 \emp}
\end{mathpar}
\caption{Circular Derivation Rules for Emptiness}
\label{fig:empty}
\end{figure}
\begin{example}
  We continue Example~\ref{ex:semantic-typing}, part (4), building a formal circular
  derivation.  We first bring the signature into normal form, $\Sigma =\{
    u_0 = \one,\enspace
    t_0 = u_0 \tensor t_0
  \}$,
  and then construct
  \[
    \infer[\remp{\tensor}_2]
    {t_0 \emp}
    {t_0 = u_0 \tensor t_0
      & \deduce{t_0 \emp}
      {\cycle{}}}
  \]
\end{example}

\begin{theorem}[Emptiness]
  \label{thm:empty}
  If $t \emp$ then for all $k$ and $v$, $\nsof[k]{v}{t}$.
\end{theorem}
\begin{proof}
  We interpret the judgment $t \emp$ semantically as $\sof[k]{v}{t} \vdash \cdot$ (which
  expresses $\nsof[k]{v}{t}$ in a sequent), where $t$ is given and $k$ and $v$ are
  parameters and therefore implicitly universally quantified.  The proof of this judgment is
  carried out in a circular metalogic.  We translate each inference rule for $t \emp$
  into a derivation for $\sof[k]{v}{t} \vdash \cdot$, where each unproven subgoal
  corresponds to a premise of the rule.  When the derivation of $t \emp$ is closed
  by a cycle, the corresponding derivation of $\sof[k]{v}{t} \vdash \cdot$ is
  closed by a corresponding cycle in the metalogic.  The cases can be found
  in Appendix~\ref{app:emptiness}.
\end{proof}

Next we symmetrically define what it means for a computation type $\n{\sigma}$ to be
\emph{full}, namely that it is inhabited by \emph{every (semantically well-typed)
  computation}.  A simple example is the type $\with*{\,}$, that is, the lazy record
without any fields.  It contains every well-typed expression because \emph{all}
projections (of which there are none) are well-typed.  It turns out the fullness is
directly defined from emptiness.

We may construct a derivation using the following rules.  It could be circular,
since the judgment $t \emp$ allows circular derivations.
\begin{mathpar}
  \infer[\rfull{\imp}]
  {s \full}
  {s = t_1 \imp s_2 \in \Sigma & t_1 \emp}
  \and
  \infer[\rfull{\with}]{s \full}{s = \with*{\,} \in \Sigma}
  \and
  \mbox{(no rule for $s = \up t$)}
\end{mathpar}

We interpret $s \full$ as the entailment $\sof[k]{e}{r} \vdash \sof[k]{e}{s}$.  In other
words, we are assuming that $e$ is semantically well-typed at some $r$ and use that to
show that it then will also be well-typed at the unrelated $s$.

\begin{theorem}[Fullness]
  \label{thm:full}
  If $s \full$ then $\sof[k]{e}{r}$ implies $\sof[k]{e}{s}$ for all $k$, $e$, and $r$.
\end{theorem}
\begin{proof}
  (see Appendix~\ref{app:fullness})
\end{proof}

Note that there is no rule that would allow us to conclude that $s = t_1 \imp s_2$ is full
if $s_2$ is full.  Such a rule would be unsound: consider $\sof{\record{\,}}{\with*{\,}}$.
It is not the case that $\sof{\record{\,}}{\one \imp \with*{\,}}$, so
$\one \imp \with*{\,}$ is not full, even though $\with*{\,}$ is.  Similarly,
$\sof{\lam{x}{\record{\,}}}{\one \imp \with*{\,}}$ but
$\nsof{\lam{x}{\record{\,}}}{\with*{l : \with*{\,}}}$, so $\with*{l : \with*{\,}}$ is not
full.

\subsection{Syntactic Subtyping}
\label{sec:syntactic-sub}

The rules for syntactic subtyping build a \emph{circular derivation} of $\p{t} \leq \p{u}$
and $\n{s} \leq \n{r}$.  A circularity arises when a goal $t \leq u$ or $s \leq r$ arises
as a subgoal strictly above a goal  that is of one of these two forms.
In general, we use $t$ and $u$ to stand for
positive type names and $s$ and $r$ for negative type names without annotating those
names.  The polarity will also be clear from the context.  Moreover, in the interest of
saving space, we write $t = \p{\tau}$ and $s = \n{\sigma}$ when these definitions are in
the fixed global signature $\Sigma$.  The rules can be found in Figure~\ref{fig:subtyping}.
In particular, we would like to highlight the $\p{\rsub{\bot}}$,
$\n{\rsub{\bot}}$, and $\rsub{\top}$ rules, which incorporate emptiness and fullness into
syntactic subtyping.  For example, among other subtypings, the $\p{\rsub{\bot}}$ rule
establishes $t \leq u$ whenever $t = t_1 \tensor t_2$ and either $t_1 \emp$ or $t_2 \emp$.

\begin{figure}
\begin{mathpar}
  \infer[\rsub{\tensor}]
  {t \leq u}
  {t = t_1 \tensor t_2 & u = u_1 \tensor u_2
    & t_1 \leq u_1 & t_2 \leq u_2}
  \and
  \infer[\rsub{\one}]
  {t \leq u}
  {t = \one & u = \one}
  \and
  \infer[\rsub{\oplus}]
  {t \leq u}
  {t = \plus*[\ell \in L]{\ell : t_\ell}
    & u = \plus*[k \in K]{k : u_k}
    & \forall \ell \in L.\, t_\ell \emp \lor (\ell \in K \land t_\ell \leq u_\ell)}
  \and
  \infer[\rsub{\dn}]
  {t \leq u}
  {t = \dn s & u = \dn r & s \leq r}
  \quad
  \infer[\rsub{\imp}]
  {s \leq r}
  {s = t_1 \imp s_2 & r = u_1 \imp r_2
    & u_1 \leq t_1 & s_2 \leq r_2}
  \and
  \infer[\rsub{\up}]
  {s \leq r}
  {s = \up t & r = \up u
    & t \leq u}
  \and
  \infer[\rsub{\with}]
  {s \leq r}
  {s = \with*[\ell \in L]{\ell : s_\ell}
    & r = \with*[j \in K]{j : r_j}
    & \mbox{$\forall j \in K.\, j \in L \land s_j \leq r_j$}}
  \\
  \infer[\p{\rsub{\bot}}]{t \leq u}{t \emp & u = \p{\tau}}
  \quad
  \infer[\n{\rsub{\bot}}]{s \leq r}{s = \up t & t \emp & r = \n{\sigma}}
  \quad
  \infer[\rsub{\top}]{s \leq r}{s = \n{\sigma} & r \full}
\end{mathpar}
\caption{Circular Derivation Rules for Subtyping}
\label{fig:subtyping}
\end{figure}

\begin{example}
  We revisit Example~\ref{ex:bin} to show that $\ms{pos} \leq \ms{std}$.  We have
  annotated each subgoal from the $\rsub{\oplus}$ rule with the corresponding label;
  we have elided the reference to the $\rsub{\oplus}$ rule in the derivation for
  lack of space.
  Again, we normalize the signature before running the algorithm.
  \begin{tabbing}
    $\p{\ms{u}} = \one$ \\
    $\p{\ms{std}} =$ \;\= $\plusn{\mb{e} : \ms{u}, \mb{b0} : \ms{pos}, \mb{b1} : \ms{std}}$ \\
    $\p{\ms{pos}} =$ \> $\plusn{\phantom{\mb{e} : \ms{u},}\, \mb{b0} : \ms{pos}, \mb{b1} : \ms{std}}$
  \end{tabbing}
  \begin{mathpar}
    \infer[]
    {\ms{pos} \leq \ms{std}}
    {\infer[]
      {[\mb{b0}]\; \ms{pos} \leq \ms{pos}\; (*)}
      {[\mb{b0}]\; \deduce{\ms{pos} \leq \ms{pos}}{\cycle{*}}
        & \infer[]
        {[\mb{b1}]\; \ms{std} \leq \ms{std}\; (\dagger)}
        {\infer[\rsub{\one}]{[\mb{e}]\; \ms{u} \leq \ms{u}}{}
          & \deduce{[\mb{b0}]\; \ms{pos} \leq \ms{pos}}{\cycle{*}}
          & \!\deduce{[\mb{b1}]\; \ms{std} \leq \ms{std}}{\cycle{\dagger}}}}
      \deduce{[\mb{b1}]\; \ms{std} \leq \ms{std}}{\vdots}}
  \end{mathpar}
\end{example}

From a circular derivation we now construct a valid circular proof in an intuitionistic
metalogic~\cite{Berardi18cmcs}.  For example, $t \leq u$ is interpreted as
$t \subseteq u$, that is, every value in $t$ is also a value in $u$.  We actually prove a
slightly stronger theorem, namely that for the step index on both sides
can remain the same.

\begin{theorem}[Soundness of Subtyping]
  \label{thm:sd-subtyping}
  \mbox{}
  \begin{enumerate}
  \item If $t \leq u$ then $\sof[k]{v}{t} \vdash \sof[k]{v}{u}$ for all $k$ and $v$ (and so, $t \subseteq u$).
  \item If $s \leq r$ then $\sof[k]{e}{s} \vdash \sof[k]{e}{r}$ for all $k$ and $e$ (and so, $s \subseteq r$).
\end{enumerate}
\end{theorem}

\begin{proof}
  We proceed by a compositional translation of the circular derivation of subtyping into a
  circular derivation in the metalogic.  For each rule we construct a derived rule
  on the semantic side with corresponding premises and conclusion.

  When the subtyping proof is closed due to a cycle, we close the proof in the metalogic
  with a corresponding cycle.  In order for this cycle to be valid, it is critical that
  the judgments in the premises of the derived rule are \emph{strictly smaller} than the
  judgments in the conclusion.  Since our mixed logical relation is defined by nested
  induction, first on the step index $k$ and second on the structure of the value $v$
  or expression $e$, the lexicographic measure $(k,v/e)$ should strictly decrease.
  Some sample cases can be found in Appendix~\ref{app:subtyping}.
\end{proof}

Besides soundness, reflexivity and transitivity of syntactic
subtyping are two other properties that we prove for assurance that the syntactic
subtyping rules are sensible and have no obvious gaps.
These proofs can be found in Appendix~\ref{app:refl-trans}.
Ligatti et~al.~\cite{Ligatti17toplas} also consider a notion of \emph{preciseness}
as a syntactic means for judging the correctness of their syntactic subtyping rules.
As they mention in~\cite[Sec.\ 6.2]{Ligatti17toplas}, this property is highly
language-sensitive, depending on the choice of evaluation strategy (strict vs.\ nonstrict),
where nonstrict subtyping relies on ``which primitives are present in the language,
sometimes in nonorthogonal ways.'' Moreover, preciseness requires syntactically
well-typed counterexamples, whereas we also consider ill-typed terms.
We can straightforwardly prove that syntactic subtyping for purely
positive types (in relation to strict evaluation) is complete with respect to
semantic subtyping.  We leave the preciseness of syntactic subtyping of negative
types for future consideration.

\section{Syntactic Typing and Soundness}
\label{sec:syntactic-typing}
\label{sec:type-sound}

We now introduce a syntactic typing judgment, at the moment without regard to
decidability.  Such a judgment is often called \emph{declarative typing} in contrast with
what is \emph{algorithmic typing} in Section~\ref{sec:bidirectional} (Figure~\ref{fig:bidirectional}).
We prove that all
syntactically well-typed terms are also semantically well-typed.  Conceptually, a
declarative system is \emph{unnecessary} because the bidirectional system is very closely
related, and there are no problems in justifying the soundness of the the bidirectional
system directly with respect to our semantics.  Besides the fact that there is a small
amount of additional bureaucracy (the rules are divided between four judgments instead of
two, and there are two additional rules), it is also the case that the standard versions
of call-by-name and call-by-value use a similar form of declarative typing and are
therefore easier to relate to our system in
Section~\ref{sec:call-by-name-and-value}.

Because all declarations in a signature can be mutually recursive, each declaration
$f : \n{\sigma} = e$ is checked assuming all other declarations are valid.  The soundness
proof below justifies this.  The complete set of judgments and rules with their
corresponding presuppositions can be found in Appendix~\ref{app:syntactic-typing},
Figures~\ref{fig:typing-judgments} and~\ref{fig:declarative-typing}.  For these rules, we need contexts $\Gamma$,
defined as usual with the presupposition that all variables declared in a context
are distinct.
\begin{align*}
  \ctx &\Coloneqq \ctxe \mid \ctx , x{:}\p{\tau}
\end{align*}
The rules for key judgments $\Gamma \vdash v : \p{\tau}$ and
$\Gamma \vdash e : \n{\sigma}$ can be obtained from the bidirectional rules in
Section~\ref{sec:bidirectional} by replacing both $v \checks \p{\tau}$ and
$v \synths \p{\tau}$ with $v : \p{\tau}$ and, similarly, $e \checks \n{\sigma}$ and
$e \synths \n{\sigma}$ with $e : \n{\sigma}$.  Moreover, one should drop the two
annotation rules $\p{\jrule{ANNO}}$ and $\n{\jrule{ANNO}}$ because these are
not in the source language for declarative typing.

We would like to show that the syntactic typing rules are sound with respect to their
semantic interpretation.  For that, we first define simultaneous substitutions $\theta$ of
closed values for variables and $\sof[k]{\theta}{\Gamma}$ for the semantic interpretation
of contexts as sets of substitutions at step index $k$.
\begin{align*}
  \theta &\Coloneqq \cdot \mid \theta, v/x \\
  \sof[k]{(\cdot)}{(\cdot)} &\quad\mbox{always} \\
  \sof[k]{(\theta, v/x)}{(\Gamma, x:\p{\tau})} &\defd \mbox{$\sof[k]{\theta}{\Gamma}$ and $\sof[k]{v}{\p{\tau}}$}
\end{align*}
On the semantic side, we define
\begin{enumerate}
\item $\Gamma \models \sof[k]{v}{\p{\tau}}$ iff for all
  $\sof[k]{\theta}{\Gamma}$ we have $\sof[k]{v[\theta]}{\p{\tau}}$
\item $\Gamma \models \sof[k]{e}{\n{\sigma}}$ iff for all
  $\sof[k]{\theta}{\Gamma}$ we have $\sof[k]{e[\theta]}{\n{\sigma}}$
\end{enumerate}

We now can prove a number of lemmas, one for each syntactic typing rule.  A representative
selection of the lemmas, each written as an admissible rule for semantic typing, can be given by:
\begin{mathpar}
    \infer-{\sof[k]{\app{e}{v}}{\n{\sigma}}}
    {\sof[k]{e}{\p{\tau} \imp \n{\sigma}}
      & \sof[k]{v}{\p{\tau}}}
    \and
    \infer-{\sof[k]{\lam{x}{e}}{\p{\tau} \imp \n{\sigma}}}
    {x : \p{\tau} \models \sof[k]{e}{\n{\sigma}}}
    \and
    \infer-{\sof[k]{\pair{v_1,v_2}}{\p{\tau_1} \tensor \p{\tau_2}}}
    {\sof[k]{v_1}{\p{\tau_1}} & \sof[k]{v_2}{\p{\tau_2}}}
    \and
    \infer-{\sof[k]{\letpair{x,y}{v}{e}}{\n{\sigma}}}
    {\sof[k]{v}{\p{\tau_1} \tensor \p{\tau_2}}
      & x : \p{\tau_1}, y : \p{\tau_2} \models \sof[k]{e}{\n{\sigma}}}
    \and
    \infer-{\sof[k]{\return{v}}{\up \p{\tau}}}
    {\sof[k]{v}{\p{\tau}}}
    \and
    \infer-{\sof[k]{\force{v}}{\n{\sigma}}}
    {\sof[k]{v}{\dn \n{\sigma}}}
    \and
    \infer-{\sof[k]{\letup{x}{e_1}{e_2}}{\n{\sigma}}}
    {\sof[k]{e_1}{\up \p{\tau}}
      & x : \p{\tau} \models \sof[k]{e_2}{\n{\sigma}}}
    \and
    \infer-{\sof[k]{\thunk{e}}{\dn \n{\sigma}}}
    {\sof[k]{e}{\n{\sigma}}}
    \and
    \infer-{\sof[k]{v}{\p{\sigma}}}
    {\sof[k]{v}{\p{\tau}} & \p{\tau} \leq \p{\sigma}}
\end{mathpar}

The proofs are somewhat interesting: some require
induction on $k$, others follow more directly by definition.  Due to a lack of space, the
proofs can be found in Appendix~\ref{app:sd-typing}, each admissible rule formulated as a
separate lemma.

\begin{theorem}[Soundness of Syntactic Typing]\label{thm:type-sound}
  Assume $\sof[k]{\theta}{\Gamma}$.
\begin{enumerate}
\item If $\Gamma \vdash v : \p{\tau}$ then $\sof[k]{v[\theta]}{\p{\tau}}$
\item If $\Gamma \vdash e : \n{\sigma}$ then $\sof[k]{e[\theta]}{\n{\sigma}}$
\end{enumerate}
\end{theorem}

\begin{proof}
  We construct a circular proof based on the typing derivation, and the typing derivations
  for all definitions $f : \n{\sigma} = e \in \Sigma$.  There are three kinds of cases
  (see Appendix~\ref{app:sd-typing} for samples of each):
  \begin{enumerate}
  \item The case of variables $x$ follows by assumption on $\theta$.
  \item In the case of names $f : \n{\sigma} = e \in \Sigma$ we either expand to $e$ or
    close the proof with a cycle if we have expanded $f$ already.
    \item All other rules follow by the lemmas presented above.

    In all these lemmas the step index remains constant for the premises, which is important
    so we can form a circular proof in the case of names.
  \end{enumerate}
\end{proof}

Because soundness is stated for all $\theta$, $\Gamma$, and $k$, we can immediately obtain
corollaries such as that $\cdot \vdash v : \p\tau$ implies that $\sof{v}{\p\tau}$, and
that $\cdot \vdash e : \n\sigma$ implies that $\sof{e}{\n\sigma}$.

\section{Bidirectional Typing}
\label{sec:bidirectional}

We now shift from our declarative typing system into an algorithmic one that describes a
practical decision procedure. We choose to express it as a bidirectional typechecking
algorithm, particularly to avoid inference issues regarding
subsumption~\cite{Jafery17popl} and our extensive use of type names and variant records,
as well as the approach's deep integration with polarized logics~\cite[Section 8.3]{Dunfield19corr}.
Moreover, bidirectional typing is quite robust with respect to
language extensions where various inference procedures are not.

Bidirectional typechecking~\cite{Pierce98popl} has been a popular choice for presenting
algorithmic typing, especially when concerned with subtyping~\cite{Dunfield19popl},
and is decidable for a wide range of rich type systems. This approach splits each of the
typing judgments, $\ctx \vdash v : \p{\tau}$ and $\ctx \vdash e : \n{\sigma}$, into
\emph{checking} (\!\!$\checks$\!\!) and \emph{synthesis} (\!\!$\synths$\!\!) judgments
for values and expressions, respectively:
$\Gamma \vdash v \checks \p{\tau}$, $\Gamma \vdash v \synths \p{\tau}$ and
$\Gamma \vdash e \checks \n{\sigma}$, $\Gamma \vdash e \synths \n{\sigma}$.

\begin{figure}[t]
\begin{mathpar}
  \infer[\irule{\tensor}]{\ctx \vdash \pair{v_1,v_2} \checks \p{\tau_1} \tensor \p{\tau_2}}
  {\ctx \vdash v_1 \checks \p{\tau_1} & \ctx \vdash v_2 \checks \p{\tau_2}}
  \quad
  \infer[\erule{\tensor}]{\ctx \vdash \letpair{x,y}{v}{e} \checks \n{\sigma}}
  {\ctx \vdash v \synths \p{\tau_1} \tensor \p{\tau_2} &
    \ctx , x{:}\p{\tau_1} , y{:}\p{\tau_2} \vdash e \checks \n{\sigma}}
  \\
  \infer[\jrule{VAR}]{\ctx \vdash x \synths \p{\tau} }{x : \p{\tau} \in \ctx}
  \and
  \infer[\irule{\one}]{\ctx \vdash \pair{} \checks \one}{}
  \and
  \infer[\erule{\one}]{\ctx \vdash \letpair{}{v}{e} \checks \n{\sigma}}
  {\ctx \vdash v \synths \one & \ctx \vdash e \checks \n{\sigma}}
  \\
  \infer[\irule{\dn}]{\ctx \vdash \thunk{e} \checks \dn \n{\sigma}}
  {\ctx \vdash e \checks \n{\sigma}}
  \and
  \infer[\erule{\dn}]{\ctx \vdash \force{v} \synths \n{\sigma}}
  {\ctx \vdash v \synths \dn \n{\sigma}}\and
  \infer[\irule{\plus}]{\ctx \vdash \inj{j}{v} \checks \plus*[\ell \in L]{\ell\colon \p{\tau_\ell}}}
  {\text{($j \in L$)} & \ctx \vdash v \checks \p{\tau_j}}
  \and
  \infer[\erule{\plus}]{\ctx \vdash \case[\ell \in L]{v}{{\ell}{x_\ell} => e_\ell} \checks \n{\sigma}}
  {\ctx \vdash v \synths \plus*[\ell \in L]{\ell\colon \p{\tau_\ell}} &
    \metaall{\ell \in L} \ctx , x_\ell{:}\p{\tau_\ell} \vdash e_\ell \checks \n{\sigma}}
  \quad
  \infer[\irule{\imp}]{\ctx \vdash \lam{x}{e} \checks \p{\tau} \imp \n{\sigma}}
  {\ctx , x{:}\p{\tau} \vdash e \checks \n{\sigma}}
  \\
  \infer[\erule{\imp}]{\ctx \vdash \app{e}{v} \synths \n{\sigma}}
  {\ctx \vdash e \synths \p{\tau} \imp \n{\sigma} & \ctx \vdash v \checks \p{\tau}}
  \and
  \infer[\irule{\with}]{\ctx \vdash \record[\ell \in L]{\ell = e_\ell} \checks \with*[\ell \in L]{\ell\colon \n{\sigma_\ell}}}
  {\metaall{\ell \in L} \ctx \vdash e_\ell \checks \n{\sigma_\ell}}
  \and
  \infer[\erule{\with}_k]{\ctx \vdash \proj{j}{e} \synths \n{\sigma_j}}
  {\ctx \vdash e \synths \with*[\ell \in L]{\ell\colon \n{\sigma_\ell}} &
    \text{($j \in L$)}}
  \quad
  \infer[\jrule{NAME}]
  {\Gamma \vdash f \synths \n{\sigma}}
  {f : \n{\sigma} =  e \in \Sigma}
  \quad
  \infer[\irule{\up}]{\ctx \vdash \return{v} \checks \up \p{\tau}}
  {\ctx \vdash v \checks \p{\tau}}
  \and
  \infer[\erule{\up}]{\ctx \vdash \letup{x}{e_1}{e_2} \checks \n{\sigma}}
  {\ctx \vdash e_1 \synths \up \p{\tau} & \ctx , x{:}\p{\tau} \vdash e_2 \checks \n{\sigma}}
  \quad
  \infer[\p{\jrule{SUB}}]
  {\ctx \vdash v \checks \p{\sigma} }
  {\ctx \vdash v \synths \p{\tau} & \p{\tau} \leq \p{\sigma}}
  \\
  \infer[\n{\jrule{SUB}}]
  {\ctx \vdash e \checks \n{\sigma} }
  {\ctx \vdash e \synths \n{\tau} \enspace \n{\tau} \leq \n{\sigma}}
  \quad
  \infer[\!\p{\jrule{ANNO}}]
  {\ctx \vdash (v : \p{\tau}) \synths \p{\tau}}
  {\ctx \vdash v \checks \p{\tau}}
  \quad
  \infer[\!\n{\jrule{ANNO}}]
  {\ctx \vdash (e : \n{\sigma}) \synths \n{\sigma}}
  {\ctx \vdash e \checks \n{\sigma}}
\end{mathpar}
\caption{Bidirectional Typing}
\label{fig:bidirectional}
\end{figure}

We follow the recipe laid out
by~\cite{Davies00icfp,Dunfield04popl}: introduction rules check and elimination
rules synthesize.  More precisely, the \emph{principal judgment}, premise or conclusion,
has the connective being introduced by checking or eliminated by synthesis.

We introduce two new forms of syntactic values $(v : \p{\tau})$ and computations $(e : \n{\sigma})$ which exist
purely for typechecking purposes and are erased before evaluation.
This is not actually used on any of our examples because definitions in
the signature already require annotations.

Applying the recipe, we can easily convert our declarative rules into bidirectional ones,
as laid out in Section~\ref{sec:syntactic-typing}. The only rules we add to the system are
$\p{\jrule{ANNO}}$ and $\n{\jrule{ANNO}}$, which allow us to prove completeness.
All the examples in Section~\ref{sec:sample-programs}
check with these rules and only require type annotations at the top level of the
declarations in the signature.

Due to our use of equirecursive types, the implementation of this system can closely
follow the structure of the rules in Figures~\ref{fig:empty},~\ref{fig:subtyping},
and~\ref{fig:bidirectional}.  First, as mentioned in Section~\ref{sec:empty-full}, we
convert the signature into a normal form that alternates structural types and type names.
Then, we determine all the empty type names using a memoization table for $\p{t} \emp$ to
easily construct circular derivations of emptiness (bottom-up) using the rules in
Figure~\ref{fig:empty}. If constructing such a derivation fails then $\p{t}$ is
nonempty. Fullness is derived from emptiness non-recursively.
From there, we build a memoization table for $\p{t} \leq \p{u}$ and $\n{s} \leq \n{r}$,
for positive and negative type names, so we can construct circular derivations
of subtyping between names (also bottom-up). This happens lazily, only
computing $\p{t} \leq \p{u}$ or $\n{s} \leq \n{r}$ if typechecking
requires this information.

Bidirectional typing, given subtyping, follows the rules
in Figure~\ref{fig:bidirectional}, including the rules for positive and negative subsumption, but
it requires that the types in annotations are also translated to normal form, possibly
introducing new (user-invisible) definitions in the signature.

The theorems (with straightforward proofs) for soundness and completeness of bidirectional
typechecking can be found in Appendix~\ref{app:bidir} (Theorems~\ref{thm:bidir-sound}
and~\ref{thm:bidir-complete}).

\section{Interpretation of Isorecursive Types}
\label{sec:isorecursive}

Our system uses equirecursive types, which allow many subtyping relations since there are no
term constructors for folding recursive types.  Moreover, equirecursive types support the
normal form where constructors are always applied to type names (see
Section~\ref{sec:empty-full}), simplifying our algorithms, their description and
implementations.  Most importantly, perhaps, equirecursive types are more general because
we can directly interpret isorecursive types, which are
embodied by \emph{fold} and \emph{unfold} operators, into our equirecursive
setting and apply our results.

We give a short sketch here; details can be found in Appendix~\ref{app:isorecursive}.  For
every recursive type $\mu \p\alpha.\, \p\tau$ we introduce a definition
$\p{t} = \plus*[]{\foldmu \colon [t/\alpha]\tau}$.  Similarly, for every corecursive type
$\nu \n\alpha.\, \n\sigma$ we introduce a definition
$s^- = \with*[]{\foldnu \colon [s/\alpha]\sigma}$.  Now, the labels $\mathsf{fold_{\mu}}$
and $\mathsf{fold_{\nu}}$ tagging the sole choice of a unary variant or lazy record,
respectively, play exactly the role that the $\mathsf{fold}$ constructor plays for
recursive types.  This entirely straightforward translation is enabled by our
generalization of the binary sum and lazy pairs to variant and lazy records, respectively,
so we can use them in their unary form.

\section{Call-by-Name and Call-by-Value}
\label{sec:call-by-name-and-value}

More familiar than call-by-push-value (CBPV) are the lazy, call-by-name (CBN) and eager, call-by-value (CBV)
operational semantics that underlie the Haskell and ML families of functional programming
languages.  Levy~\cite{Levy06hosc} has shown that both CBN and CBV exist as
fragments of CBPV, exhibiting translations from CBN and
CBV types and terms into the CBPV language.  In this section,
we derive systems of subtyping for CBN and CBV from these translations
into ours and prove them sound and complete.  We discover that they are minor variants
of existing systems for CBN~\cite{Gay05acta} and CBV~\cite{Ligatti17toplas} subtyping.

Because polarized subtyping is able to connect Levy's translations with existing systems
for CBN and CBV subtyping, it serves as further evidence that those
prior translations and our subtyping rules are, in some sense, canonical.  Moreover, it is
yet one more piece of evidence that CBPV is an effective synthesis of
evaluation orders in which to study the theory of functional programming.

\subsection{Call-by-name}
\label{sec:call-by-name}

Consider a CBN language with the following types.
The language of terms and the standard statics and dynamics
can be found in Appendix~\ref{app:cbn}.
\begin{align*}
  \tau, \sigma &\Coloneqq \tau \imp \sigma \mid \tau_1 \tensor \tau_2 \mid \one \mid \plus*[\ell \in L]{\ell\colon \tau_\ell} \mid \with*[\ell \in L]{\ell\colon \tau_\ell}
\end{align*}
In this section, we will focus on function types $\tau \imp \sigma$ and variant record
types $\plus*[\ell \in L]{\ell\colon \tau_\ell}$ and their corresponding terms.

Levy~\cite{Levy06hosc} presents translations, $\cbn*{}$, from CBN types and terms to
CBPV \emph{negative} types and \emph{expressions}, respectively.  An
auxiliary translation, $\dn \cbn*{}$, on contexts is also used.  Here, we elide the
translation of terms other than variables and the terms for function and variant record
types; the full translation on terms can be found in~\cite{Levy06hosc}.
\begin{equation*}
  \begin{gathered}[t]
    \text{\emph{Types}} \\
  \begin{aligned}
    \cbn*{\tau \imp \sigma} &= \dn \cbn{\tau} \imp \cbn{\sigma} \\
    \cbn*{\tau_1 \tensor \tau_2} &= \up (\dn \cbn{\tau_1} \tensor \dn \cbn{\tau_2}) \\
    \cbn*{\one} &= \up \one \\
    \cbn*{\plus*[\ell \in L]{\ell\colon \tau_\ell}} &= \up \plus*[\ell \in L]{\ell\colon \dn \cbn{\tau_\ell}} \\
    \cbn*{\with*[\ell \in L]{\ell\colon \sigma_\ell}} &= \with*[\ell \in L]{\ell\colon \cbn{\sigma_\ell}}
  \end{aligned}
  \end{gathered}
  \qquad
  \begin{gathered}[t]
    \text{\emph{Terms}} \\
  \begin{aligned}
    \cbn*{x} &= \return{x} \\
    \cbn*{\lam{x}{e}} &= \lam{x}{\cbn{e}} \\
    \cbn*{\app{e_1}{e_2}} &= \app{\cbn{e_1}}{(\thunk{\cbn{e_2}})}
  \end{aligned}
  \end{gathered}
\end{equation*}
We also translate type names $t$ to fresh type names $\cbn{t}$, translating the body of $t$'s definition and inserting additional type names as required for the normal form that alternates between structural types and type names.
Levy~\cite{Levy06hosc} proves that well-typed terms are well-typed after the translation to CBPV is applied.
Our syntactic typing rules are the same, so the theorem carries over to our setting.

We adapt the subtyping system of Gay and Hole~\cite{Gay05acta} to a $\lambda$-calculus
from the $\pi$-calculus, which reverses the direction of subtyping from their classical
system and adds empty records, obtaining the CBN syntactic subtyping rules shown
in Figure~\ref{fig:cbn-subtyping}.

These rules introduce a CBN syntactic subtyping judgment $t \leq u$.  To
distinguish it from CBPV syntactic subtyping, we will take care in this
section to always include superscript pluses and minuses for CBPV type
names, with CBN type names being unmarked.  As for CBPV syntactic
subtyping, the rules for CBN subtyping shown in Figure~\ref{fig:cbn-subtyping}
build a \emph{circular derivation}.  Just as before, a circularity arises when a goal
$t \leq u$ arises as a proper subgoal of itself.

\begin{figure}
\begin{mathpar}
  \infer[\rsub{\imp}_{\jrule{N}}]{t \leq u}
  {t = t_1 \imp t_2 & u = u_1 \imp u_2 & u_1 \leq t_1 & t_2 \leq u_2}
  \\
  \infer[\rsub{\tensor}_{\jrule{N}}]{t \leq u}
  {t = t_1 \tensor t_2 & u = u_1 \tensor u_2 & t_1 \leq u_1 & t_2 \leq u_2}
  \and
  \infer[\rsub{\one}_{\jrule{N}}]{t \leq u}
  {t = \one & u = \one}
  \\
  \infer[\rsub{\plus}_{\jrule{N}}]{t \leq u}
  {t = \plus*[\ell \in L]{\ell\colon t_\ell} & u = \plus*[j \in J]{j\colon u_j} &
    \text{($L \subseteq J$)} & \metaall{\ell \in L} t_\ell \leq u_\ell}
  \and
  \infer[\rsub{\with}_{\jrule{N}}]{t \leq u}
  {t = \with*[\ell \in L]{\ell\colon t_\ell} & u = \with*[j \in J]{j\colon u_j} &
   \text{($L \supseteq J$)} & \metaall{j \in J} t_j \leq u_j}
  \\
  \infer[\rsub{\bot}_{\jrule{N}}]{t \leq u}
  {t = \plus*{\,} & u = \sigma}
  \and
  \infer[\rsub{\top}_{\jrule{N}}]{t \leq u}
  {t = \tau & u \full}
  \and
  \infer[\rfull{\with}_{\jrule{N}}]{t \full}
  {t = \with*{\,}}
\end{mathpar}
\caption{Circular Derivation Rules for Call-by-Name Subtyping}
\label{fig:cbn-subtyping}
\end{figure}

These rules are exact analogues of those of Gay and Hole~\cite{Gay05acta}, with one exception.
The three rules involving empty variants and records,
namely $\rsub{\bot}_{\jrule{N}}$, $\rsub{\top}_{\jrule{N}}$, and $\rfull{\with}_{\jrule{N}}$,
have no analogues in~\cite{Gay05acta} only because their language did not
include the corresponding empty internal and external  choice types.

As we will prove below, the CBN subtyping rules in Figure~\ref{fig:cbn-subtyping} are exactly those for which
$t \leq u$ in the CBN language if and only if $\cbn{t} \leq \cbn{u}$ in the
CBPV metalanguage.
We thereby show that our polarized subtyping on the image of Levy's CBN translation is sound and complete with respect to Gay and Hole's CBN subtyping.

Before proceeding to those proofs, it is worth pointing out that many of these
CBN subtyping rules exactly follow CBPV, with a few notable
differences.  First, the $\rsub{\plus}_{\jrule{N}}$ rule does not permit empty branches that do
not occur in the supertype.
This is because the $\dn$ shifts that appear in $\cbn*{\plus*[\ell \in L]{\ell\colon \tau_\ell}}$ prevent each branch from being empty\textemdash there is no emptiness rule for $\dn$ shifts in the CBPV subtyping.
Second, for this CBN language, only types $t = \with*{\,}$ are full.  In
particular, a CBN function type $t = t_1 \imp t_2$ is never full, even though a
CBPV function type $\n{s} = \p{t_1} \imp \n{s_2}$ is full if the argument
type $\p{t_1}$ is empty.  This stems from the $\dn$ shift that appears in the argument type in $\cbn*{\tau \imp \sigma} = \dn \cbn{\tau} \imp \cbn{\sigma}$.
Third, the reader may be
surprised by the omission of an emptiness judgment for CBN types. The
$\rsub{\bot}_{\jrule{N}}$ rule mentions the CBN type $t = \plus*{\,}$, which
looks like it ought to be an empty type\textemdash the CBPV type
$\p{t_0} = \plus*{\,}$ is empty, after all.  Yes, but the CBN translation of
$t = \plus*{\,}$ is in fact the negative type $\cbn{t} = \up \plus*{\,}$, and negative types are
never empty.
Nevertheless, $\cbn{t} = \up \plus*{\,} \leq \cbn{u}$ in this case.

Now we prove that polarized subtyping on the image of Levy's CBN embedding, $\cbn*{}$, is sound and complete with respect to the CBN subtyping rules of Figure~\ref{fig:cbn-subtyping}.
The proofs can be found in Appendix~\ref{app:cbn}.
\begin{theorem}[Soundness of Polarized Subtyping, Call-by-Name]\leavevmode
  \label{thm:cbn-soundness}
  \begin{enumerate}
  \item\label{item:cbn-full-sound} If $\cbn{t} \full$, then $t \full$.
  \item\label{item:cbn-leq-sound} If $\cbn{t} \leq \cbn{u}$, then $t \leq u$.
  \end{enumerate}
\end{theorem}

\begin{theorem}[Completeness of Polarized Subtyping, Call-by-Name]\leavevmode
  \label{thm:cbn-completeness}
  \begin{enumerate}
  \item\label{item:cbn-full-complete} If $t \full$, then $\cbn{t} \full$.
  \item\label{item:cbn-leq-complete} If $t \leq u$, then $\cbn{t} \leq \cbn{u}$.
  \end{enumerate}
\end{theorem}

\subsection{Call-by-Value}
\label{sec:call-by-value}

We can play through a similar procedure for Levy's CBV translation.
Consider a CBV language with the following types.  The language
of terms, typing rules, and standard dynamics can be found in Appendix~\ref{app:cbv}.
\begin{align*}
  \tau, \sigma &\Coloneqq \tau \imp \sigma \mid \tau_1 \tensor \tau_2 \mid \one \mid \plus*[\ell \in L]{\ell\colon \tau_\ell} \mid \with*[\ell \in L]{\ell\colon \sigma_\ell}
\end{align*}

The translations that Levy~\cite{Levy06hosc} presents from CBV types and terms to CBPV \emph{positive types} and \emph{expressions} are as follows.
We only present the translation of variables, function abstractions, and function applications; the full translation on terms can be found in~\cite{Levy06hosc}.
\begin{equation*}
  \begin{gathered}[t]
    \text{\emph{Types}} \\
    \begin{aligned}
      \cbv*{\tau \imp \sigma} &= \dn (\cbv{\tau} \imp \up \cbv{\sigma}) \\
      \cbv*{\tau_1 \tensor \tau_2} &= \cbv{\tau_1} \tensor \cbv{\tau_2} \\
      \cbv*{\one} &= \one \\
      \cbv*{\plus*[\ell \in L]{\ell\colon \tau_\ell}}
        &= \plus*[\ell \in L]{\ell\colon \cbv{\tau_\ell}} \\
      \cbv*{\with*[\ell \in L]{\ell\colon \sigma_\ell}}
        &= \dn \with*[\ell \in L]{\ell\colon \up \cbv{\sigma_\ell}}
    \end{aligned}
  \end{gathered}
  \qquad
  \begin{gathered}[t]
    \text{\emph{Terms}} \\
    \begin{aligned}
      \cbv*{x} &= \return{x} \\
      \cbv*{f} &= \force{f} \text{\ for $f : \tau = e \in \sig$} \\
      \cbv*{\lam{x}{e}} &= \return{(\thunk{(\lam{x}{\cbv{e}})})} \\
      \cbv*{\app{e_1}{e_2}} &= \begin{array}[t]{@{}l@{}}
                                 \letup{x}{\cbv{e_2}}{ \\
                                 \letup{f}{\cbv{e_1}}{ \\ \quad
                                 \app{(\force{f})}{x}}}
                               \end{array} 
    \end{aligned}
  \end{gathered}
\end{equation*}
We also translate type names $t$ to fresh type names $\cbv{t}$, translating the
body of $t$'s definition and inserting additional type names as required for the
normal form that alternates between structural types and type names.

Levy proves that well-typed terms translate to well-typed expressions.
Because our syntactic typing rules are the same as his, his theorem carries over.

We adapt the CBV subtyping system of Ligatti et
al.~\cite{Ligatti17toplas} to our setting, which means that we include variants and lazy records with
width and depth subtyping and replace isorecursive with equirecursive types.  We obtain
the syntactic subtyping rules shown in Figure~\ref{fig:call-by-value}.  Once again, we
will take care to distinguish the CBV syntactic subtyping judgment, $t \leq u$,
from CBPV syntactic subtyping by marking CBPV type names with
pluses and minuses.  The rules shown in Figure~\ref{fig:call-by-value} build
\emph{circular derivations}.
\begin{figure}[ht]
\begin{mathpar}
  \infer[\rsub{\imp}_{\jrule{V}}]{t \leq u}
  {t = t_1 \imp t_2 & u = u_1 \imp u_2 &
    u_1 \leq t_1 & t_2 \leq u_2}
  \and
  \infer[\rsub{\tensor}_{\jrule{V}}]{t \leq u}
  {t = t_1 \tensor t_2 & u = u_1 \tensor u_2 &
    t_1 \leq u_1 & t_2 \leq u_2}
  \and
  \infer[\rsub{\one}_{\jrule{V}}]{t \leq u}
  {t = \one & u = \one}
  \\
  \infer[\rsub{\plus}_{\jrule{V}}]{t \leq u}
  {\begin{array}[b]{@{}l@{}}
     t = \plus*[\ell \in L]{\ell\colon t_\ell} \\[0.5ex]
     u = \plus*[j \in J]{j\colon u_j}
   \end{array} &
    \metaall{\ell \in L \setminus J} t_\ell \emp &
    \metaall{\ell \in L \cap J} t_\ell \leq u_\ell}
  \and
  \infer[\rsub{\with}_{\jrule{V}}]{t \leq u}
  {t = \with*[\ell \in L]{\ell\colon t_\ell} & u = \with*[j \in J]{j\colon u_j} &
    \text{($L \supseteq J$)} &
    \metaall{j \in J} t_j \leq u_j}
  \and
  \infer[\rsub{\bot}_{\jrule{V}}]{t \leq u}
  {t \emp & u = \sigma}
  \and
  \infer[\rsub{\top}^{\imp\imp}_{\jrule{V}}]{t \leq u}
  {t = t_1 \imp t_2 &
   u = u_1 \imp u_2 &
   u_1 \emp}
  \and
  \infer[\rsub{\top}^{\with\imp}_{\jrule{V}}]{t \leq u}
  {t = \with*[\ell \in L]{\ell\colon t_\ell} &
   u = u_1 \imp u_2 &
   u_1 \emp}
  \quad
  \infer[\rsub{\top}^{\imp\with}_{\jrule{V}}]{t \leq u}
  {t = t_1 \imp t_2 &
   u = \with*{\,}}
  \\
  \infer[\remp{\tensor}_{\jrule{V}}{}_i]{t \emp}
  {t = t_1 \tensor t_2 & t_i \emp}
  \and
  \infer[\remp{\plus}_{\jrule{V}}]{t \emp}
  {t = \plus*[\ell \in L]{\ell\colon t_\ell} &
    \metaall{\ell \in L} t_\ell \emp}
  \and
  \text{(no emptiness rules for $\one$, $\imp$, and $\with$)}
\end{mathpar}
\caption{Circular Derivation Rules for Call-by-Value Subtyping}
\label{fig:call-by-value}
\end{figure}

These rules match those of Ligatti et al., with one minor exception that we will
detail below.  As we will prove, these rules are exactly those for which $t \leq u$
in the CBV language if and only if $\cbv{t} \leq \cbv{u}$ in the
CBPV metalanguage.

Before proceeding to the proofs, a few remarks about these rules.  First,
unlike the CBN $\rsub{\plus}_{\jrule{N}}$ rule, the $\rsub{\plus}_{\jrule{V}}$ rule here includes the
possibility that some components of a variant record type may be empty.
More generally, the differences between CBN and
CBV subtyping arise from the differences in emptiness and fullness between the two calculi.
Emptiness and fullness are quite sensitive to the eager/lazy distinction between
the two evaluation strategies.
Because this distinction manifests in almost every layer of a complex type,
the two subtyping systems diverge more than one might expect.

Second, besides the adaptions mentioned above, the rules of Figure~\ref{fig:call-by-value}
diverge from those of Ligatti et al.\ in only one way.  Ligatti et
al.~\cite{Ligatti17toplas} have the rule ``$t \leq u$ if $u = u_1 \imp u_2$ and
$u_1 \emp$'' that generalizes the $\rsub{\top}^{\imp\imp}_{\jrule{V}}$,
$\rsub{\top}^{\with\imp}_{\jrule{V}}$, and $\rsub{\top}^{\imp\with}_{\jrule{V}}$ rules of
Figure~\ref{fig:call-by-value} (assuming that Ligatti et al.\ would also have ``$t \leq u$
if $u = \with*{\,}$'' if they had included lazy records in their language).

Somewhat unexpectedly, polarized subtyping on the image of Levy's CBV
translation would be incomplete with respect to this more general rule.  This is because the
$\dn$ shift inserted by Levy's translation acts as a barrier to fullness: ``$t \leq u$ if
$u = \dn r$ and $r \full$'' would be unsound in polarized subtyping.  For example, Ligatti et al.\ have
$\one \leq \mathbf{0} \imp \one$ for an empty type $\mathbf{0}$, but we do not have
$\cbv{\one} = \one \leq \dn (\mathbf{0} \imp \up \one) = \cbv*{\mathbf{0} \imp \one}$
because the unit value $\pair{}$ does not have type $\dn (\mathbf{0} \imp \up \one)$.
This phenomenon is primarily of theoretical interest since it is confined to functions
that can never be applied to any arguments and empty records (and only when they are
compared against CBV types $t_1 \tensor t_2$, $\one$, and
$\plus*[\ell \in L]{\ell\colon t_\ell}$).  Nevertheless, we conjecture a more
differentiated translation of types and terms could restore completeness.

These observations notwithstanding, we can prove that the CBV subtyping rules of
Figure~\ref{fig:call-by-value} are sound and complete with respect to the
subtyping rules for CBPV under Levy's translation.  The proofs can be
found in Appendix~\ref{app:cbv}.

\begin{theorem}[Soundness of Polarized Subtyping, Call-by-Value]\leavevmode
  \label{thm:cbv-soundness}
  \begin{enumerate}
  \item\label{item:cbv-empty-sound} If $\cbv{t} \emp$, then $t \emp$.
  \item\label{item:cbv-leq-sound} If $\cbv{t} \leq \cbv{u}$, then $t \leq u$.
  \end{enumerate}
\end{theorem}

\begin{theorem}[Completeness of Polarized Subtyping, Call-by-Value]\leavevmode
  \label{thm:cbv-completeness}
  \begin{enumerate}
  \item\label{item:cbv-empty} If $t \emp$, then $\cbv{t} \emp$.
  \item\label{item:cbv-leq} If $t \leq u$, then $\cbv{t} \leq \cbv{u}$.
  \end{enumerate}
\end{theorem}

\section{Related Work and Discussion}
\label{sec:related}

We now dive deeper into research related to
our underlying theme on how polarization affects the interaction and definition
of subtyping with recursive types across varying interpretations.

\paragraph*{Subtyping Recursive Types.}

The groundwork for coinductive interpretations of subtyping equirecursive types has been
laid by Amadio and Cardelli~\cite{Amadio93toplas}, subsequently refined by others~\cite{Brandt98fi,Gapeyev00icfp}.
Danielsson and Altenkirch~\cite{Danielsson10mpc} also provided significant
inspiration since they formally clarify that subtyping recursive types relies on a mixed
induction/coinduction. In using an
equirecursive presentation within different calculi, our work has been influenced by its
predominant use in session types~\cite{Chen14ppdp,Das21esop,Gay10jfp}
and, in particular, Gay and Hole's coinductive subtyping algorithm~\cite{Gay05acta}, which we
take as a template for call-by-name typing.

Another important influence has been the work on \emph{refinement
types}~\cite{Davies05phd,Freeman91} which are also recursive but exist within predefined
universes of generative types.  As such, subtyping relations are \emph{simpler} in their
interactions, but face many of the same issues such as emptiness checking.  One can see
this paper as an attempt to free refinement types from some of its restrictions
while retaining some of its good properties.  The key ingredients are (1) explicitly
separating values from computations via polarization, (2) the introduction of variant and
lazy records and their width and depth subtyping rules (owing much to~\cite{Reynolds96}),
and (3) simple bidirectional typechecking.  What is still missing is the use of
\emph{intersections} and \emph{unions} that allow subtyping to propagate more richly to
higher-order types~\cite{Dunfield03fossacs}.

Our treatment of empty\textemdash \emph{value-uninhabited}\textemdash and full
types in Section~\ref{sec:empty-full}, as well
as our call-by-value interpretation in Section~\ref{sec:call-by-value} builds
on Ligatti et al.'s work~\cite{Ligatti17toplas} on precise subtyping with
isorecursive types.

Our direct interpretation of isorecursive types and translation into an equirecursive
setting furthers numerous works either comparing or relating both
formulations~\cite{Pierce02book,Urzyczyn95mfcs,Vanderwaart03tldi}. In
particular, Abadi and Fiore~\cite{Abadi96lics} and more recently Patrigniani et al.~\cite{Patrignani21popl} prove
that terms in one equirecursive setting can be typed in the other (and vice versa)
with varying approaches. The former treats type equality inductively and is
focused on syntactic considerations. The latter treats type equality
coinductively and analyzes types semantically.
Neither of these handle subtyping or mixed coinductive/inductive types like in our study.

Finally, Zhou et al.~\cite{Zhou20oopsla} serves as a helpful overview paper on subtyping
recursive types at large and discusses how Ligatti et al.'s
complete set of rules requires very specific environments for subtyping, as well
as non-standard subtyping rules. This observation demonstrates why our
semantic typing/subtyping approach can offer a more flexible abstraction for
reasoning about expressive type systems while maintaining type safety.

\paragraph*{Semantic Typing and Subtyping.}

Semantic typing goes back to Milner's \emph{semantic soundness theorem}~\cite{Milner78}, which
defined a well-typed program being semantically free of a type violation.
Whereas syntactic typing specifies
a fixed set of syntactic rules that safe terms can be constructed from, semantic typing
here combines two requirements: positive types circumscribe observable values, \emph{exposing
their structure}, and computations of negative types are only required to \emph{behave}
in a safe way. As we demonstrate throughout section~\ref{sec:type-sound}, we can prove our
semantic definitions compatible with our syntactic type rules, leaving syntactic type
soundness to fall out easily (Theorem~\ref{thm:type-sound}).

Milner's initial model didn't scale well to richer types,
like recursive types. With a lens toward more expressive
systems, \emph{step indexing} has become a prominent approach~\cite{Ahmed04phd,Ahmed06esop,Appel01toplas,Dreyer09lics},
which we use to observe that a computation in our model \emph{steps} according
to our dynamics.

As with syntactic/semantic typing, syntactic subtyping is the more typical approach in
modeling subtyping relations over its semantic counterpart. Nonetheless, in what's
operated almost parallel to the research on semantic types, research on semantic
subtyping has also made strides~\cite{Frisch02lics, Castagna05ppdp, Petrucciani18types}.
Mainly, these exploit semantic subtyping for developing type systems based on
set-theoretic subtyping relations and properties, particularly in the context of
handling richer types, including polymorphic functions~\cite{Castagna14popl, Castagna15popl, Petrucciani19phd}
and variants~\cite{Castagna16icfp}, recursive types (interpreted coinductively),
and union, intersection, and negation connectives~\cite{Frisch08acm}.
A major theme in this line of work is excising ``circularity''~\cite{Castagna05ppdp, Frisch08acm}
by means of an involved bootstraping technique,
as issues arise when the denotation of a type is defined simply as the set of
values having that type.

We depart from this line of research in the treatment of functions (defined
computationally rather than set-theoretically), recursive types
(equirecursive setting; inductive for the positive layer and coinductive for the
negative layer), both variant and lazy record types, and the commitment to
explicit polarization (including our incorporation of emptiness/fullness). The
latter of which eliminates circularity and ties together multiple threads
defined in this study.

With this combination of semantic typing and subtyping, our work provides a metatheory for
a more interesting set of typed expressions while also providing a stronger and more
flexible basis for type soundness~\cite{Dreyer19sigplan}, as semantic typing can reason
about syntactically ill-typed expressions as long as those expressions are semantically
well-typed. This combination scales well to our polarized, mixed
setting and focus on subtyping in the presence of recursive types.

\paragraph*{Polarized Type Theory and Call-by-Push-Value.}

At the core of this work has been the
call-by-push-value~\cite{Levy01phd,Levy06hosc} (CBPV) calculus with its notions of
values, computations, and the shifts between them.  Beyond Levy's work, this subsuming
paradigm has formed the foundation of much recent research, ranging from probabilistic
domains~\cite{Ehrhard19lmcs} to those reasoning about effects~\cite{McDermott19esop} and
dependent types~\cite{Pedrot20popl}.  New et al.'s~\cite{New19popl} gradual typing extension to
the calculus shares similarities with our use of step indexing, but its relations (binary
rather than unary), dynamics, and step-counting are treated differently, and its goals are
very different as well, including no coverage on subtyping.

To our knowledge, there are no direct treatments of subtyping recursive types in a CBPV system or
applying a full semantic typing approach in this context with subtyping.  It is, as we've
shown, a fruitful setting for our investigation since the explicit polarization of the
language mirrors the mixed reasoning required to analyze the subtyping.

Though CBPV and polarized type theory typically go hand-in-hand, there are investigations
that look at polarization (\emph{focusing}) and algebraic typing and subtyping from alternate
perspectives.  Steffen~\cite{Steffen99phd} predates Levy's research and presents polarity
as a kinding system for exploiting monotone and antimonotone operators in subtyping
function application.  Abel~\cite{Abel06phd} built upon this and extended it with sized
types. The inherent connection between types and evaluation strategy has also been studied
in the setting of program synthesis~\cite{Rioux20icfp} and
proof theory~\cite{MunchMaccagnoni13phd}, but these do not share our specific semantic concerns.

Polarization as an organizing principle for subtyping is present in Zeilberger's
thesis~\cite{Zeilberger09phd}, but addresses a problem that is
fundamentally different in multiple ways, e.g.\ using ``classical'' types and
continuations, and no width and depth subtyping. The biggest
difference, however, is that its setting considers refinement types, while we
do not have a refinement relation and show that some of the
advantages of refinement types can be achieved without the additional layer.

Two studies on a global approach to algebraic subtyping~\cite{Dolan17phd,Parreaux20icfp}
define subtyping relationships with generative datatype constructors while discussing
polarity (here with a different meaning) and discarding semantic interpretations.
However, the generative nature of datatype constructors in this work makes its quite
different from ours.

\paragraph*{Mixed Coinductive/Inductive Reasoning for Recursive Types.}

The natural separation of positive and negative layers in CBPV led us through
the literature on mixed coinductive/inductive definitions for recursive types.
Related to our work in this paper, Danielsson and Altenkirch~\cite{Danielsson10mpc} and
Jones and Pearce~\cite{Jones16ftfjp}
provide definitions for equirecursive subtyping relations in a mixed setting
while using a suspension monad for non-terminating computations, which
shares an affinity with force/return CBPV computations.  Danielsson and Altenkirch,
however, do not try to justify the structural typing rules themselves via
semantic typing of values or expressions\textemdash only the subtyping rules.
Jones and Pearce are closer to our approach since they also use
a semantic interpretation of types for expressions. While not polarized, they do
consider inductive/coinductive types separately, but do not lift them
to cover function types, instead studying other constructs such as unions.

Komendantsky~\cite{Komendantsky11tfp} manages infinitary subtyping (for only function and
recursive types) via a semantic encoding by folding an inductive relation into a
coinductive one.  We work in the opposite direction, turning the coinductive
portion into an inductive one by step indexing.  Lepigre and
Raffali~\cite{Lepigre19toplas} mix induction and coinduction in a syntax-directed
framework, focusing on circular proof derivations and sized types~\cite{Abel16jfp}; also
managing inductive types coinductively. Cohen and Rowe~\cite{Cohen20ijcar} provide a
proposal for circular reasoning in a mixed setting, but the focus is on a
transitive closure logic built around least and greatest fixed point operators. It seems
quite plausible that we could use such systems to formalize our investigation, although we
found some merit in using step-indexing and Brotherston and Simpson's circular proof
system for induction~\cite{Brotherston11jlc}.

\section{Conclusion}\label{sec:conclusion}

We introduced a rich system of subtyping for an equirecusive variant of call-by-push-value
and proved its soundness via semantic means.  We also provided a bidirectional type
checking algorithm and illustrated its expressiveness through several different kinds of
examples.  We showed the fundamental nature of the results by deriving systems of
subtyping for isorecursive types and languages with call-by-name and call-by-value
dynamics. The limitations of the present systems lie primarily in the lack of intersection
and union types and parametric polymorphism which are the subject of ongoing work.

\paragraph{Acknowledgements.}
\begin{sloppypar}
  We wish to express our gratitude to the anonymous reviewers of this paper for their comments.
  Support for this research was provided by the NSF under Grant No. 1718276 and
  by FCT through
  the CMU Portugal Program,
  the LASIGE Research Unit (UIDB/00408/2020 and UIDP/00408/2020), and
  the project SafeSessions (PTDC/CCI-COM/6453/2020).
\end{sloppypar}
\newpage

\bibliographystyle{splncs04}
\bibliography{esop}

\begin{thebibliography}{10}
\providecommand{\url}[1]{\texttt{#1}}
\providecommand{\urlprefix}{URL }
\providecommand{\doi}[1]{https://doi.org/#1}

\bibitem{Abadi96lics}
Abadi, M., Fiore, M.P.: Syntactic considerations on recursive types. In:
  Proceedings of the 11th Annual IEEE Symposium on Logic in Computer Science.
  pp. 242--252. IEEE Computer Society (1996),
  \url{https://doi.org/10.1109/LICS.1996.561324}

\bibitem{Abel06phd}
Abel, A.: Polarized subtyping for sized types. In: Computer Science - Theory
  and Applications, First International Computer Science Symposium in Russia,
  {CSR} 2006, St. Petersburg, Russia, June 8-12, 2006, Proceedings. Lecture
  Notes in Computer Science, vol.~3967, pp. 381--392. Springer (2006).
  \doi{10.1007/11753728\_39}

\bibitem{Abel07aplas}
Abel, A.: Mixed inductive/coinductive types and strong normalization. In:
  Programming Languages and Systems, 5th Asian Symposium, {APLAS} 2007,
  Singapore, November 29-December 1, 2007, Proceedings. Lecture Notes in
  Computer Science, vol.~4807, pp. 286--301. Springer (2007).
  \doi{10.1007/978-3-540-76637-7\_19}

\bibitem{Abel12fics}
Abel, A.: Type-based termination, inflationary fixed-points, and mixed
  inductive-coinductive types. In: Miller, D., {\'E}sik, Z. (eds.) Proceedings
  of the 8th Workshop on Fixed Points in Computer Science. pp. 1--11. FICS
  2012, Electronic Proceedings in Theoretical Computer Science 77 (2012).
  \doi{10.4204/EPTCS.77.1}

\bibitem{Abel13icfp}
Abel, A., Pientka, B.: Wellfounded recursion with copatterns: A unified
  approach to termination and productivity. In: Morrisett, G., Uustalu, T.
  (eds.) International Conference on Functional Programming (ICFP'13). pp.
  185--196. ACM, Boston, Massachusetts (Sep 2013),
  \url{https://doi.org/10.1145/2500365.2500591}

\bibitem{Abel16jfp}
Abel, A., Pientka, B.: Well-founded recursion with copatterns and sized types.
  Journal of Functional Programming  \textbf{26}, ~e2 (2016),
  \url{https://doi.org/10.1017/S0956796816000022}

\bibitem{Ahmed04phd}
Ahmed, A.J.: Semantics of Types for Mutable State. Ph.D. thesis, Princeton
  University (2004), \url{http://www.ccs.neu.edu/home/amal/ahmedsthesis.pdf},
  aAI3136691

\bibitem{Ahmed06esop}
Ahmed, A.J.: Step-indexed syntactic logical relations for recursive and
  quantified types. In: Sestoft, P. (ed.) 15th European Symposium on
  Programming (ESOP 2006). pp. 69--83. Springer LNCS 3924, Vienna, Austria (Mar
  2006). \doi{10.1007/11693024\_6}

\bibitem{Amadio93toplas}
Amadio, R.M., Cardelli, L.: Subtyping recursive types. ACM Transactions on
  Programming Languages and Systems  \textbf{15}(4),  575--631 (1993),
  \url{https://doi.org/10.1145/155183.155231}

\bibitem{Appel01toplas}
Appel, A.W., McAllester, D.A.: An indexed model of recursive types for
  foundational proof-carrying code. Transactions on Programming Languages and
  Systems  \textbf{23}(5),  657--683 (2001),
  \url{https://doi.org/10.1145/504709.504712}

\bibitem{Barwise89csli}
Barwise, J.: The situation in logic, {CSLI} lecture notes series, vol.~17. CSLI
  (1989)

\bibitem{Berardi18cmcs}
Berardi, S., Tatsuta, M.: Intuitionistic {Podelski-Rybalchenko} theorem and
  equivalence between inductive definitions and cyclic proofs. In:
  C{\"\i}rstea, C. (ed.) Workshop on Coalgebraic Methods in Computer Science
  (CMCS 2018). pp. 13--33. Springer LNCS 11202, Thessaloniki, Greece (Apr
  2018), \url{https://doi.org/10.1007/978-3-030-00389-0\_3}

\bibitem{Brandt98fi}
Brandt, M., Henglein, F.: Coinductive axiomatization of recursive type equality
  and subtyping. Fundamenta Informaticae  \textbf{33}(4),  309--338 (1998),
  \url{https://doi.org/10.3233/FI-1998-33401}

\bibitem{Brotherston11jlc}
Brotherston, J., Simpson, A.: Sequent calculi for induction and infinite
  descent. Journal of Logic and Computation  \textbf{21}(6),  1177--1216
  (2011), \url{https://doi.org/10.1093/logcom/exq052}

\bibitem{Castagna05ppdp}
Castagna, G., Frisch, A.: A gentle introduction to semantic subtyping. In:
  Proceedings of the 7th International {ACM} {SIGPLAN} Conference on Principles
  and Practice of Declarative Programming, July 11-13 2005, Lisbon, Portugal.
  pp. 198--199. ACM (2005), \url{https://doi.org/10.1145/1069774.1069793}

\bibitem{Castagna15popl}
Castagna, G., Nguyen, K., Xu, Z., Abate, P.: Polymorphic functions with
  set-theoretic types: Part 2: Local type inference and type reconstruction.
  In: Proceedings of the 42nd Annual ACM SIGPLAN-SIGACT Symposium on Principles
  of Programming Languages. p. 289–302. POPL '15, Association for Computing
  Machinery, New York, NY, USA (2015). \doi{10.1145/2676726.2676991}

\bibitem{Castagna14popl}
Castagna, G., Nguyen, K., Xu, Z., Im, H., Lenglet, S., Padovani, L.:
  Polymorphic functions with set-theoretic types: part 1: syntax, semantics,
  and evaluation. In: Proceedings of the 41st ACM SIGPLAN-SIGACT Symposium on
  Principles of Programming Languages. p. 5–17. POPL '14 (2014).
  \doi{10.1145/2535838.2535840}

\bibitem{Castagna16icfp}
Castagna, G., Petrucciani, T., Nguyen, K.: Set-theoretic types for polymorphic
  variants. Proceedings of the 21st ACM SIGPLAN International Conference on
  Functional Programming  (2016), \url{https://doi.org/10.1145/3022670.2951928}

\bibitem{Chen14ppdp}
Chen, T.C., Dezani-Ciancaglini, M., Yoshida, N.: On the preciseness of
  subtyping in session types. In: Proceedings of the Conference on Principles
  and Practice of Declarative Programming (PPDP'14). ACM, Canterbury, UK (Sep
  2014), \url{https://doi.org/10.1145/2643135.2643138}

\bibitem{Cockett01cmcs}
Cockett, J.R.B.: Deforestation, program transformation, and cut-elimination.
  In: Coalgebraic Methods in Computer Science, {CMCS} 2001, a Satellite Event
  of {ETAPS} 2001, Genova, Italy, April 6-7, 2001. Electronic Notes in
  Theoretical Computer Science, vol.~44, pp. 88--127. Elsevier (2001),
  \url{https://doi.org/10.1016/S1571-0661(04)80904-6}

\bibitem{Cohen20ijcar}
Cohen, L., Rowe, R.N.S.: Integrating induction and coinduction via closure
  operators and proof cycles. In: 10th International Joint Conference on
  Automated Reasoning (IJCAR 2020). pp. 375--394. Springer LNCS 12166, Paris,
  France (Jul 2020), \url{https://doi.org/10.1007/978-3-030-51074-9\_21}

\bibitem{Danielsson10mpc}
Danielsson, N.A., Altenkirch, T.: Subtyping, declaratively. In: 10th
  International Conference on Mathematics of Program Construction (MPC 2010).
  pp. 100--118. Springer LNCS 6120, Qu{\'e}bec City, Canada (Jun 2010),
  \url{https://doi.org/10.1007/978-3-642-13321-3\_8}

\bibitem{Das21esop}
Das, A., DeYoung, H., Mordido, A., Pfenning, F.: Nested session types. In:
  Yoshida, N. (ed.) 30th European Symposium on Programming. pp. 178--206.
  Springer LNCS, Luxembourg, Luxembourg (Mar 2021),
  \url{http://www.cs.cmu.edu/~fp/papers/esop21.pdf}, extended version available
  as arXiv:2010.06482

\bibitem{Davies05phd}
Davies, R.: Practical Refinement-Types Checking. Ph.D. thesis, Carnegie Mellon
  University (May 2005), \url{https://www.cs.cmu.edu/~rwh/students/davies.pdf},
  available as Technical Report CMU-CS-05-110

\bibitem{Davies00icfp}
Davies, R., Pfenning, F.: Intersection types and computational effects. In:
  Wadler, P. (ed.) Proceedings of the Fifth International Conference on
  Functional Programming (ICFP'00). pp. 198--208. ACM Press, Montreal, Canada
  (Sep 2000), \url{https://doi.org/10.1145/351240.351259}

\bibitem{Dolan17phd}
Dolan, S.: Algebraic Subtyping: Distinguished Dissertation 2017. BCS, Swindon,
  GBR (2017),
  \url{https://www.cs.tufts.edu/~nr/cs257/archive/stephen-dolan/thesis.pdf}

\bibitem{Dreyer09lics}
Dreyer, D., Ahmed, A., Birkedal, L.: Logical step-indexed logical relations.
  In: Proceedings of the 24th Annual {IEEE} Symposium on Logic in Computer
  Science, {LICS} 2009, 11-14 August 2009, Los Angeles, CA, {USA}. pp. 71--80.
  {IEEE} Computer Society (2009), \url{https://doi.org/10.1109/LICS.2009.34}

\bibitem{Dreyer19sigplan}
Dreyer, D., Timany, A., Krebbers, R., Birkedal, L., Jung, R.: What type
  soundness theorem do you really want to prove? (Oct 2019),
  \url{https://blog.sigplan.org/2019/10/17/what-type-soundness-theorem-do-you-really-want-to-prove}

\bibitem{Dunfield19corr}
Dunfield, J., Krishnaswami, N.: Bidirectional typing. CoRR
  \textbf{abs/1908.05839} (2019), \url{http://arxiv.org/abs/1908.05839}

\bibitem{Dunfield19popl}
Dunfield, J., Krishnaswami, N.R.: Sound and complete bidirectional typechecking
  for higher-rank polymorphism with existentials and indexed types. Proc. {ACM}
  Program. Lang.  \textbf{3}({POPL}),  9:1--9:28 (2019). \doi{10.1145/3290322}

\bibitem{Dunfield03fossacs}
Dunfield, J., Pfenning, F.: Type assignment for intersections and unions in
  call-by-value languages. In: Gordon, A. (ed.) Proceedings of the 6th
  International Conference on Foundations of Software Science and Computation
  Structures (FOSSACS'03). pp. 250--266. Springer-Verlag LNCS 2620, Warsaw,
  Poland (Apr 2003), \url{https://doi.org/10.1007/3-540-36576-1\_16}

\bibitem{Dunfield04popl}
Dunfield, J., Pfenning, F.: Tridirectional typechecking. In: X.Leroy (ed.)
  Conference Record of the 31st Annual Symposium on Principles of Programming
  Languages (POPL'04). pp. 281--292. ACM Press, Venice, Italy (Jan 2004),
  \url{https://doi.org/10.1145/964001.964025}, extended version available as
  Technical Report CMU-CS-04-117, March 2004

\bibitem{Ehrhard19lmcs}
Ehrhard, T., Tasson, C.: Probabilistic call by push value. Log. Methods Comput.
  Sci.  \textbf{15}(1) (2019), \url{https://doi.org/10.23638/LMCS-15(1:3)2019}

\bibitem{Freeman91}
Freeman, T., Pfenning, F.: Refinement types for {ML}. In: Proceedings of the
  {SIGPLAN '91} Symposium on Language Design and Implementation. pp. 268--277.
  ACM Press, Toronto, Ontario (Jun 1991),
  \url{https://doi.org/10.1145/113445.113468}

\bibitem{Frisch02lics}
Frisch, A., Castagna, G., Benzaken, V.: Semantic subtyping. In: 17th {IEEE}
  Symposium on Logic in Computer Science {(LICS} 2002), 22-25 July 2002,
  Copenhagen, Denmark, Proceedings. pp. 137--146. {IEEE} Computer Society
  (2002), \url{https://doi.org/10.1109/LICS.2002.1029823}

\bibitem{Frisch08acm}
Frisch, A., Castagna, G., Benzaken, V.: Semantic subtyping: Dealing
  set-theoretically with function, union, intersection, and negation types. J.
  ACM  \textbf{55},  19:1--19:64 (2008),
  \url{https://dl.acm.org/doi/10.1145/1391289.1391293}

\bibitem{Gapeyev00icfp}
Gapeyev, V., Levin, M.Y., Pierce, B.C.: Recursive subtyping revealed:
  functional pearl. In: Proceedings of the Fifth {ACM} {SIGPLAN} International
  Conference on Functional Programming {(ICFP} '00), Montreal, Canada,
  September 18-21, 2000. pp. 221--231. ACM (2000),
  \url{https://doi.org/10.1145/351240.351261}

\bibitem{Garcia20wgt}
Garcia, R., Tanter, {\'E}.: Gradual typing as if types mattered. In: Informal
  Proceedings of the ACM SIGPLAN Workshop on Gradual Typing (WGT20) (2020),
  \url{https://wgt20.irif.fr/wgt20-final28-acmpaginated.pdf}

\bibitem{Gay05acta}
Gay, S.J., Hole, M.: Subtyping for session types in the {$\pi$}-calculus. Acta
  Informatica  \textbf{42}(2--3),  191--225 (2005),
  \url{https://doi.org/10.1007/s00236-005-0177-z}

\bibitem{Gay10jfp}
Gay, S.J., Vasconcelos, V.T.: Linear type theory for asynchronous session
  types. Journal of Functional Programming  \textbf{20}(1),  19--50 (Jan 2010),
  \url{https://doi.org/10.1017/S0956796809990268}

\bibitem{Gradel03lics}
Gr{\"a}del, E., Kreutzer, S.: Will deflation lead to depletion? {O}n
  non-monotone fixed point inductions. In: Symposium on Logic in Computer
  Science (LICS 2003). pp. 158--167. IEEE Computer Society, Ottawa, Canada (Jun
  2003), \url{https://doi.org/10.1109/LICS.2003.1210055}

\bibitem{Harper16book}
Harper, R.: Practical Foundations for Programming Languages. Cambridge
  University Press, second edn. (Apr 2016)

\bibitem{Hermida98ic}
Hermida, C., Jacobs, B.: Structural induction and coinduction in a fibrational
  setting. Inf. Comput.  \textbf{145}(2),  107--152 (1998),
  \url{https://doi.org/10.1006/inco.1998.2725}

\bibitem{Hinrichsen21cpp}
Hinrichsen, J.K., Louwrink, D., Krebbers, R., Bengtson, J.: Machine-checked
  semantic session typing. In: {CPP} '21: 10th {ACM} {SIGPLAN} International
  Conference on Certified Programs and Proofs, Virtual Event, Denmark, January
  17-19, 2021. pp. 178--198. ACM (2021). \doi{10.1145/3437992.3439914}

\bibitem{Jafery17popl}
Jafery, K.A., Dunfield, J.: Sums of uncertainty: refinements go gradual. In:
  Proceedings of the 44th {ACM} {SIGPLAN} Symposium on Principles of
  Programming Languages, {POPL} 2017, Paris, France, January 18-20, 2017. pp.
  804--817. {ACM} (2017). \doi{10.1145/3009837.3009865}

\bibitem{Jones16ftfjp}
Jones, T., Pearce, D.J.: A mechanical soundness proof for subtyping over
  recursive types. In: Proceedings of the 18th Workshop on Formal Techniques
  for Java-like Programs, FTfJP@ECOOP 2016, Rome, Italy, July 17-22, 2016.
  p.~1. ACM (2016). \doi{10.1145/2955811.2955812}

\bibitem{Jung18popl}
Jung, R., Jourdan, J., Krebbers, R., Dreyer, D.: Rustbelt: securing the
  foundations of the rust programming language. Proc. {ACM} Program. Lang.
  \textbf{2}(POPL),  66:1--66:34 (2018). \doi{10.1145/3158154}

\bibitem{Komendantsky11tfp}
Komendantsky, V.: Subtyping by folding an inductive relation into a coinductive
  one. In: Trends in Functional Programming, 12th International Symposium,
  {TFP} 2011, Madrid, Spain, May 16-18, 2011, Revised Selected Papers. Lecture
  Notes in Computer Science, vol.~7193, pp. 17--32. Springer (2011),
  \url{https://doi.org/10.1007/978-3-642-32037-8\_2}

\bibitem{Lepigre16corr}
Lepigre, R., Raffalli, C.: Subtyping-based type-checking for system {F} with
  induction and coinduction. CoRR  \textbf{abs/1604.01990} (2016),
  \url{http://arxiv.org/abs/1604.01990}

\bibitem{Lepigre19toplas}
Lepigre, R., Raffalli, C.: Practical subtyping for {C}urry-style languages. ACM
  Transactions on Programming Languages and Systems (TOPLAS)  \textbf{41},  1
  -- 58 (2019), \url{https://doi.org/10.1145/3285955}

\bibitem{Levy01phd}
Levy, P.B.: Call-by-Push-Value. Ph.D. thesis, University of London (2001),
  \url{http://www.cs.bham.ac.uk/~pbl/papers/thesisqmwphd.pdf}

\bibitem{Levy06hosc}
Levy, P.B.: Call-by-push-value: Decomposing call-by-value and call-by-name.
  Higher-Order and Symbolic Computation  \textbf{19}(4),  377--414 (2006),
  \url{https://doi.org/10.1007/s10990-006-0480-6}

\bibitem{Ligatti17toplas}
Ligatti, J., Blackburn, J., Nachtigal, M.: On subtyping-relation completeness,
  with an application to iso-recursive types. ACM Transactions on Programming
  Languages and Systems  \textbf{39}(4),  4:1--4:36 (Mar 2017),
  \url{https://doi.org/10.1145/2994596}

\bibitem{McDermott19esop}
McDermott, D., Mycroft, A.: Extended call-by-push-value: Reasoning about
  effectful programs and evaluation order. In: Programming Languages and
  Systems - 28th European Symposium on Programming, {ESOP} 2019, Held as Part
  of the European Joint Conferences on Theory and Practice of Software, {ETAPS}
  2019, Prague, Czech Republic, April 6-11, 2019, Proceedings. Lecture Notes in
  Computer Science, vol. 11423, pp. 235--262. Springer (2019),
  \url{https://doi.org/10.1007/978-3-030-17184-1\_9}

\bibitem{Milner78}
Milner, R.: A theory of type polymorphism in programming. Journal of Computer
  and System Sciences  \textbf{17},  348--375 (Aug 1978),
  \url{https://doi.org/10.1016/0022-0000(78)90014-4}

\bibitem{MunchMaccagnoni13phd}
Munch{-}Maccagnoni, G.: Syntax and Models of a non-Associative Composition of
  Programs and Proofs. (Syntaxe et mod{\`{e}}les d'une composition
  non-associative des programmes et des preuves). Ph.D. thesis, Paris Diderot
  University, France (2013),
  \url{https://tel.archives-ouvertes.fr/tel-00918642}

\bibitem{Nakata10sos}
Nakata, K., Uustalu, T.: Resumptions, weak bisimilarity and big-step semantics
  for while with interactive {I/O:} an exercise in mixed induction-coinduction.
  In: Proceedings Seventh Workshop on Structural Operational Semantics, {SOS}
  2010, Paris, France, 30 August 2010. EPTCS, vol.~32, pp. 57--75 (2010),
  \url{https://doi.org/10.4204/EPTCS.32.5}

\bibitem{New19popl}
New, M.S., Licata, D.R., Ahmed, A.: Gradual type theory. Proc. {ACM} Program.
  Lang.  \textbf{3}(POPL),  15:1--15:31 (2019),
  \url{https://doi.org/10.1145/3290328}

\bibitem{Park79ass}
Park, D.M.R.: On the semantics of fair parallelism. In: Bj{\o}rner, D. (ed.)
  Abstract Software Specifications, 1979 Copenhagen Winter School, January 22 -
  February 2, 1979, Proceedings. Lecture Notes in Computer Science, vol.~86,
  pp. 504--526. Springer (1979),
  \url{https://doi.org/10.1007/3-540-10007-5\_47}

\bibitem{Parreaux20icfp}
Parreaux, L.: The simple essence of algebraic subtyping: principal type
  inference with subtyping made easy (functional pearl). Proc. {ACM} Program.
  Lang.  \textbf{4}(ICFP),  124:1--124:28 (2020),
  \url{https://doi.org/10.1145/3409006}

\bibitem{Patrignani21popl}
Patrignani, M., Martin, E.M., Devriese, D.: On the semantic expressiveness of
  recursive types. Proceedings of the ACM on Programming Languages  \textbf{5},
   1--29 (2021), \url{https://doi.org/10.1145/3434302}

\bibitem{Pedrot20popl}
P{\'{e}}drot, P., Tabareau, N.: The fire triangle: how to mix substitution,
  dependent elimination, and effects. Proc. {ACM} Program. Lang.
  \textbf{4}(POPL),  58:1--58:28 (2020), \url{https://doi.org/10.1145/3371126}

\bibitem{Petrucciani19phd}
Petrucciani, T.: Polymorphic set-theoretic types for functional languages.
  (Types ensemblistes polymorphes pour les langages fonctionnels). Ph.D.
  thesis, Sorbonne Paris Cit{\'{e}}, France (2019),
  \url{https://tel.archives-ouvertes.fr/tel-02119930}

\bibitem{Petrucciani18types}
Petrucciani, T., Castagna, G., Ancona, D., Zucca, E.: Semantic subtyping for
  non-strict languages. In: 24th International Conference on Types for Proofs
  and Programs, {TYPES} 2018, June 18-21, 2018, Braga, Portugal. LIPIcs,
  vol.~130, pp. 4:1--4:24. Schloss Dagstuhl - Leibniz-Zentrum f{\"{u}}r
  Informatik (2018). \doi{10.4230/LIPIcs.TYPES.2018.4},
  \url{https://arxiv.org/abs/1810.05555}

\bibitem{Pierce02book}
Pierce, B.: Types and Programming Languages. MIT Press (2002)

\bibitem{Pierce98popl}
Pierce, B.C., Turner, D.N.: Local type inference. In: Conference Record of the
  25th Symposium on Principles of Programming Languages (POPL'98) (1998),
  \url{https://doi.org/10.1145/268946.268967}, full version in \bgroup\em ACM
  Transactions on Programming Languages and Systems (TOPLAS)\egroup, 22(1),
  January 2000, pp.~1--44

\bibitem{Raffa94phd}
Raffalli, C.: L'arithmetique fonctionnelle du second ordre avec points fixes.
  Ph.D. thesis, Paris 7 (1994), \url{http://www.theses.fr/1994PA077080}, thèse
  de doctorat dirigée par Krivine, Jean-Louis Mathématiques. Logique et
  fondements de l'informatique Paris 7 1994

\bibitem{Reynolds96}
Reynolds, J.C.: Design of the programming language {Forsythe}. Tech. Rep.
  CMU-CS-96-146, Carnegie Mellon University (Jun 1996)

\bibitem{Rioux20icfp}
Rioux, N., Zdancewic, S.: Computation focusing. Proc. {ACM} Program. Lang.
  \textbf{4}({ICFP}),  95:1--95:27 (2020). \doi{10.1145/3408977}

\bibitem{Steffen99phd}
Steffen, M.: Polarized higher-order subtyping. Ph.D. thesis, University of
  Erlangen-Nuremberg, Germany (1999), \url{http://d-nb.info/958020493}

\bibitem{Urzyczyn95mfcs}
Urzyczyn, P.: Positive recursive type assignment. In: Mathematical Foundations
  of Computer Science 1995. pp. 382--391. Springer Berlin Heidelberg, Berlin,
  Heidelberg (1995), \url{https://doi.org/10.1007/3-540-60246-1\_144}

\bibitem{Vanderwaart03tldi}
Vanderwaart, J., Dreyer, D., Petersen, L., Crary, K., Harper, R., Cheng, P.:
  Typed compilation of recursive datatypes. In: Proceedings of TLDI'03: 2003
  {ACM} {SIGPLAN} International Workshop on Types in Languages Design and
  Implementation, New Orleans, Louisiana, USA, January 18, 2003. pp. 98--108.
  ACM (2003), \url{https://doi.org/10.1145/604174.604187}

\bibitem{Zeilberger09phd}
Zeilberger, N.: The Logical Basis of Evaluation Order and Pattern-Matching.
  Ph.D. thesis, Carnegie Mellon University, USA (2009),
  \url{http://noamz.org/thesis.pdf}

\bibitem{Zhou20oopsla}
Zhou, Y., d.~S.~Oliveira, B.C., Zhao, J.: Revisiting iso-recursive subtyping.
  Proc. {ACM} Program. Lang.  \textbf{4}(OOPSLA),  223:1--223:28 (2020),
  \url{https://doi.org/10.1145/3428291}

\end{thebibliography}

\newpage
\appendix\label{app}

\section{Type Examples}
\label{app:type-examples}
\[
  \begin{array}{rcll}
    \p{\ms{bool}} & = & \plusn{\mb{false} : \one, \mb{true} : \one} & \mbox{Booleans} \\[1ex]
    \p{\ms{nat}} & = & \plusn{\mb{z} : \one, \mb{s} : \ms{nat}} & \mbox{Unary numbers} \\[1ex]
    \p{\ms{bin}} & = & \plusn{\mb{e} : \one, \mb{b0} : \ms{bin}, \mb{b1} : \ms{bin}}
    & \mbox{Binary numbers (least significant bit }\\
    &&&\mbox{first)} \\
    \p{\ms{std}} & = & \plusn{\mb{e} : \one, \mb{b0} : \ms{pos}, \mb{b1} : \ms{std}}
    & \mbox{Standard binary numbers (no trailing}\\
    &&&\mbox{$\mb{b0}$s)} \\
    \p{\ms{pos}} & = & \plusn{\phantom{\mb{e} : \one,}\, \mb{b0} : \ms{pos}, \mb{b1} : \ms{std}}
    & \mbox{Positive standard binary numbers} \\[1em]
    \p{\ms{list}} & = & \plusn{\mb{nil} : \one, \mb{cons} : \ms{std} \tensor \ms{list}}
    & \mbox{Lists of standard binary numbers} \\[1ex]
    \n{\ms{stream}} & = & \withn{\mb{hd} : \up \ms{std}, \mb{tl} : \ms{stream}}
    & \mbox{Streams of standard binary numbers} \\[1ex]
    \n{\ms{pstream}} & = & \up (\ms{std} \tensor \ms{padding})
    & \mbox{Streams with finite padding} \\
    \p{\ms{padding}} & = & \oplus\{{
    \begin{gathered}[t]
    	\mb{pad} : \ms{padding}, \\
    	\mb{next} : \dn \ms{pstream} \}
       \end{gathered}}\\[1ex]
    \n{\ms{zstream}} & = & \up (\ms{std} \tensor \plusn{\mb{next} : \dn \ms{zstream}})
    & \mbox{Streams with zero padding} \\[1em]
    \n{\ms{U}} & = & (\dn \ms{U}) \imp \ms{U}
    & \mbox{Embedding untyped $\lambda$-calculus}
  \end{array}
\]

\section{Examples of Semantic Typing}
\label{app:semtypes}

\begin{example}[Identity Function]
  \mbox{}
  $\sof{\lam{x}{\return{x}}}{\p{\tau} \imp \up \p{\tau}}$ \quad \mbox{for all $\p\tau$}
  \begin{tabbing}
    Reason for $k \geq 2$: \\
    \hspace{2em} \= $\sof[k]{\lam{x}{\return{x}}}{\p{\tau} \imp \up \p{\tau}}$ \\
    if \> $\sof*[k]{\lam{x}{\return{x}}}{\p{\tau} \imp \up \p{\tau}}$ \` Since $\lam{x}{\return{x}}$ is terminal \\
    if \> $\sof[k]{\app{(\lam{x}{\return{x}})}{v}}{\up \p{\tau}}$ for all $i < k$, $\sof[i]{v}{\p{\tau}}$
    \` By definition \\
    if \> $\sof[k-1]{\return{v}}{\up \p{\tau}}$ \\
    if \> $\sof[k-2]{v}{\p{\tau}}$ \` Since $k-2 < k$ \\[1em]
       \> $\sof{\lam{x}{\return{x}}}{\p{\tau} \imp \up \p{\tau}}$ \` By downward closure
  \end{tabbing}
\end{example}

\begin{example}[Right Recursion]
  Define:
  $\qquad
    s_0 = \one \imp s_0
    \qquad\qquad
    e_0 = \lam{x}{e_0}
  $\\
  Claim: $\sof{e_0}{s_0}$.
  \begin{tabbing}
    Prove $\sof[k]{e_0}{s_0}$ for all $k$ by induction on $k$. \\
    Reason for $k \geq 2$: \\
    \qquad \= $\sof[k]{e_0}{s_0}$ \\
    if \> $\sof[k-1]{\lam{x}{e_0}}{s_0}$ \\
    if \> $\sof*[k-1]{\lam{x}{e_0}}{\one \imp s_0}$ \\
    if \> $\sof[k-1]{\app{(\lam{x}{e_0})}{v}}{s_0}$ \quad \mbox{for all $i < k-1$, $\sof[i]{v}{\one}$} \\
    if \> $\sof[k-2]{e_0}{s_0}$ \` By ind.\ hyp \\[1em]
    $\sof[k]{e_0}{s_0}$ for all $k$ \` By downward closure \\
    $\sof{e_0}{s_0}$ \` By definition
  \end{tabbing}
\end{example}

\begin{example}[$\Omega$]
  \label{ex:omega-full}
  Define:
  $\qquad
    \omega = \lam{x}{\app{(\force{x})}{x}}
    \qquad\qquad
    \Omega = \app{\omega}{(\thunk{\omega})}
  $\\
  Claim: $\sof{\Omega}{\n{\sigma}}$ for every $\n{\sigma}$
  \begin{tabbing}
    Prove $\sof[k]{\app{\omega}{(\thunk{\omega})}}{\n{\sigma}}$ for every $k$ by induction on $k$ \\
    Reason for $k \geq 3$: \\
    \qquad \= $\sof[k]{\app{\omega}{(\thunk{\omega})}}{\n{\sigma}}$ \\
    if \> $\sof[k-1]{\app{(\lam{x}{\app{(\force{x})}{x}})}{(\thunk{\omega})}}{\n{\sigma}}$ \` By definition \\
    if \> $\sof[k-2]{\app{(\force{(\thunk{\omega})})}{(\thunk{\omega})}}{\n{\sigma}}$ \` By definition \\
    if \> $\sof[k-3]{\app{\omega}{(\thunk{\omega})}}{\n{\sigma}}$ \` By definition \\[1em]
    Holds by ind.\ hyp.\ and then \\
    $\sof[k]{\app{\omega}{(\thunk{\omega})}}{\n{\sigma}}$ \` By downward closure \\
    $\sof[k+1]{\Omega}{\n{\sigma}}$ \` By definition \\
    $\sof{\Omega}{\n{\sigma}}$ \` By downward closure
  \end{tabbing}
\end{example}

\begin{example}[Empty Recursive Type]
  \label{ex:empty}
  Define:
  $\qquad
    t_0 = \one \tensor t_0
  $\\
  Claim: Does not exist a $v$ such that $\sof{v}{t_0}$.

  We prove something stronger: for all $k$ and $v$, it is not the
  case that $\sof[k]{v}{t_0}$.  The proof is by induction on $v$.
  \begin{tabbing}
    Assume $\sof[k]{v}{t_0}$ \\
    $\sof[k]{v}{\one \tensor t_0}$ \` By definition \\
    $v = \pair{v_1,v_2}$ with $\sof[k]{v_1}{\one}$ and $\sof[k]{v_2}{t_0}$ \` By definition \\
    Contradiction \` By ind.\ hyp.\ since $v_2 < v$
  \end{tabbing}

  Continuing the example: Assume $e$ has any type at all (that is $\sof{e}{\n{\rho}}$ for
  some $\n{\rho}$).  Then for all $\n{\sigma}$ we have $\sof{e}{t_0 \imp \n{\sigma}}$.

  \begin{tabbing}
    We prove $\sof[k]{e}{t_0 \imp \n{\sigma}}$ by induction on $k$ \\
    Because $\sof{e}{\n{\rho}}$ we know one of the following cases applies: \\[1ex]
    \textbf{Case:} $k = 0$.  Then $\sof[0]{e}{t_0 \imp \n{\sigma}}$ by definition \\[1ex]
    \textbf{Case:} $k > 0$ and $e \reduces e'$.
    Then $\sof[k-1]{e'}{t_0 \imp \n{\sigma}}$ by ind.\ hyp. \\[1ex]
    \textbf{Case:} $k > 0$ and $e$ is terminal. \\
    By definition, it remains to show that $\sof[k]{\app{e}{v}}{\n{\sigma}}$ for all $i < k$ and $v$
    with $\sof[i]{v}{t_0}$ \\
    But that's vacuously true by the first part of this example.
  \end{tabbing}
\end{example}

\section{Properties of Semantic Typing}

\begin{lemma}[Closure under Expansion~\footnote{Included as a useful property and sanity check.}]
  \label{lm:expansion}
  If $\sof{e'}{\n{\sigma}}$ and $e \reduces e'$ then $\sof{e}{\n{\sigma}}$.
\end{lemma}
\begin{proof} Direct, using the definitions.
  \begin{tabbing}
    $\sof{e'}{\n{\sigma}}$ \` Given \\
    $\sof[k]{e'}{\n{\sigma}}$ for all $k$ \` By definition \\
    $\sof[k+1]{e}{\n{\sigma}}$ for all $k$ \` By definition, since $e \reduces e'$ \\
    $\sof[i]{e}{\n{\sigma}}$ for all $i$ \` By downward closure \\
    $\sof{e}{\n{\sigma}}$ \` By definition
  \end{tabbing}
\end{proof}

\section{Emptiness}
\label{app:emptiness}

\begin{proof} (of Theorem~\ref{thm:empty})
  We interpret the judgment $t \emp$ semantically as $\sof[k]{v}{t} \vdash \cdot$ (which
  expresses $\nsof[k]{v}{t}$ in a sequent), where $t$ is given and $k$ and $v$ are
  parameters and therefore implicitly universally quantified.  The proof of this judgment is
  carried out in a circular metalogic.  We translate each inference rule for $t \emp$
  into a derivation for $\sof[k]{v}{t} \vdash \cdot$, where each unproven subgoal
  corresponds to a premise of the rule.  When the derivation of $t \emp$ is closed
  by a cycle, the corresponding derivation of $\sof[k]{v}{t} \vdash \cdot$ is
  closed by a corresponding cycle in the metalogic.

  During this compositional translation of $t \emp$ we need to ensure that the
  lexicographically ordered pair $(k,v)$ is smaller for each subgoal on the semantic side.
  This ensures that we can build a valid cycle in the metalogic whenever we have a cycle
  in the derivation of $t \emp$.  As we will see, $k$ never changes, and $v$ becomes
  smaller.  Recall that we write $v < v'$ when $v$ is a strict subterm of $v'$.

  This shows we prove a slightly stronger statement than simply that $v$ is empty, namely
  that $\nsof[k]{v}{t}$ for all $k$.  When showing a derivation we implicitly apply
  weakening when we do not use an assumption any longer, reading in proof construction
  order from the conclusion to the premises.
\begin{description}
\item[Case:]
  \begin{mathpar}
    \infer[\remp{\tensor}_i]
    {t \emp}
    {t = t_1 \tensor t_2 \in \Sigma & t_i \emp}
    \and
    \infer[]
    {\sof[k]{v}{t} \vdash \cdot}
    {\infer[]
      {v = \pair{v_1,v_2}, \sof[k]{v_1}{t_1}, \sof[k]{v_2}{t_2} \vdash \cdot}
      {\sof[k]{v_i}{t_i} \vdash \cdot}}
  \end{mathpar}
  Observe that $v_i < v = \pair{v_1,v_2}$ in the premise.
\item[Case:]
  \begin{mathpar}
    \infer[\remp{\oplus}]
    {t \emp}
    {t = \plus*[\ell \in L]{\ell : t_\ell} \in \Sigma &
      t_j \emp & (\forall j \in L)}
    \and
    \infer[]
    {\sof[k]{v}{t} \vdash \cdot}
    {\infer[]
      {\bigvee_{\ell \in L} (v = \inj{j}{v_j} \land \sof[k]{v_j}{t_j}) \vdash \cdot}
      {\sof[k]{v_j}{t_j} \vdash \cdot & (\forall j \in L)}}
  \end{mathpar}
  In each of the $|L|$ premises we have $v_j < v = \inj{j}{v_j}$ so the structure of $v$
  decreases.
\item[Case:] $t \emp$ is justified by a cycle.  Then $\sof[k]{v}{t}$ is justified by a
  corresponding cycle.
\end{description}
\end{proof}

\section{Fullness}
\label{app:fullness}

\begin{proof}
  (of Theorem~\ref{thm:full})

  There are three cases for why $\sof[k]{e}{r}$ could be true: (1) $k = 0$, (2)
$k > 0 \land e \reduces e' \land \sof[k-1]{e'}{r}$ and (3)
$k > 0 \land e \term \land \sof*[k]{e}{r}$.  Only in the last do we distinguish between
the rules for the $s \full$ judgment.

\begin{description}
\item[Case:] $k = 0$.  Then $\sof[0]{e}{s}$ is true by definition.
\item[Case:] $k > 0$ and $e \reduces e'$ with $\sof[k-1]{e'}{r}$.
  \begin{mathpar}
    \infer[]
    {\sof[k]{e}{r} \vdash \sof[k]{e}{s}}
    {\cdots
      & \infer[]
      {k > 0, e \reduces e', \sof[k-1]{e'}{r} \vdash \sof[k]{e}{s}}
      {\deduce
        {\sof[k-1]{e'}{r} \vdash \sof[k-1]{e'}{s}}
        {\cycle{\scriptstyle k-1/k, e'/e}}}
      & \cdots}
  \end{mathpar}
  So in this case we close the derivation with a local cycle, which corresponds to an
  appeal of the induction hypothesis on $k-1$, regardless of $s$.  We indicate
  here the substitution for the parameters that is applied as part of forming
  the cycle.
\item[Case:] $k > 0$ and $e \term$ with $\sof*[k]{e}{r}$.  We do not use the
  last assumption.  Now we distinguish cases on the rule use to derive $\sof[k]{e}{s}$.
  \begin{description}
  \item[Subcase:]
    \begin{mathpar}
      \infer[\rfull{\imp}]
      {s \full}
      {s = t_1 \imp s_2 & t_1 \emp}
      \and
      \infer[]
      {\sof[k]{e}{r} \vdash \sof[k]{e}{s}}
      {\cdots & \cdots
        & \infer[]
        {k > 0, e \term, \sof*[k]{e}{r} \vdash \sof[k]{e}{t_1 \imp s_2}}
        {\infer[]
          {k > 0 \vdash \forall i < k.\, \forall v.\, \sof[i]{v}{t_1} \mimp \sof[k]{\app{e}{v}}{s_2}}
          {\infer[]
            {k > 0, i < k, \sof[i]{v}{t_1} \vdash \sof[k]{\app{e}{v}}{s_2}}
            {\deduce{\sof[i]{v}{t_1} \vdash \cdot}{(t_1 \emp)}}}}}
    \end{mathpar}
    Here, we reduce the result to Theorem~\ref{thm:empty} (using weakening here not only in the
    antecedent but also in the succedent).
  \item[Subcase:] Similar to the previous case.
    \begin{mathpar}
      \infer[\rfull{\with}]{\with*{\,} \full}{}
      \and
      \infer[]
      {\sof[k]{e}{r} \vdash \sof[k]{e}{s}}
      {\cdots & \cdots
        & \infer[]
        {k > 0, e \term, \sof*[k]{e}{r} \vdash \sof[k]{e}{s}}
        {\infer[]
          {k > 0 \vdash \sof[k]{e}{\with*{\,}}}
          {\infer[]
            {k > 0 \vdash \forall j \in \emptyset.\, \sof[k]{\proj{j}{e}}{\with*{\,}}}
            {\deduce
              {j \in \emptyset \vdash \sof[k]{\proj{j}{e}}{\with*{\,}}}
              {(\mbox{antecedent contradictory})}}}}}
    \end{mathpar}
  \end{description}
\end{description}

\end{proof}

\section{Subtyping}
\label{app:subtyping}

\begin{proof} (of Theorem~\ref{thm:sd-subtyping})
  We proceed by a compositional translation of the circular derivation of subtyping into a
  circular derivation in the metalogic.  For each rule we construct a derived rule
  on the semantic side with corresponding premises and conclusion.

  When the subtyping proof is closed due to a cycle, we close the proof in the metalogic
  with a corresponding cycle.  In order for this cycle to be valid, it is critical that
  the judgments in the premises of the derived rule are \emph{strictly smaller} than the
  judgments in the conclusion.  Since our mixed logical relation is defined by nested
  induction, first on the step index $k$ and second on the structure of the value $v$,
  this is the lexicographic measure $(k,v)$ should strictly decrease.

  We provide some sample cases.  We freely apply weakening to simplify the judgment under
  consideration.
  \begin{description}
  \item[Case:]
    \begin{mathpar}
      \infer[\rsub{\tensor}]
      {t \leq u}
      {t = t_1 \tensor t_2 & u = u_1 \tensor u_2
        & t_1 \leq u_1 & t_2 \leq u_2}
      \and
      \infer[]
      {\sof[k]{v}{t} \vdash \sof[k]{v}{u}}
      {\infer[]
        {v = \pair{v_1,v_2}, \sof[k]{v_1}{t_1}, \sof[k]{v_2}{t_2} \vdash \sof[k]{\pair{v_1,v_2}}{u}}
        {\infer[]
          {\sof[k]{v_1}{t_1}, \sof[k]{v_2}{t_2} \vdash
            \sof[k]{v_1}{u_1} \land \sof[k]{v_2}{u_2}}
          {\deduce
            {\sof[k]{v_1}{t_1} \vdash \sof[k]{v_1}{u_1}}
            {}
            & \deduce
            {\sof[k]{v_2}{t_2} \vdash \sof[k]{v_2}{u_2}}
            {}}}}
    \end{mathpar}
    Observe that $v = \pair{v_1,v_2}$ so $(k,v_1) < (k,v)$ in the left branch
    and $(k,v_2) < (k,v)$ in the second branch.
  \item[Case:]
    \begin{mathpar}
      \infer[\rsub{\one}]
      {t \leq u}
      {t = \one & u = \one}
      \and
      \infer[]
      {\sof[k]{v}{t} \vdash \sof[k]{v}{u}}
      {\infer[]
        {\sof[k]{v}{\one} \vdash \sof[k]{v}{\one}}
        {}}
    \end{mathpar}
  \item[Case:]
    \begin{mathpar}
      \infer[\rsub{\oplus}]
      {t \leq u}
      {\begin{array}{rlc}
      t = &\plus*[\ell \in L]{\ell : t_\ell}&\\
         u = & \plus*[k \in K]{k : u_k}
        &\qquad \mbox{$\forall \ell \in L.\,$ $t_\ell \emp$ or $\ell \in K$ and $t_\ell \leq u_\ell$}
        \end{array}}
      \and
      \infer[]
      {\sof[k]{v}{t} \vdash \sof[k]{v}{u}}
      {\infer[\!(*)]
        {v = \inj{j}{v_j}, j \in L, \sof[k]{v_j}{t_j} \vdash \sof[k]{v}{u}}
        {\infer[]
          {t_j \emp, j \in L, \sof[k]{v_j}{t_j} \vdash \sof[k]{\inj{j}{v_j}}{u}}
          {t_j \emp \vdash \cdot}
          \,
          \infer[]
          {j \in K, j \in L, \sof[k]{v_j}{t_j} \vdash \sof[k]{\inj{j}{v_j}}{u}}
          {\deduce
            {\sof[k]{v_j}{t_j} \vdash \sof[k]{v_j}{u_j}}
            {}}}}
    \end{mathpar}
    At the inference $(*)$ we distinguish the two cases from the premise of $\rsub{\oplus}$
    for $\ell = j$: either $t_j \emp$ or $j \in K$.  Observe that $v_j < v = \inj{j}{v_j}$.
  \end{description}

  For computations, we separate out the cases the $k = 0$ and $k > 0$ with
  $e \reduces e'$ because the argument is essentially the same except in the case of
  $\rsub{\top}$.  When $k > 0$ and $e \term$ we distinguish cases based on the various
  rules.

  \begin{description}
  \item[Case:]
    \begin{mathpar}
      \infer[\rsub{\dn}]
      {t \leq u}
      {t = \dn s & u = \dn r & s \leq r}
      \and
      \infer[]
      {\sof[k]{v}{t} \vdash \sof[k]{v}{u}}
      {\infer[]
        {v = \thunk{e}, \sof[k]{e}{s} \vdash \sof[k]{v}{u}}
        {\deduce
          {\sof[k]{e}{s} \vdash \sof[k]{e}{r}}
          {}}}
    \end{mathpar}
  \item[Case:] $s \leq r$ and $\sof[k]{e}{s}$ for $k = 0$.
    Then, $\sof[0]{e}{r}$ directly by definition.
    \begin{mathpar}
      \infer[]
      {\sof[k]{e}{s} \vdash \sof[k]{e}{r}}
      {\infer[]
        {k = 0 \vdash \sof[k]{e}{r}}
        {}}
    \end{mathpar}
  \item[Case:] $k > 0$ and $e \reduces e'$.  Then we can close of the derivation with a
    (local) cycle, representing an immediate appeal to the induction hypothesis with
    $k-1 < k$.
    \begin{mathpar}
      \infer[]
      {\sof[k]{e}{s} \vdash \sof[k]{e}{r}}
      {\infer[]
        {k > 0, e \reduces e', \sof[k-1]{e'}{s} \vdash \sof[k]{e}{r}}
        {\infer[]
          {k > 0, e \reduces e', \sof[k-1]{e'}{s} \vdash k > 0 \land e \reduces e' \land \sof[k-1]{e'}{r}}
          {\deduce
            {\sof[k-1]{e'}{s} \vdash \sof[k-1]{e'}{r}}
            {\cycle{k-1/k,e'/e}}}}}
    \end{mathpar}
  \item[Case:] $k > 0$ and $e \term$.  Then we distinguish subcases based on the
    rule to conclude $\sof[k]{e}{s}$.
    \begin{description}
    \item[Subcase:]
      \begin{mathpar}
        \infer[\rsub{\imp}]
        {s \leq r}
        {s = t_1 \imp s_2 & r = u_1 \imp r_2
          & u_1 \leq t_1 & s_2 \leq r_2}
        \and
        \infer[]
        {\sof[k]{e}{s} \vdash \sof[k]{e}{r}}
        {\cdots  &
          \infer[]
          {\forall i < k.\, \forall v.\, \sof[i]{v}{t_1} \mimp \sof[k]{\app{e}{v}}{s_2}
            \vdash \sof[k]{e}{r}}
          {\infer[]
            {\forall i < k.\, \forall v.\, \sof[i]{v}{t_1} \mimp \sof[k]{\app{e}{v}}{s_2}
              \vdash
              \forall j < k.\, \forall w.\, \sof[j]{w}{u_1} \mimp \sof[k]{\app{e}{w}}{r_2}}
            {\infer[]
              {\forall i < k.\, \forall v.\, \sof[i]{v}{t_1} \mimp \sof[k]{\app{e}{v}}{s_2},
                j < k, \sof[j]{w}{u_1} \vdash \sof[k]{\app{e}{w}}{r_2}}
              {\infer[]
                {j < k, \sof[j]{w}{t_1} \mimp \sof[k]{\app{e}{w}}{s_2},
                  \sof[j]{w}{u_1} \vdash \sof[k]{\app{e}{w}}{r_2}}
                {\deduce{\sof[j]{w}{u_1} \vdash \sof[j]{w}{t_1}} {}
                  &
                  \infer[]
                  {\sof[k]{\app{e}{w}}{s_2} \vdash \sof[k]{\app{e}{w}}{r_2}}
                  {\infer[(*)]
                    {\app{e}{w} \reduces e', \sof[k-1]{e'}{s_2} \vdash \sof[k]{\app{e}{w}}{r_2}}
                    {\deduce{\sof[k-1]{e'}{s_2} \vdash \sof[k-1]{e'}{r_2}}{}}}}}}}}
      \end{mathpar}
      In the place marked $(*)$ we only have one possible case since $k > 0$ and $e \term$
      and therefore $\app{e}{w}$ is not terminal and must reduce since
      $\sof[k]{\app{e}{w}}{s_2}$.

      In the first open premise we have $(j,w) < (k,w)$ because $j < k$ (even if $w$ is
      arbitrary).  In the second open premise we have $k-1 < k$.
    \item[Subcase:] Recall that $k > 0$ and $e \term$.
      \begin{mathpar}
        \infer[\rsub{\with}]
        {s \leq r}
        {s = \with*[\ell \in L]{\ell : s_\ell}
          & r = \with*[j \in K]{j : r_j}
          & \mbox{$\forall j \in K.\, j \in L \land s_j \leq r_j$}}
        \and
        \infer[]
        {\sof[k]{e}{s} \vdash \sof[k]{e}{r}}
        {\infer[]
          {\sof*[k]{e}{\with*[\ell \in L]{\ell : s_\ell}} \vdash \sof[k]{e}{r}}
          {\infer[]
            {\forall \ell.\, \ell \in L \mimp \sof[k]{\proj{\ell}{e}}{s_\ell} \vdash \sof[k]{e}{r}}
            {\infer[(*)]
              {\forall \ell.\, \ell \in L \mimp \sof[k]{\proj{\ell}{e}}{s_\ell} \vdash \forall j.\, j \in K \mimp
                \sof[k]{\proj{j}{e}}{r}}
              {\infer[(**)]
                {\sof[k]{\proj{j}{e}}{s_j} \vdash \sof[k]{\proj{j}{e}}{r_j}}
                {\infer[]
                  {\proj{j}{e} \reduces e', \sof[k-1]{e'}{s_j} \vdash \sof[k]{\proj{j}{e}}{r_j}}
                  {\deduce{\sof[k-1]{e'}{s_j} \vdash \sof[k-1]{e'}{r_j}}{}}}}}}}
      \end{mathpar}
      At the inference $(*)$ we use that $j \in L$ by the premise of
      $\rsub{\with}$.  At the inference $(**)$ with use that $k > 0$ and
      $\proj{j}{e}$ is not terminal.
    \item[Subcase:] Recall that $k > 0$ and $e \term$.
      \begin{mathpar}
        \infer[\rsub{\up}]
        {s \leq r}
        {s = \up t & r = \up u
          & t \leq u}
        \and
        \infer[]
        {\sof[k]{e}{s} \vdash \sof[k]{e}{r}}
        {\infer[]
          {e = \return{v}, \sof[k-1]{v}{t} \vdash \sof[k]{e}{r}}
          {\sof[k-1]{v}{t} \vdash \sof[k-1]{v}{u}}}
      \end{mathpar}
      Observe that in the translation of $t \leq u$ we have $k-1 < k$.
    \item[Subcase:]  Recall that $k > 0$ and $e \term$.
      \begin{mathpar}
        \infer[\n{\rsub{\bot}}]{s \leq r}{s = \up t & t \emp & r = \n{\sigma}}
        \and
        \infer[]
        {\sof[k]{e}{s} \vdash \sof[k]{v}{r}}
        {\infer[]
          {\sof*[k]{e}{\up t} \vdash \sof[k]{e}{r}}
          {\infer[]
            {e = \return{v}, \sof[k-1]{v}{t} \vdash \sof[k]{e}{r}}
            {\deduce{\sof[k-1]{v}{t} \vdash \cdot}{(t \emp)}}}}
    \end{mathpar}
    \end{description}
  \end{description}
  The last two cases follow immediately from the properties of the emptiness and fullness
  judgments.
  \begin{description}
  \item[Case:]
    \begin{mathpar}
      \infer[\p{\rsub{\bot}}]{t \leq u}{t \emp & u = \p{\tau}}
      \and
      \infer[]
      {\sof[k]{v}{t} \vdash \sof[k]{v}{u}}
      {\deduce[]{\sof[k]{v}{t} \vdash \cdot}{(t \emp)}}
    \end{mathpar}
  \item[Case:]
    \begin{mathpar}
      \infer[\rsub{\top}]{s \leq r}{s = \n{\sigma} & r \full}
      \and
      \deduce[]
      {\sof[k]{e}{t} \vdash \sof[k]{e}{u}}
      {(r \full)}
    \end{mathpar}
    In this case, we can appeal to the lemma for fullness because we have the assumption
    that $\sof[k]{v}{t}$.
  \end{description}
\end{proof}

\section{Reflexivity and Transitivity of Syntactic Subtyping}
\label{app:refl-trans}

\begin{theorem}[Reflexivity and Transitivity]
  \label{thm:transitivity}
  \mbox{}
  \begin{enumerate}
  \item $t \leq t$ and $s \leq s$ for all types names $s$ and $t$ in signature $\Sigma$
  \item $t_1 \leq t_2$ and $t_2 \leq t_3$ implies $t_1 \leq t_3$
  \item $s_1 \leq s_2$ and $s_2 \leq s_3$ implies $s_1 \leq s_3$
  \end{enumerate}
\end{theorem}
\begin{proof}
  All rules except $\p{\rsub{\bot}}$, $\n{\rsub{\bot}}$, and $\rsub{\top}$
  simply compare components, thus are directly amenable to reflexivity.
  For rule $\p{\rsub{\bot}}$, $t \emp$ implies $t \leq t$.
  For rule $\n{\rsub{\bot}}$, $t \emp$ again implies $\up t \leq \up t$.
  Finally, for $\rsub{\top}$, $s \full$ implies $s \leq s$.

  Proving transitivity requires an additional lemma: If $t_1 \leq t_2$
  and $t_2 \emp$, then $t_1 \emp$. This lemma follows by applying inversion
  to the syntactic subtyping judgment.
  This lemma can then be utilized to take circular proofs of $t_1 \leq t_2$
  and $t_2 \leq t_3$ to assemble a circular proof of $t_1 \leq t_3$.
  A similar proof technique holds for (3).
\end{proof}

\section{Declarative Typing Judgments}
\label{app:syntactic-typing}

While semantic typing worked with closed values and computations only, the syntactic rules
require consideration of free variables.  In a polarized presentation they always stand
for values and therefore have positive type.  We collect them in a \emph{context} $\Gamma$
and, as usual, presuppose that all variables declared in a context are distinct.
\begin{align*}
  \ctx &\Coloneqq \ctxe \mid \ctx , x{:}\p{\tau}
\end{align*}
There are several official judgments for the syntactic validity of signatures, contexts,
types, and the typing of values and computations.  In order to avoid excessive bureaucracy
we use some presuppositions and some implicit checking or renaming to maintain these.  The
complete list of judgments can be found in Figure~\ref{fig:typing-judgments}.

\begin{figure}
\[
  \begin{tabular}{lll}
    Valid Signatures & $\vdash_\Sigma \Sigma' \vsig$ \\
    Valid Contexts & $\vdash_\Sigma \Gamma \pctx$ & presupposing $\vdash_\Sigma \Sigma \vsig$ \\
    Valid Pos.\ Types & $\vdash_\Sigma \p{\tau} \ptype$ & presupposing $\vdash_\Sigma \Sigma \vsig$ \\
    Valid Neg.\ Types & $\vdash_\Sigma \n{\sigma} \ntype$ & presupposing $\vdash_\Sigma \Sigma \vsig$ \\
    Value Typing & $\Gamma \vdash_\Sigma v : \p{\tau}$ & presupposing $\vdash_\Sigma \Gamma \pctx$
    and $\vdash_\Sigma \p{\tau} \ptype$ \\
    Computation Typing & $\Gamma \vdash_\Sigma e : \n{\sigma}$
                   & presupposing $\vdash_\Sigma \Gamma \pctx$ and $\vdash_\Sigma \n{\sigma} \ntype$ \\
    Positive Subtyping & $\vdash_\Sigma \p{\tau} \leq \p{\sigma}$
                   & presupposing $\vdash_\Sigma \p{\tau} \ptype$ and $\vdash_\Sigma \p{\sigma} \ptype$ \\
    Negative Subtyping & $\vdash_\Sigma \n{\tau} \leq \n{\sigma}$
                   & presupposing $\vdash_\Sigma \n{\tau} \ntype$ and $\vdash_\Sigma \n{\sigma} \ntype$
  \end{tabular}
\]
\caption{Syntactic Typing Judgments}
\label{fig:typing-judgments}
\end{figure}

The last two arise from $\p{\tau} \leq \p{\sigma}$ and $\n{\tau} \leq \n{\sigma}$ in that
they do not require the normal form of alternating names and structural types introduced
in Section~\ref{sec:empty-full}.  In particular, we have $\vdash_\Sigma \p{\tau} \leq t$
and $\vdash_\Sigma t \leq \p{\tau}$ if $t = \p{\tau} \in \Sigma$, and analogously for
negative types.  This captures the equirecursive nature of type definitions.  We omit its
straightforward rules, as well as the rules for valid types which only check that all its
type names are defined in the signature.  We write $\p{\tau} \structural$ and
$\n{\sigma} \structural$ if $\p{\tau}$ and $\n{\sigma}$ are not type names, which is needed
to guarantee that type definitions are contractive.

\begin{figure}
\begin{mathpar}
  \infer[\irule{\tensor}]{\Gamma \vdash \pair{v_1,v_2} : \p{\tau_1} \tensor \p{\tau_2}}
  {\Gamma \vdash v_1 : \p{\tau_1} & \Gamma \vdash v_2 : \p{\tau_2}}
  \and
  \infer[\erule{\tensor}]{\Gamma \vdash \letpair{x,y}{v}{e} : \n{\sigma}}
  {\Gamma \vdash v : \p{\tau_1} \tensor \p{\tau_2} &
    \Gamma , x{:}\p{\tau_1} , y{:}\p{\tau_2} \vdash e : \n{\sigma}}
  \\
  \infer[\jrule{VAR}]{\Gamma \vdash x : \p{\tau}}{x : \p{\tau} \in \Gamma}
  \and
  \infer[\irule{\one}]{\Gamma \vdash \pair{} : \one}{}
  \and
  \infer[\erule{\one}]{\Gamma \vdash \letpair{}{v}{e} : \n{\sigma}}
  {\Gamma \vdash v : \one & \Gamma \vdash e : \n{\sigma}}
  \\\
  \infer[\irule{\plus}]{\Gamma \vdash \inj{j}{v} : \plus*[\ell \in L]{\ell\colon \p{\tau_\ell}}}
  {\text{($j \in L$)} & \Gamma \vdash v : \p{\tau_j}}
  \quad
  \infer[\erule{\plus}]{\Gamma \vdash \case[\ell \in L]{v}{{\ell}{x_\ell} => e_\ell} : \n{\sigma}}
  {\Gamma \vdash v : \plus*[\ell \in L]{\ell\colon \p{\tau_\ell}} &
   \Gamma , x_\ell{:}\p{\tau_\ell} \vdash e_\ell : \n{\sigma}\; (\forall \ell \in L)}
  \\\
  \infer[\irule{\dn}]{\Gamma \vdash \thunk{e} : \dn \n{\sigma}}
  {\Gamma \vdash e : \n{\sigma}}
  \and
  \infer[\erule{\dn}]{\Gamma \vdash \force{v} : \n{\sigma}}
  {\Gamma \vdash v : \dn \n{\sigma}}
  \\
  \infer[\irule{\imp}]{\Gamma \vdash \lam{x}{e} : \p{\tau} \imp \n{\sigma}}
  {\Gamma , x{:}\p{\tau} \vdash e : \n{\sigma}}
  \and
  \infer[\erule{\imp}]{\Gamma \vdash \app{e}{v} : \n{\sigma}}
  {\Gamma \vdash e : \p{\tau} \imp \n{\sigma} & \Gamma \vdash v : \p{\tau}}
  \\
  \infer[\irule{\with}]{\Gamma \vdash \record[\ell \in L]{\ell = e_\ell} : \with*[\ell \in L]{\ell\colon \n{\sigma_\ell}}}
  {\Gamma \vdash e_\ell : \n{\sigma_\ell}\; (\forall \ell \in L)}
  \and
  \infer[\erule{\with}_k]{\Gamma \vdash \proj{j}{e} : \n{\sigma_j}}
  {\Gamma \vdash e : \with*[\ell \in L]{\ell\colon \n{\sigma_\ell}} &
   \text{($j \in L$)}}
  \\
  \infer[\irule{\up}]{\Gamma \vdash \return{v} : \up \p{\tau}}
  {\Gamma \vdash v : \p{\tau}}
  \and
  \infer[\erule{\up}]{\Gamma \vdash \letup{x}{e_1}{e_2} : \n{\sigma}}
  {\Gamma \vdash e_1 : \up \p{\tau} & \Gamma , x{:}\p{\tau} \vdash e_2 : \n{\sigma}}
  \\
  \infer[\jrule{NAME}]
  {\Gamma \vdash f : \n{\sigma}}
  {f : \n{\sigma} = e \in \Sigma}
  \quad
  \infer[\p{\jrule{SUB}}]
  {\Gamma \vdash v : \p{\sigma}}
  {\Gamma \vdash v : \p{\tau} & \p{\tau} \leq \p{\sigma}}
  \quad
  \infer[\n{\jrule{SUB}}]
  {\Gamma \vdash e : \n{\sigma}}
  {\Gamma \vdash e : \n{\tau} & \n{\tau} \leq \n{\sigma}}
  \\
  \infer[]
  {\vdash_\Sigma \ctxe* \vsig}
  {}
  \and
  \infer[]
  {\vdash_\Sigma (\Sigma', t = \p{\tau}) \vsig}
  {\vdash_\Sigma \Sigma' \vsig
    & \vdash_\Sigma \p{\tau} \ptype
    & \p{\tau} \structural}
  \and
  \infer[]
  {\vdash_\Sigma (\Sigma', s = \n{\sigma}) \vsig}
  {\vdash_\Sigma \Sigma' \vsig
    & \vdash_\Sigma \n{\sigma} \ntype
    & \n{\sigma} \structural}
  \and
  \infer[]
  {\vdash_\Sigma (\Sigma', f : \n{\sigma} = e) \vsig}
  {\vdash_\Sigma \Sigma'
    & \vdash_\Sigma \n{\sigma} \ntype
    & \cdot \vdash_\Sigma e : \n{\sigma}}
  \and
  \infer[]
  {\vdash_\Sigma \ctxe* \pctx}
  {}
  \and
  \infer[]
  {\vdash_\Sigma (\Gamma, x : \p{\tau})}
  {\vdash_\Sigma \Gamma \pctx
    & \vdash_\Sigma \p{\tau} \ptype}
\end{mathpar}
\caption{Declarative Syntactic Typing}
\label{fig:declarative-typing}
\end{figure}

\section{Soundness of Syntactic Typing}
\label{app:sd-typing}

We state and proof the rules for semantic typing from Section~\ref{sec:syntactic-typing} separately.

\begin{lemma}
  \label{lm:sd-app}
  \label{lm:sd-first}
  \[
    \infer-{\sof[k]{\app{e}{v}}{\n{\sigma}}}
    {\sof[k]{e}{\p{\tau} \imp \n{\sigma}}
      & \sof[k]{v}{\p{\tau}}}
  \]
\end{lemma}
\begin{proof}
  The proof is by induction on $k$.
  \begin{description}
  \item[Case:] $k = 0$.
  \begin{tabbing}
    $\sof[0]{\app{e}{v}}{\n{\sigma}}$ \` By definition
  \end{tabbing}
  \item[Case:] $k > 0$ and $e \term$
  \begin{tabbing}
    $\sof[k]{e}{\p{\tau} \imp \n{\sigma}}$ \` First premise \\
    $\sof[i]{v}{\p{\tau}}$ for all $i < k$ \` From second premise by downward closure \\
    $\sof[k]{\app{e}{v}}{\n{\sigma}}$ \` By definition and second premise
  \end{tabbing}
  \item[Case:] $e \reduces e'$ and $k > 0$
  \begin{tabbing}
    $\sof[k-1]{e'}{\p{\tau} \imp \n{\sigma}}$ \` By definition \\
    $\sof[i]{v}{\p{\tau}}$ for all $i < k-1$ \` From second premise and downward closure \\
    $\sof[k-1]{\app{e'}{v}}{\p{\tau} \imp \n{\sigma}}$ \` By ind.\ hyp. \\
    $\sof[k]{\app{e}{v}}{\p{\tau}}$ \` Since $\app{e}{v} \reduces \app{e'}{v}$
  \end{tabbing}
  \end{description}
\end{proof}

\begin{lemma}
  \label{lm:sd-lam}
  \[
    \infer-{\sof[k]{\lam{x}{e}}{\p{\tau} \imp \n{\sigma}}}
    {x : \p{\tau} \models \sof[k]{e}{\n{\sigma}}}
  \]

  \begin{proof}
    \mbox{}
    \begin{tabbing}
      $\sof[k]{[v/x]e}{\n{\sigma}}$ for all $i \leq k$ and $v$ with $\sof[i]{v}{\p{\tau}}$ \` Premise \\
      $\sof[k+1]{\app{(\lam{x}{e})}{v}}{\n{\sigma}}$ \` By definition since $\app{(\lam{x}{e})}{v} \reduces
      [v/x]e$ \\
      $\sof[k+1]{\lam{x}{e}}{\n{\sigma}}$ \` By definition \\
      $\sof[k]{\lam{x}{e}}{\n{\sigma}}$ \` By downward closure
    \end{tabbing}
  \end{proof}
\end{lemma}

\begin{lemma}
  \[
    \infer-{\sof[k]{\pair{v_1,v_2}}{\p{\tau_1} \tensor \p{\tau_2}}}
    {\sof[k]{v_1}{\p{\tau_1}} & \sof[k]{v_2}{\p{\tau_2}}}
  \]
  \begin{proof}
  \mbox{}
    \begin{tabbing}
      $\sof[k]{v_1}{\p{\tau_1}}$ and $\sof[k]{v_2}{\p{\tau_2}}$ \` Premises \\
      $\sof[k]{\pair{v_1,v_2}}{\p{\tau_1} \tensor \p{\tau_2}}$ \` By definition
    \end{tabbing}
  \end{proof}
\end{lemma}

\begin{lemma}
  \[
    \infer-{\sof[k]{\letpair{x,y}{v}{e}}{\n{\sigma}}}
    {\sof[k]{v}{\p{\tau_1} \tensor \p{\tau_2}}
      & x : \p{\tau_1}, y : \p{\tau_2} \models \sof[k]{e}{\n{\sigma}}}
  \]
\end{lemma}
\begin{proof}
  \mbox{}
  \begin{tabbing}
    $\sof[i]{v}{\p{\tau_1} \tensor \p{\tau_2}}$ for all $i \leq k$ \` From first premise by downward closure \\
    $v = \pair{v_1,v_2}, \sof[i]{v_1}{\p{\tau_1}}, \sof[i]{v_2}{\p{\tau_2}}$ \` By definition \\
    $\sof[k]{[v_1/x][v_2/y]e}{\n{\sigma}}$ \` By second premise \\
    $\sof[k+1]{\letpair{x,y}{\pair{v_1,v_2}}{e}}{\n{\sigma}}$ \` By definition \\
    $\sof[k]{\letpair{x,y}{v}{e}}$ \` Since $v = \pair{v_1,v_2}$ and downward closure
  \end{tabbing}
\end{proof}

\begin{lemma}
  \[
    \infer-{\sof[k]{\return{v}}{\up \p{\tau}}}
    {\sof[k]{v}{\p{\tau}}}
  \]
\end{lemma}
\begin{proof}
  \mbox{}
  \begin{tabbing}
    $\sof[k]{v}{\p{\tau}}$ \` Premise \\
    $\sof[k+1]{\return{v}}{\up \p{\tau}}$ \` By definition \\
    $\sof[k]{\return{v}}{\up \p{\tau}}$ \` By downward closure \\
  \end{tabbing}
\end{proof}

\begin{lemma}
  \[
    \infer-{\sof[k]{\letup{x}{e_1}{e_2}}{\n{\sigma}}}
    {\sof[k]{e_1}{\up \p{\tau}}
      & x : \p{\tau} \models \sof[k]{e_2}{\n{\sigma}}}
  \]
\end{lemma}
\begin{proof}
  By induction on $k$.
  \begin{description}
  \item[Case:] $k = 0$.
    \begin{tabbing}
      $\sof[0]{\letup{x}{e_1}{e_2}}{\n{\sigma}}$ \` By definitions
    \end{tabbing}
  \item[Case:] $k > 0$.  Now we distinguish subcases on $\sof[k]{e_1}{\up \p{\tau}}$.
    \begin{description}
      \item[Subcase:] $e_1 \reduces e_1'$ and $\sof[k-1]{e_1'}{\up \p{\tau}}$.
        \begin{tabbing}
          $\sof[i]{e_1'}{\up \p{\tau}}$ for all $i \leq k-1$ \` By downward closure \\
          $x : \p{\tau} \models \sof[k-1]{e_2}{\n{\sigma}}$ \` From second premise \\
          $\letup{x}{e_1}{e_2} \reduces \letup{x}{e_1'}{e_2}$ \` By rule \\
          $\sof[k-1]{\letup{x}{e_1'}{e_2}}{\up \p{\tau}}$ \` By ind.\ hyp \\
          $\sof[k]{\letup{x}{e_1}{e_2}}{\up \p{\tau}}$ \` By definition
        \end{tabbing}
      \item[Subcase:] $e_1 \term$ and $e_1 = \return{v_1}$ with
        $\sof[k-1]{v_1}{\p{\tau}}$
        \begin{tabbing}
          $\sof[i]{v_1}{\p{\tau}}$ for all $i \leq k-1$ \` By downward closure \\
          $\sof[k-1]{[v_1/x]e_2}{\n{\sigma}}$ \` From second premise and downward closure \\
          $\sof[k]{\letup{x}{\return{v_1}}{e_2}}{\n{\sigma}}$ \` By definition \\
          $\sof[k]{\letup{x}{e_1}{e_2}}{\n{\sigma}}$ \` Since $e_1 = \return{v_1}$
    \end{tabbing}
  \end{description}
\end{description}
\end{proof}

\begin{lemma}
  \[
    \infer-{\sof[k]{\thunk{e}}{\dn \n{\sigma}}}
    {\sof[k]{e}{\n{\sigma}}}
  \]
\end{lemma}
\begin{proof}
  \mbox{}
  \begin{tabbing}
    $\sof[k]{e}{\n{\sigma}}$ \` Premise \\
    $\sof[k]{\thunk{e}}{\n{\sigma}}$ \` By definition
  \end{tabbing}
\end{proof}

\begin{lemma}
  \[
    \infer-[\p{\jrule{SUB}}]
    {\sof[k]{v}{\p{\sigma}}}
    {\sof[k]{v}{\p{\tau}} & \p{\tau} \leq \p{\sigma}}
  \]
  \begin{proof}
  \mbox{}
    \begin{tabbing}
      $\sof[k]{v}{\p{\tau}}$ \` First premise \\
      $\sof[k]{v}{\p{\sigma}}$ \` By second premise and Theorem~\ref{thm:sd-subtyping}
    \end{tabbing}
  \end{proof}
\end{lemma}

\begin{lemma}
  \label{lm:sd-last}
  \[
    \infer-{\sof[k]{\force{v}}{\n{\sigma}}}
    {\sof[k]{v}{\dn \n{\sigma}}}
  \]
\end{lemma}
\begin{proof}
  \mbox{}
  \begin{tabbing}
    $\sof[k]{v}{\dn \n{\sigma}}$ \` Given \\
    $v = \thunk{e}$ and $\sof[k]{e}{\n{\sigma}}$ \` By definition
  \end{tabbing}
\end{proof}

\begin{proof}
  (of Theorem~\ref{thm:type-sound})

  \begin{description}
  \item[Case:]
    \begin{mathpar}
      \infer[\jrule{VAR}]{\Gamma \vdash x : \p{\tau}}{x{:}\p{\tau} \in \Gamma}
    \end{mathpar}
    \begin{tabbing}
      $\sof[k]{\theta}{\Gamma}$ \` Given \\
      $(v/x) \in \theta$ with $\sof[k]{v}{\p{\tau}}$ \` By definition \\
      $\sof[k]{x[\theta]}{\p{\tau}}$
    \end{tabbing}
  \item[Case:]
    \begin{mathpar}
      \infer[\jrule{NAME}]
      {\Gamma \vdash f : \n{\sigma}}
      {f : \n{\sigma} = e \in \Sigma}
      \and
      \infer[]
      {\sof[k]{f}{\n{\sigma}}}
      {\infer[]{k = 0 \vdash \sof[k]{f}{\n{\sigma}}}{}
        &
        \infer[]
        {k > 0, f \reduces e \vdash \sof[k]{e}{\n{\sigma}}}
        {\sof[k-1]{e}{\n{\sigma}}}}
    \end{mathpar}
    If $f$ has not yet been translated, we deduce $f[\theta] = \sof[k]{f}{\n{\sigma}}$
    from $f \reduces e$ and $\sof[k-1]{e}{\n{\sigma}}$ if $k > 0$.  In this case it is
    important that $k > k-1$.

    If $f$ has already been translated (that is, we are in the premise of the translation
    of this rule application), then it will be at a judgment $\sof[k']{f}{\n{\sigma}}$ for
    some $k > k'$ and we can form a valid cycle.

    This translation results in a finite circular proof for two reasons:
    \begin{enumerate}
    \item There are only finitely many definitions $f : \n{\sigma} = e \in \Sigma$.
    \item The type for $f$ is fixed to be $\n{\sigma}$, so when $f$ is encountered
      in the derivation of $\sof[k-1]{e}{\n{\sigma}}$ we can always form a valid cycle.
    \end{enumerate}
  \item[Case:]
    \begin{mathpar}
      \infer[\erule{\imp}]
      {\ctx \vdash \app{e}{v} : \n{\sigma}}
      {\ctx \vdash e : \p{\tau} \imp \n{\sigma} & \ctx \vdash v : \p{\tau}}
      \and
      \infer=[]
      {\sof[k]{(\app{e}{v})[\theta]}{\n{\sigma}}}
      {\infer[Lemma~\ref{lm:sd-app}]
        {\sof[k]{\app{(e[\theta])}{(v[\theta])}}{\n{\sigma}}}
        {\sof[k]{e[\theta]}{\p{\tau} \imp \n{\sigma}}
          & \sof[k]{v[\theta]}{\p{\tau}}}}
    \end{mathpar}
    Note the the step index $k$ remains the same in all premises.
  \item[Case:]
    \begin{mathpar}
      \infer[\irule{\imp}]
      {\ctx \vdash \lam{x}{e} : \p{\tau} \imp \n{\sigma}}
      {\ctx , x{:}\p{\tau} \vdash e : \n{\sigma}}
      \and
      \infer=[]
      {\sof[k]{(\lam{x}{e})[\theta]}{\p{\tau} \imp \n{\sigma}}}
      {\infer[Lemma~\ref{lm:sd-lam}]
        {\sof[k]{\lam{x}{e[\theta,x/x]}}{\p{\tau} \imp \n{\sigma}}}
        {\infer=[]
          {x : \p{\tau} \models \sof[k]{e[\theta,x/x]}{\n{\sigma}}}
          {\sof[k]{e[\theta,v/x]}{\n{\sigma}}\; (\forall \sof[k]{v}{\p{\tau}})}}}
    \end{mathpar}
    Note that the step index $k$ remains the same.
  \end{description}

\end{proof}

\section{Soundness and Completeness of Bidirectional Typechecking}
\label{app:bidir}

Of note, $\lvert v'\rvert=v$ and $\lvert e'\rvert=e$ used in the theorems below
refer to the same value and/or expression, but with the possibility of extra
annotations that will also be erased from the term at runtime.

\begin{theorem}[Soundness of Bidirectional Typechecking]\label{thm:bidir-sound}
  \mbox{}
  \begin{enumerate}
    \item If $\Gamma \vdash v \checks \p{\tau}$ or $\Gamma \vdash v \synths \p{\tau}$
          then there exists an $v'$ such that $\Gamma \vdash v' : \p{\tau}$
          and $\lvert v'\rvert=v$
    \item If $\Gamma \vdash e \checks \n{\sigma}$ or $\Gamma \vdash e \synths \n{\sigma}$
          then there exists an $e'$ such that $\Gamma \vdash e' : \n{\sigma}$
          and $\lvert e'\rvert=e$
\end{enumerate}
\end{theorem}
\begin{proof}
  By straightforward induction on the structure of the typing derivation.
\end{proof}

We can also show that our bidirectional system is complete, as annotations can
always be added to make values and/or computations well-typed.

\begin{theorem}[Completeness of Bidirectional Typechecking]\label{thm:bidir-complete}
  \mbox{}
  \begin{enumerate}
    \item If $\Gamma \vdash v : \p{\tau}$ then there exists $v'$ and $v''$ s.t.
          $\Gamma \vdash v' \checks \p{\tau}$ and
          $\Gamma \vdash v'' \synths \p{\tau}$ where $\lvert v'\rvert=\lvert v''\rvert=v$
    \item If $\Gamma \vdash e : \n{\sigma}$ then there exists $e'$ and $e''$ s.t.
          $\Gamma \vdash e' \checks \n{\sigma}$ and
          $\Gamma \vdash e'' \synths \n{\sigma}$ where $\lvert e'\rvert=\vert e''\rvert=e$
  \end{enumerate}
\end{theorem}
\begin{proof}
  By straightforward induction on the structure of the typing
  derivation and using the rules $\p{\jrule{ANNO}}$
  and $\n{\jrule{ANNO}}$ where needed.
\end{proof}

\section{Interpretation of Isorecursive Types}
\label{app:isorecursive}

As we discussed in Section~\ref{sec:isorecursive}, we can also directly interpret
isorecurisve types\textemdash\!
types that are isomorphic, embodied by \emph{fold} and \emph{unfold} operators, but
not equal to their expansions\textemdash in order to obtain a formulation for
isorecursive semantic typing, and therefore semantic subtyping, within our
equirecursive setting. Previous work has studied the relation between these two
formulations from a syntactic perspective~\cite{Abadi96lics}, via type assignment with
positive recursive typing~\cite{Urzyczyn95mfcs}, and, more recently, in relation to
semantic expressiveness~\cite{Patrignani21popl}. For our needs, we demonstrate a
semantic translation from the iso- to equi-recursive settings that showcases no
significant differences between these formulations.

\subsubsection*{Syntax}
While our focus is on an isorecursive semantic interpretation, we need to
facilitate some additional syntax for values and computations
to establish our operational semantics, introducing the \emph{fold} constructor
and \emph{unfold} destructor (for computations only), typical of isorecursive
formulations.
\begin{align*}
  \p\tau &\Coloneqq \ldots \mid \p{\alpha} \mid \mu\p{\alpha}.\,\p{\tau} \\
  v &\Coloneqq \ldots \mid \foldmu{v} \mid \foldcase{x}{v}{e} \\
  \n\sigma &\Coloneqq \ldots \mid \n{\alpha} \mid \nu\n{\alpha}.\,\n{\sigma} \\
  e &\Coloneqq \ldots \mid \foldnu{e} \mid \unfold{e}
\end{align*}

\subsubsection*{Dynamics}

In the isorecursive interpretation, two reduction rules are added
for the judgment $e \reduces e'$. As before, values do not reduce.
\begin{mathpar}
  \infer{\foldcase{x}{(\foldmu{v})}{e} \reduces \subst{v/x}{e}}{}
  \and
  \infer{\unfoldfold{e} \reduces e}{}
\end{mathpar}

For this interpretation, we also expand our set of \emph{terminal computations}
to include one additional computation, following Lemma~\ref{lm:terminal}.
\begin{mathpar}
  \infer{(\foldnu{e}) \term}{\mathstrut}
\end{mathpar}

\subsubsection*{Semantic Typing}

We extend our semantic typing definitions from Section~\ref{sec:semantic-typing}
to incorporate the isorecursive $\mathsf{fold}$, introduced for values, and
$\mathsf{unfold}$, as the elimination form of $\mathsf{fold}$ for computations.
\begin{align*}
  \sof[k]{v}{\mu\p{\alpha}.\,\p{\tau}} &\defd \mbox{$v = \foldmu{v'}$ and
                                          $\sof[k]{v'}
                                          {\subst{\mu \p{\alpha}.\, \p{\tau} / \p{\alpha}}{\p{\tau}}}$} \\
  \sof*[k+1]{e}{\nu\n{\alpha}.\,\n{\sigma}} &\defd \sof[k+1]{\unfold{e}}
                                            {\subst{\nu \n{\alpha}.\, \n{\sigma} / \n{\alpha}}{\n{\sigma}}}
\end{align*}

With the addition of this setting, we can model recursive types with explicit
constructors $\mu \p{\alpha}.\, \p{\tau}$ and $\nu \n{\alpha}.\, \n{\sigma}$.
As before, we observe computations steps according to our reduction rules.

\subsection{Recursive Types as Type Definitions}
\label{app:mu-nu-mapping}

We define a generalized translation for mapping
all recursive variables, $\alpha$, and all arbitrary recursive types,
$\mu\alpha.\ \tau$ (with the possibility of nested recursion),
to \emph{fresh} equirecursive type names $\p{t}$ and $\n{s}$.
We encode this translation and mapping over a series of steps:

\begin{enumerate}
  \item Define a translation function $\llbracket \cdot \rrbracket$,
        distinguishing recursive types and type
        variables from all other types. Each recursive type variable,
        positive and negative, is translated into a fresh type name:
        \begin{align*}
          &\qquad\llbracket\p{\alpha}\rrbracket = \p{t}\enspace\mbox{for $\p{t}$ fresh and for all $\p{\alpha}$}\\
          &\qquad\llbracket\n{\alpha}\rrbracket = \n{s}\enspace\mbox{for $\n{s}$ fresh and for all $\n{\alpha}$}
        \end{align*}
		The translation maps each $\mu$- or $\nu$-type to the corresponding type name
		and is the identity function for all other types.
		\begin{equation*}
          \llbracket\mu \p{\alpha}.\, \p{\tau}\rrbracket
          = \p{t}
          \qquad\qquad
          \llbracket\nu \n{\alpha}.\, \n{\sigma}\rrbracket
          = \n{s}
        \end{equation*}

  \item Define type names through \emph{possibly non-contractive} type
        definitions:
        \begin{enumerate}
        	\item each $\mu$-type $\mu \p{\alpha}.\, \p{\tau}$ establishes
        	the definition for the corresponding type name $\p{t}$ through the
        	definition $\p{t} = \llbracket\p{\tau}\rrbracket$;
        	\item each $\nu$-type $\nu \n{\alpha}.\, \n{\sigma}$ prescribes
        	the definition for $\n{s}$ through the
        	definition $\n{s} = \llbracket\n{\sigma}\rrbracket$.
        \end{enumerate}

  \item Collect the new type definitions in a global type signature:
        \begin{align*}
          &\Sigma = [\ldots, \p{t} = \llbracket\p{\tau}\rrbracket, \ldots,
            \n{s} = \llbracket\n{\sigma}\rrbracket, \ldots]
        \end{align*}
\end{enumerate}

\subsection{Iso- to Equi-recursive Translation}
\label{sec:iso-to-equi}

In the spirit of Ligatti et al.'s~\cite{Ligatti17toplas}'s conjecture that equirecursive subtypes
could automatically be translated into isorecursive subtypes by inserting any
``missing $\mu's$'' (in a call-by-value language), our
translation works in the other direction, from an iso- to
equi-recursive setting, by inserting unary variant records for
$\mu$ positive types and unary lazy records for
$\nu$ negative types with corresponding \emph{fold} labels.

We start with a set of translations, $\llbracket \cdot \rrbracket$, involving
$\mathsf{fold_{\mu}}$ and $\mathsf{fold_{\nu}}$ introduction and elimination
forms for values and computations:

\begin{definition}[Iso- to equi-recursive translations for values and expressions]
  \label{def:iso-equi}
  	\begin{align*}
      \llbracket \foldmu{v} \rrbracket
      &= \inj{\mathsf{fold_{\mu}}}{\llbracket v\rrbracket} \\
      \llbracket \foldcase{x}{v}{e} \rrbracket
      &= \case[\ell \in L]{\llbracket v\rrbracket}
        {{\mathsf{fold_{\mu}}}{x} => \llbracket e_{\ell}\rrbracket}\\\smallskip
      \llbracket \foldnu{e} \rrbracket
      &= \{\mathsf{fold_{\nu}} = \llbracket e\rrbracket \}\\
      \llbracket \unfold{e} \rrbracket
      &= \proj{\foldnu}{\llbracket e\rrbracket}
    \end{align*}

  This is extended compositionally to all other constructs.
\end{definition}

Going further, for this specific interpretation, we extend the generalized
translation function $\llbracket \cdot \rrbracket$ in Section~\ref{app:mu-nu-mapping}
with an additional transformation: for every positive type definition
encountered, a unary variant record is inserted into the global
signature $\Sigma_{i2e}$. Similarly, for every negative type definition, a unary lazy
record is inserted into the signature as well:
\begin{align*}
  &\Sigma_{i2e} = [\ldots, \p{t} = \plus*[]{\foldmu \colon \llbracket\p{\tau}\rrbracket}, \ldots,
  \n{s} = \with*[]{\foldnu \colon \llbracket\n{\sigma}\rrbracket}, \ldots]
\end{align*}

\begin{definition}[Iso- to equi-recursive translation for isorecursive types]
  \label{def:iso-trans}
  The translation of a (positive or negative) isorecursive type $\tau$, is the
  equirecursive type $\llbracket\tau\rrbracket$ defined over the extended global signature
  $\Sigma_{i2e}$.
\end{definition}

Now, we can show translations for some examples.

\begin{example}[Translation for recursive positive types]
  \label{ex:iso-even}
  \begin{equation*}
  	\llbracket
    {\mu\p{\ms{\alpha}}.\, \plusn{\mb{z} : \one, \mb{s} : \p{\mb{\alpha}}}}
    \rrbracket
    =
    \p{\ms{nat}}
    \qquad
    \llbracket
    \mu\p{\ms{\beta}}.\, \plusn{\mb{z} : \one, \mb{s} : \plusn{\mb{s} : \p{\ms{\beta}}}}
    \rrbracket
    =
    \p{\ms{even}}
  \end{equation*}
  \begin{align*}
    \text{where }
    \Sigma_{i2e} = \{
    \p{\ms{nat}} &= \plusn{\mb{fold_{\mu}} : \plusn{\mb{z}: \one, \mb{s} : \ms{nat}}} ,\\
    \p{\ms{even}} &= \plusn{\mb{fold_{\mu}} : \plusn{\mb{z} : \one, \mb{s} :\plusn{\mb{s} : \ms{even}}}}
    \}
  \end{align*}
  This example demonstrates an interesting subtyping relation where
  $\ms{even} \nleq \ms{nat}$ in the isorecursive setting, yet $\ms{even} \leq \ms{nat}$ in
  the equirecursive one.  It also shows how
  our equirecursive formulation supports richer subtyping properties because of variant
  records\textemdash it is difficult to see a simple compositional translation into binary
  sums that would preserve subtyping.
\end{example}

\begin{example}[Translation for a recursive negative type]
  \begin{equation*}
    \llbracket
    {\nu\n{\ms{\alpha}}.\, \withn{\mb{hd} : \up \ms{std}, \mb{tl} : \n{\ms{\alpha}}}}
    \rrbracket
    = \n{\ms{stream}}
  \end{equation*}
  \begin{align*}
    \text{where }
     \n{\ms{stream}} = \withn{\mb{fold_{\nu}} : \withn{\mb{hd} : \up \ms{std}, \mb{tl} : \ms{stream}}}
  \end{align*}
\end{example}

Next, we show that isorecursive values and expressions are well-typed
semantically if and only if they are well-typed in our equirecursive
semantic interpretation with translated types and terms.

\begin{theorem}[Semantic Type Simulation]
  \mbox{}
  \begin{enumerate}
    \item
          $\sof[k]{v}{t}$ iff \ $ \sof[k]{\llbracket v \rrbracket}{\p{\tau}}$
          for $t=\p{\tau} \in \Sigma_{i2e}$
    \item
          $\sof[k]{e}{s}$ iff \ $\sof[k]{\llbracket e \rrbracket}{\n{\sigma}}$
          for $s=\n{\sigma} \in \Sigma_{i2e}$
  \end{enumerate}
\end{theorem}
\begin{proof}
  Directly, using our translation definitions in~\ref{def:iso-equi}
  and~\ref{def:iso-trans}; the rest follows from Section~\ref{sec:semantic-typing}.
\end{proof}

\subsubsection*{A Note on Contractiveness}
Given our translation in definition~\ref{def:iso-trans}, we can have the following
two translations, for example:
\begin{align*}
  \llbracket
  \mu\p{\ms{\alpha}}.\,\p{\ms{\alpha}}
  \rrbracket
  &= \p{\ms{t}}
  \qquad \text{where }
  \p{\ms{t}}=\plusn{\mb{fold_{\mu}} : \p{\ms{t}}}
  \\
  \llbracket
  \nu\n{\ms{\alpha}}.\,\n{\ms{\alpha}}
  \rrbracket
  &= \n{\ms{s}}
  \qquad \text{where }
  \n{\ms{s}}=\withn{\mb{fold_{\nu}} : \n{\ms{s}}}
\end{align*}

While these isorecursive types on the left may seem to break our contractive
restriction on first glance, $\mu\p{\ms{\alpha}}.\,\p{\ms{\alpha}}$ and
$\nu\n{\ms{\alpha}}.\,\n{\ms{\alpha}}$
are not restricted in the isorecursive setting~\cite{Zhou20oopsla}. Instead,
these translations demonstrate that our restriction on being contractive in the
equirecursive formulation is preserved for any given isorecursive type by
the insertion of unary variants.

\section{Call-by-Name}
\label{app:cbn}

\begin{align*}
  \tau, \sigma &\Coloneqq \tau \imp \sigma \mid \tau_1 \tensor \tau_2 \mid \one \mid \plus*[\ell \in L]{\ell\colon \tau_\ell} \mid \with*[\ell \in L]{\ell\colon \tau_\ell} \\
  e &\Coloneqq x \begin{array}[t]{@{{} \mid {}}l@{}}
                   \lam{x}{e} \mid \app{e_1}{e_2} \\
                   \pair{e_1,e_2} \mid \letpair{x,y}{e_1}{e_2} \\
                   \pair{} \mid \letpair{}{e_1}{e_2} \\
                   \inj{j}{e} \mid \case[\ell \in L]{e_1}{{\ell}{x_\ell} => e_\ell} \\
                   \record[\ell \in L]{\ell = e_\ell} \mid \proj{j}{e}
                 \end{array}
\end{align*}

Syntactic typing for this call-by-name language is captured by the judgment $\ctx \vdash e : \tau$, and its rules are standard.
For function types, variant record types, and the corresponding terms, these rules are:
\begin{mathpar}
  \infer{\ctx \vdash x : \tau}
  {x{:}\tau \in \ctx}
  \and
  \infer{\ctx \vdash \lam{x}{e} : \tau \imp \sigma}
  {\ctx , x{:}\tau \vdash e : \sigma}
  \and
  \infer{\ctx \vdash \app{e_1}{e_2} : \sigma}
  {\ctx \vdash e_1 : \tau \imp \sigma & \ctx \vdash e_2 : \tau}
  \\
  \infer{\ctx \vdash \inj{j}{e} : \plus*[\ell \in L]{\ell\colon \tau_\ell}}
  {\text{($j \in L$)} & \ctx \vdash e : \tau_j}
  \and
  \infer{\ctx \vdash \case[\ell \in L]{e_1}{{\ell}{x_\ell} => e_\ell} : \sigma}
  {\ctx \vdash e_1 : \plus*[\ell \in L]{\ell\colon \tau_\ell} &
   \metaall{\ell \in L} \ctx , x_{\ell}{:}\tau_\ell \vdash e_\ell : \sigma}
\end{mathpar}

The operational semantics is what distinguishes the call-by-name language from a call-by-value language.
We give a small-step operational semantics using the judgments $e \reduces e'$ and $e \term$.
The rules involving functions and variant records are:
\begin{mathpar}
  \infer{\lam{x}{e} \term}{}
  \and
  \infer{\app{e_1}{e_2} \reduces \app{e'_1}{e_2}}
  {e_1 \reduces e'_1}
  \and
  \infer{\app{(\lam{x}{e'_1})}{e_2} \reduces \subst{e_2/x}{e'_1}}{}
  \\
  \infer{\inj{j}{e} \term}{}
  \and
  \infer{\case[\ell \in L]{e_1}{{\ell}{x_\ell} => e_\ell} \reduces \case[\ell \in L]{e'_1}{{\ell}{x_\ell} => e_\ell}}
  {e_1 \reduces e'_1}
  \and
  \infer{\case[\ell \in L]{(\inj{j}{e'_1})}{{\ell}{x_\ell} => e_\ell} \reduces \subst{e'_1/x_j}{e_j}}
  {\text{($j \in L$)}}
\end{mathpar}
This semantics is call-by-name because, for example, in a function application $\app{e_1}{e_2}$, the argument $e_2$ remains unevaluated when substituted into the body of an abstraction $\lam{x}{e'_1}$.
Similarly, in $\case[\ell \in L]{e_1}{{\ell}{x_\ell} => e_\ell}$, the term $e_1$ is evaluated to the form $\inj{j}{e'_1}$ for some $j \in L$.
Because $\inj{j}{e'_1}$ is terminal, $e'_1$ is not further evaluated; instead, $e'_1$ is substituted into $e_j$, the body of the $j$\textsuperscript{th} branch.

Because of the way that we handle equirecursive types and expressions, we translate signatures.
\begin{equation*}
  \begin{gathered}[t]
    \text{\emph{Signatures}} \\
    \begin{aligned}
      \cbn*{\sige} &= \sige* \\
      \cbn*{\sig , s = \sigma} &= \cbn{\sig} , \cbn{s} = \cbn{\sigma} \\
      \cbn*{\sig , f : \sigma = e} &= \cbn{\sig} , f : \cbn{\sigma} = \cbn{e}
    \end{aligned}
  \end{gathered}
\end{equation*}

Levy~\cite{Levy06hosc} proves that well-typed terms are well-typed after the translation to call-by-push-value is applied.
Our syntactic typing rules are the same as his, so the theorem carries over to our setting.
\begin{theorem}[\cite{Levy06hosc}]
 $\ctx \vdash e : \tau$ if and only if $\dn \cbn{\ctx} \vdash \cbn{e} : \cbn{\tau}$.
\end{theorem}

Now, we prove that polarized subtyping on the image of Levy's call-by-name translation is sound.
We begin with an easy lemma.
\begin{lemma}\label{lem:dn-invert}
  If $\p{t} \leq \p{u}$ with $\p{t} = \dn \n{s}$, then $\p{u} = \dn \n{r}$ and
  $\n{s} \leq \n{r}$ for some $\n{r}$.
\end{lemma}
\begin{proof}
  By a straightforward examination of the call-by-push-value syntactic subtyping rules,
  observing that $\p{t} \emp$ is not derivable when $\p{t} = \dn \n{s}$.
\end{proof}

The soundness theorem is then proved as follows.
\begin{proof}
  (of Theorem~\ref{thm:cbn-soundness})
  Part~\ref{item:cbn-full-sound} is easy to prove directly.
  By inversion on the body of $t$'s definition, there are several cases:
  \begin{itemize}
    \item If $t = t_1 \imp t_2$, then $\cbn{t} = \p{t_0} \imp \cbn{t_2}$, where $\p{t_0} = \dn \cbn{t_1}$
    is an auxiliary definition introduced for the normal form of type definitions.
    By inversion on $\cbn{t} \full$, we must have $\p{t_0} \emp$.
    Because $\p{t_0} = \dn \cbn{t_1}$, this case is contradictory: there is no emptiness rule for the $\dn$ shift.

  \item If $t = \with*[\ell \in L]{\ell\colon \tau_\ell}$, then $\cbn{t} = \with*[\ell \in L]{\ell\colon \cbn{\tau_\ell}}$.
    By inversion on $\cbn{t} \full$, we must have $L = \emptyset$.
    In this case, we indeed have $t \full$ by the call-by-name $\rfull{\with}$ rule.

  \item In all other cases, $\cbn{t} = \up \p{t_0}$ with $\p{t_0}$ introduced for the normal form of type definitions.
    However, there is no call-by-push-value fullness rule for $\up$ shifts, so these cases are contradictory as well.
  \end{itemize}

  Part~\ref{item:cbn-leq-sound} is proved by mapping a circular proof $\cbn{t} \leq \cbn{u}$ to a circular proof of $t \leq u$.
  We can also prove:
  \begin{equation}
    \label{eq:cbn-leq-pos-sound}\mbox{If $\cbn{t} = \up \p{t_0}$ and $\cbn{u} = \up \p{u_0}$ with $\p{t_0} \leq \p{u_0}$, then $t \leq u$}
  \end{equation}
  This is done by mapping a circular proof $\p{t_0} \leq \p{u_0}$ to a circular proof of $t \leq u$.
  This and Part~\ref{item:cbn-leq-sound} are proved simultaneously.
  \begin{itemize}
  \item Consider the case in which $\cbn{t} \leq \cbn{u}$ is derived by the $\rsub{\top}$ rule.
    \begin{equation*}
      \infer[\rsub{\top}]{\cbn{t} \leq \cbn{u}}
      {\cbn{t} = \n{\sigma} & \cbn{u} \full}
    \end{equation*}
    By part~\ref{item:cbn-full-sound}, $u \full$ in the call-by-name language.
    It follows from the $\rsub{\top}_{\jrule{N}}$ rule that $\cbn{t} \leq \cbn{u}$.

  \item
    Consider the case in which $\cbn{t} \leq \cbn{u}$ is derived by the $\rsub{\imp}$ rule.
    By inversion on $\cbn{t}$ and $\cbn{u}$, this can only happen if $t = t_1 \imp t_2$ and $u= u_1 \imp u_2$, with $\cbn{t} = \p{t_0} \imp \cbn{t_2}$ and $\cbn{u} = \p{u_0} \imp \cbn{u_2}$, where $\p{t_0} = \dn \cbn{t_1}$ and $\p{u_0} = \dn \cbn{u_1}$ are auxiliary definitions introduced for the normal form of type definitions.
    \begin{equation*}
      \infer[\rsub{\imp}]{\cbn{t} \leq \cbn{u}}
      {\cbn{t} = \p{t_0} \imp \cbn{t_2} &
       \cbn{u} = \p{u_0} \imp \cbn{u_2} &
       \p{u_0} \leq \p{t_0} &
       \cbn{t_2} \leq \cbn{u_2}}
    \end{equation*}
    By Lemma~\ref{lem:dn-invert}, $\cbn{u_1} \leq \cbn{t_1}$.
    By transforming according to part~\ref{item:cbn-leq-sound}, we have both $u_1 \leq t_1$ and $t_2 \leq u_2$.
    From these we can derive $t \leq u$ with the $\rsub{\imp}_{\jrule{N}}$ rule.

  \item
    Consider the case in which $\cbn{t} \leq \cbn{u}$ is derived by the $\rsub{\up}$ rule.
    \begin{equation*}
      \infer[\rsub{\up}]{\cbn{t} \leq \cbn{u}}
      {\cbn{t} = \up \p{t_0} &
       \cbn{u} = \up \p{u_0} &
       \p{u_0} \leq \p{t_0}}
    \end{equation*}
    By item~\ref{eq:cbn-leq-pos-sound}, $t \leq u$.

  \item
    Consider the case in which $\cbn{t} = \up \p{t_0}$ and $\cbn{u} = \up \p{u_0}$ with $\p{t_0} \leq \p{u_0}$ being derived by the $\rsub{\plus}$ rule.
    In this case, $t = \plus*[\ell \in L]{\ell\colon t_\ell}$ and $u = \plus*[j \in J]{j\colon u_j}$, with $\p{t_0} = \plus*[\ell \in L]{\ell\colon \p{t_\ell}}$ and $\p{t_\ell} = \dn \cbn{t_\ell}$ and $\p{u_0} = \plus*[j \in J]{j\colon \p{u_j}}$ and $\p{u_j} = \dn \cbn{u_j}$ are auxiliary definitions introduced for the normal of type definitions.
    \begin{mathpar}
      \infer[\rsub{\plus}]
      {\p{t_0} \leq \p{u_0}}
      {\p{u_0} = \plus*[j \in J]{j\colon \p{u_j}} &
       \begin{array}[b]{@{}c@{}}
         \p{t_0} = \plus*[\ell \in L]{\ell\colon \p{t_\ell}} \\
         \metaall{\ell \in L \setminus J} \p{t_\ell} \emp
       \end{array} &
       \metaall{\ell \in L \cap J} \p{t_\ell} \leq \p{u_\ell}}
    \end{mathpar}
    Observe that $\p{t_\ell} \emp$ is not derivable for any $\ell \in L \setminus J$ because $\p{t_\ell} = \dn \cbn{t_\ell}$.
    Therefore, $L \subseteq J$ must hold.
    By Lemma~\ref{lem:dn-invert}, $\cbn{t_\ell} \leq \cbn{u_\ell}$ for all $\ell \in L \cap J = L$.
    By transforming according to part~\ref{item:cbn-leq-complete}, we have $t_\ell \leq u_\ell$ for all $\ell \in L$.
    From these we can derive $t \leq u$ with the $\rsub{\plus}_{\jrule{N}}$ rule.

  \item
    Consider the case in which $\cbn{t} = \up \p{t_0}$ and $\cbn{u} = \up \p{u_0}$ with $\p{t_0} \leq \p{u_0}$ being derived by the $\p{\rsub{\bot}}$ rule.
    In this case, $\p{t_0} \emp$.
    There are three subcases.
    \begin{itemize}
    \item If $t = \plus*[\ell \in L]{\ell\colon t_\ell}$, then $\p{t_0} = \plus*[\ell \in L]{\ell\colon \p{t_\ell}}$, with $\p{t_\ell} = \dn \cbn{t_\ell}$ being auxiliary definitions introduced for the normal form of type definitions.
      None of $\p{t_\ell} \emp$ are derivable because there is no call-by-push-value emptiness rule for the $\dn$ shift.
      Therefore, $\p{t_0} \emp$ is derivable only if $L = \emptyset$.
      In this case, the call-by-name $\rsub{\bot}_{\jrule{N}}$ rule derives $t \leq u$.

    \item The subcase in which $t = t_1 \tensor t_2$ is similarly impossible.

    \item If $t = \one$, then $\p{t_0} = \one$.
      The judgment $\p{t_0} \emp$ is not derivable, as there is no call-by-push-value emptiness rule for $\one$.
    \end{itemize}
  \end{itemize}
  The remaining cases are handled similarly.
\end{proof}

Next, we prove that polarized subtyping on the image of Levy's call-by-name translation is complete.
\begin{proof} (of Theorem~\ref{thm:cbn-completeness})
  Part~\ref{item:cbn-full-complete} is easy to prove directly.
  There is exactly one case: $t \full$ because $t = \with*{\,}$.
  In this case, $\cbn{t} = \with*{\,}$ and so $\cbn{t} \full$.

  Part~\ref{item:cbn-leq-complete} is proved by mapping a circular proof $t \leq u$ to a circular proof of $\cbn{t} \leq \cbn{u}$.
  The image of each call-by-name subtyping rule is derivable with the call-by-push-value syntactic subtyping rules.
  For example, consider the following call-by-name subtyping rule for function types.
  \begin{equation*}
    \infer[\rsub{\imp}_{\jrule{N}}]{t \leq u}
    {t = t_1 \imp t_2 & u = u_1 \imp u_2 &
     u_1 \leq t_1 & t_2 \leq u_2}
  \end{equation*}
  The translation of $t$ is $\cbn{t} = \p{t_0} \imp \cbn{t_2}$, where $\p{t_0} = \dn \cbn{t_1}$ is an auxiliary definition introduced by internal renaming; the translation of $u$ is analogous.
  The call-by-value subtyping rule $\rsub{\imp}_{\jrule{N}}$ is then derivable as:
  \begin{equation*}
    \infer[\!\!\rsub{\imp}]{\cbn{t} \leq \cbn{u}}
    {\cbn{t} = \p{t_0} \!\!\imp \cbn{t_2} & \!\!\!\cbn{u} = \p{u_0} \!\!\imp \cbn{u_2} &
      \infer[\!\!\!\rsub{\dn}]{\p{u_0} \leq \p{t_0}}
      {\p{u_0} = \dn \cbn{u_1} \!\!\!& \p{t_0} = \dn \cbn{t_1} \!\!\!&
       \cbn{u_1} \leq \cbn{t_1}} &
       \!\!\!\!\!\!\cbn{t_2} \leq \cbn{u_2}}
  \end{equation*}

  As another example, consider the $\rsub{\bot}_{\jrule{N}}$ call-by-name subtyping rule.
  \begin{equation*}
    \infer[\rsub{\bot}_{\jrule{N}}]{t \leq u}
    {t = \plus*{\,} & u = \sigma}
  \end{equation*}
  The translation of $t$ is $\cbn{t} = \up \p{t_0}$, where $\p{t_0} = \plus*{\,}$ is an auxiliary definition introduced by internal renaming.
  The call-by-name subtyping rule $\rsub{\bot}_{\jrule{N}}$ is then derivable as:
  \begin{equation*}
    \infer[\n{\rsub{\bot}}]{\cbn{t} \leq \cbn{u}}
    {\cbn{t} = \up \p{t_0} &
      \infer[\remp{\plus}]{\p{t_0} \emp}
      {\p{t_0} = \plus*{\,}} &
       \cbn{u} = \n{\sigma}}
  \end{equation*}

  The other cases are handled similarly.
\end{proof}

\section{Call-by-Value}
\label{app:cbv}

\begin{align*}
  \tau, \sigma &\Coloneqq \tau \imp \sigma \mid \tau_1 \tensor \tau_2 \mid \one \mid \plus*[\ell \in L]{\ell\colon \tau_\ell} \mid \with*[\ell \in L]{\ell\colon \sigma_\ell} \\
  e &\Coloneqq x \begin{array}[t]{@{{} \mid {}}l@{}}
                   \lam{x}{e} \mid \app{e_1}{e_2} \\
                   \pair{e_1,e_2} \mid \letpair{x,y}{e_1}{e_2} \\
                   \pair{} \mid \letpair{}{e_1}{e_2} \\
                   \inj{j}{e} \mid \case[\ell \in L]{e_1}{{\ell}{x_\ell} => e_\ell} \\
                   \record[\ell \in L]{\ell = e_\ell} \mid \proj{j}{e}
                 \end{array}
\end{align*}

We again use a small-step operational semantics that relies on the judgments $e \reduces e'$ and $e \val$.
The rules involving functions are:
\begin{mathpar}
  \infer{\lam{x}{e} \val}{}
  \and
  \infer{\app{e_1}{e_2} \reduces \app{e_1}{e'_2}}
  {e_2 \reduces e'_2}
  \and
  \infer{\app{e_1}{e_2} \reduces \app{e'_1}{e_2}}
  {e_2 \val & e_1 \reduces e'_1}
  \and
  \infer{\app{(\lam{x}{e_1})}{e_2} \reduces \subst{e_2/x}{e_1}}
  {e_2 \val}
\end{mathpar}
This semantics is call-by-value because, for example, in a function application $\app{e_1}{e_2}$, the argument $e_2$ is evaluated to a value first: only values are ever substituted into the body of an abstraction $\lam{x}{e_1}$.

\begin{theorem}[\cite{Levy06hosc}]
 $\ctx \vdash e : \tau$ if and only if $\cbv{\ctx} \vdash \cbv{e} : \up \cbv{\tau}$.
\end{theorem}

We now prove that polarized subtyping on the image of Levy's call-by-value translation is sound with respect to Figure~\ref{fig:call-by-value}.
We begin with an easy lemma.

\begin{lemma}\label{lem:up-invert}
  If $\n{s} \leq \n{r}$ with $\n{s} = \up \p{t}$, then either $\p{t} \emp$ or $\n{r} = \up \p{t}$ and $\p{s} \leq \p{r}$ for some $\p{r}$.
\end{lemma}
\begin{proof}
  By a straightforward examination of the call-by-push-value subtyping rules.
\end{proof}

Now we prove the main soundness theorem.

\begin{proof} (of Theorem~\ref{thm:cbv-soundness})
  Part~\ref{item:cbv-empty-sound} is easy to prove directly.
  By inversion on the body of $t$'s definition, there are several cases:
  \begin{itemize}
  \item If $t = t_1 \imp t_2$, then $\cbv{t} = \dn \n{t_0}$, where $\n{t_0} = \cbv{t_1} \imp \n{t_3}$ and $\n{t_3} = \up \cbv{t_2}$ are auxiliary definitions introduced for the normal form of type definitions.
    Because $\cbv{t} = \dn \n{t_0}$, this case is contradictory: there is no emptiness rule for the $\dn$ shift.
  \item The case for $t = \with*[\ell \in L]{\ell\colon t_\ell}$ is similar.

  \item In all other cases, $\cbv{t}$ proceeds homomorphically and the call-by-push-value emptiness rules have corresponding call-by-value rules.
  \end{itemize}

  Part~\ref{item:cbv-leq-sound} is proved by mapping a circular proof of $\cbv{t} \leq \cbv{u}$ to a circular proof of $t \leq u$. We can also similarly show:
  \begin{equation}
    \label{eq:cbv-leq-neg-sound}\mbox{If $\cbv{t} = \dn \n{t_0}$ and $\cbv{u} = \dn \n{u_0}$ with $\n{t_0} \leq \n{u_0}$, then $t \leq u$}.
  \end{equation}
  \begin{itemize}
  \item Consider the case in which $\cbv{t} \leq \cbv{u}$ is derived by the $\p{\rsub{\bot}}$ rule.
    \begin{equation*}
      \infer[\rsub{\top}]{\cbv{t} \leq \cbv{u}}
      {\cbv{t} \emp & \cbv{u} = \n{\sigma}}
    \end{equation*}
    By part~\ref{item:cbv-empty-sound}, $t \emp$ in the call-by-value language.
    It follows from the $\rsub{\bot}_{\jrule{V}}$ rule that $t \leq u$.

  \item
    Consider the case in which $\cbv{t} = \dn \n{t_0}$ and $\cbv{u} = \dn \n{u_0}$, with $\n{t_0} \leq \n{u_0}$ being derived by the $\rsub{\imp}$ rule.
    By inversion on $\cbv{t}$ and $\cbv{u}$, this can only happen if $t = t_1 \imp t_2$ and $u= u_1 \imp u_2$, where $\n{t_0} = \cbv{t_1} \imp \n{t_3}$ and $\n{t_3} = \up \cbv{t_2}$ and $\n{u_0} = \cbv{u_1} \imp \n{u_3}$ and $\n{u_3} = \up \cbv{u_2}$ are auxiliary definitions introduced for the normal form of type definitions.
    \begin{equation*}
      \infer[\rsub{\imp}]{\n{t_0} \leq \n{u_0}}
      {\n{t_0} = \cbv{t_1} \imp \n{t_3} &
       \n{u_0} = \cbv{u_1} \imp \n{u_3} &
       \cbv{u_1} \leq \cbv{t_1} &
       \n{t_3} \leq \n{u_3}}
    \end{equation*}
    By Lemma~\ref{lem:up-invert}, either $\cbv{t_2} \emp$ or $\cbv{t_2} \leq \cbv{u_2}$.
    In the former case, $t_2 \emp$, and we have $t_2 \leq u_2$ by the $\rsub{\bot}_{\jrule{V}}$ rule.
    In the latter case, by transforming according to part~\ref{item:cbv-leq-sound}, we have $t_2 \leq u_2$.
    In both cases, by transforming according to part~\ref{item:cbv-leq-sound}, we have $u_1 \leq t_1$.
    From these we can derive $t \leq u$ with the $\rsub{\imp}_{\jrule{V}}$ rule.

  \item
    Consider the case in which $\cbv{t} = \dn \n{t_0}$ and $\cbv{u} = \dn \n{u_0}$, with $\n{t_0} \leq \n{u_0}$ being derived by the $\rsub{\top}$ rule.
    By inversion on $\cbv{t}$ and $\cbv{u}$, this can only happen in cases where $t$ and $u$ are either function types or lazy record types.
    If both $t$ and $u$ are lazy record types, then $t \leq u$ is derivable by $\rsub{\with}_{\jrule{V}}$.
Otherwise, $t \leq u$ is derivable by one of the $\rsub{\top}^{\imp\imp}_{\jrule{V}}$, $\rsub{\top}^{\with\imp}_{\jrule{V}}$, or $\rsub{\top}^{\imp\with}_{\jrule{V}}$ rules.
  \item
    Consider the case in which $\cbv{t} \leq \cbv{u}$ is derived by the $\rsub{\dn}$ rule.
    \begin{equation*}
      \infer[\rsub{\dn}]{\cbv{t} \leq \cbv{u}}
      {\cbv{t} = \dn \n{t_0} &
       \cbv{u} = \dn \n{u_0} &
       \n{t_0} \leq \n{u_0}}
   \end{equation*}

    By item~\ref{eq:cbv-leq-neg-sound}, $t \leq u$.
  \end{itemize}
  The remaining cases are handled similarly.
\end{proof}

Next, we prove that polarized subtyping on the image of Levy's call-by-value translation is complete with respect to Figure~\ref{fig:call-by-value}.

\begin{proof} (of Theorem~\ref{thm:cbv-completeness})
  Part~\ref{item:cbv-empty} is proved by mapping a circular proof of $t \emp$ to a circular proof of $\cbv{t} \emp$.
  As an example, consider the case in which $t \emp$ is derived by the $\remp{\plus}_{\jrule{V}}$ rule:
  \begin{equation*}
    \infer[\remp{\plus}_{\jrule{V}}]{t \emp}
    {t = \plus*[\ell \in L]{\ell\colon t_\ell} &
    \metaall{\ell \in L} t_\ell \emp}
  \end{equation*}
  Transforming each circular proof of $t_\ell \emp$ according to part~\ref{item:cbv-empty}, we have $\cbv{t} \emp$ for each $\ell \in L$.
  Because  $\cbv{t} = \plus*[\ell \in L]{\ell\colon \cbv{t_\ell}}$, we can derive $\cbv{t} \emp$ with the $\remp{\plus}$ rule.
  The other cases are similar.

  Part~\ref{item:cbv-leq} is proved by mapping a circular proof $t \leq u$ to a circular proof of $\cbv{t} \leq \cbv{u}$.
  The image of each call-by-value subtyping rule is derivable with the call-by-push-value syntactic subtyping rules.
  For example, consider the following call-by-value subtyping rule for function types.
  \begin{equation*}
    \infer[\rsub{\imp}_{\jrule{V}}]{t \leq u}
    {t = t_1 \imp t_2 & u = u_1 \imp u_2 &
     u_1 \leq t_1 & t_2 \leq u_2}
  \end{equation*}
  The translation of $t$ is $\cbv{t} = \dn \n{t_0}$ where $\n{t_0} = \cbv{t_1} \imp \n{t_3}$ and $\n{t_3} = \up \cbv{t_2}$ are auxiliary definitions introduced for the normal form of type definitions; the translation of $u$ is analogous.
  The call-by-value subtyping rule $\rsub{\imp}_{\jrule{V}}$ is then derivable as:
  \begin{mathpar}
    \infer[\rsub{\up}]{\n{t_3} \leq \n{u_3}}
    {\n{t_3} = \up \cbv{t_2} \and \n{u_3} = \up \cbv{u_2} \and
     \cbv{t_2} \leq \cbv{u_2}}
    \and
    \infer[\rsub{\imp}]{\n{t_0} \leq \n{u_0}}
    {\n{t_0} = \cbv{t_1} \imp \n{t_3} & \n{u_0} = \cbv{u_1} \imp \n{u_3} \and
     \cbv{u_1} \leq \cbv{t_1} \and \infer[]{\n{t_3} \leq \n{u_3}}{}
    }
    \and
    \infer[\rsub{\dn}]{\cbv{t} \leq \cbv{u}}
    {\cbv{t} = \dn \n{t_0} \and \cbv{u} = \dn \n{u_0} \and
      \infer[]{\n{t_0} \leq \n{u_0}}{}}
  \end{mathpar}

  The other cases are handled similarly.
\end{proof}

\end{document}